\documentclass[11pt,a4paper]{article} 

\usepackage{graphicx}
\usepackage{amsmath}
\usepackage{amsthm}
\usepackage{amsfonts}
\usepackage{amssymb}
\usepackage{verbatim}
\usepackage{tabularx}
\usepackage{enumerate}
\usepackage{cancel}

\usepackage{color}
\usepackage{fullpage}
\usepackage{float}

\usepackage{marginnote}
\usepackage{mathtools}

\usepackage{tikz}
\usepackage{wrapfig}
\usetikzlibrary{arrows,automata,positioning}   %,snakes} 
\tikzset{every state/.style={minimum size=30}}

\usepackage[breaklinks]{hyperref}

\DeclareMathOperator{\dsc}{\mathrm{sc}}
\DeclareMathOperator{\lcm}{\mathrm{lcm}}

\newcommand{\eps}{\varepsilon}
\newcommand{\emp}{\emptyset}
\newcommand{\cD}{\mathcal{D}}

\newtheorem{theorem}{Theorem}
\newtheorem{corollary}[theorem]{Corollary}
\newtheorem{lemma}[theorem]{Lemma}
\newtheorem{proposition}[theorem]{Proposition}
\newtheorem{claim}[theorem]{Claim}
\newtheorem{definition}[theorem]{Definition}
\newtheorem{example}[theorem]{Example}

\begin{document}

\title{State Complexity of Multiple Concatenation\thanks{A preliminary version of this paper appeared in the proceedings  of DCFS~2020 conference~\cite{jj20}.}} 

\author{Jozef Jir\'asek\footnote{Research supported by VEGA grant 1/0350/22 and grant APVV-24-0103.}\\
	Institute of Computer Science,
	P.\,J. \v Saf\'arik University \\
	Jesenn\'a 5, 040 01 Ko\v sice, 
	Slovakia\\
	\tt jozef.jirasek@upjs.sk
% 	\href{mailto:jozef.jirasek@upjs.sk}{\tt jozef.jirasek@upjs.sk} 
	\and Galina Jir\'askov\'a\footnote{Research supported by VEGA grant 2/0096/23 and grant APVV-24-0103.}\\
	Mathematical Institute,
	Slovak Academy of Sciences\\
	Gre\v{s}\'akova 6, 040 01 Ko\v{s}ice,
	Slovakia\\
	\tt jiraskov@saske.sk}

\date{\today}	

\maketitle

\begin{abstract}
We describe witness languages meeting the upper bound 
on the state complexity of the multiple concatenation of $k$ regular languages
over an alphabet of size~$k+1$  with a significantly simpler proof
than that in the literature. 
We also consider the case where some languages may be   
recognized by two-state automata.
Then we show that one symbol can be  saved,  
and we define witnesses for the multiple concatenation 
of $k$   languages over a~$k$-letter alphabet.
This solves an open problem stated by Caron et al. [2018, 	Fundam. Inform. 160, 255--279].
We prove that for the concatenation of three languages,   
the ternary alphabet is optimal.
We also show that a trivial upper bound
on the state complexity of multiple concatenation is 
asymptotically tight for ternary languages,
and that a lower bound remains exponential in the binary case. 
Finally, we  obtain a tight upper bound
for unary cyclic languages and languages
recognized by unary automata that do not have  final states in their tails.
\end{abstract}

\section{Introduction}
\label{s:in}
Given formal languages~$L_1,L_2,\ldots,L_k$
over an alphabet~$\Sigma$,
their concatenation is the 
language $L_1L_2\cdots L_k=\{u_1u_2\cdots u_k\mid
u_i\in L_i \text{ for }i=1,2,\ldots,k\}$.
Here we consider the case where all languages
are regular and ask the question of how many states
are sufficient and necessary in the worst case
for a deterministic finite automaton %(DFA)
to recognize their concatenation
assuming that each~$L_i$
is recognized by an~$n_i$-state 
deterministic finite automaton.

The first results for the concatenation
of two regular languages
were obtained by Maslov~\cite{ma70} in~1970.
In particular, he described binary witnesses meeting the upper bound $n_12^{n_2}-2^{n_2-1}$.
In~1994 Yu et al.~\cite{yzs94} proved that this upper bound
cannot be met if the first language
is recognized by a minimal deterministic
finite automaton that has more than one final state.

The concatenation of three and four regular languages
was considered by \'Esik et al.~\cite{esik09} in~2009,
where the witnesses for the concatenation of three languages
over a five-letter alphabet can be found.
The rather complicated
expression  for the  upper bounds for the concatenation of~$k$ languages,
as well as witnesses over a~$(2k-1)$-letter alphabet  were given by
Gao and Yu \cite{gy09}. 

Caron et al. \cite{clp18} presented recursive formulas for the  upper bounds,
and described witnesses over a $(k+1)$-letter alphabet using 
Brzozowski's universal automata. 
They also showed that to meet the upper bound for the concatenation of two or three languages,
the binary or ternary alphabet, respectively, is enough,
and they conjectured that $k$ symbols could be enough
to describe witnesses for the concatenation of~$k$ languages.  
 
In this paper, we study in detail the state complexity 
of   multiple   concatenation of~$k$ regular languages. 
We first describe witnesses over an alphabet consisting of~$k+1$ symbols
with a significantly simpler proof than that in~\cite{clp18}. 
Our witness automata~$A_1,A_2,\ldots,A_k$
are defined over the alphabet $\{b,a_1,\ldots,a_k\}$.
Each $a_i$
performs the circular shift in~$A_i$   and the identity in all the other automata.
These~$k$ permutation symbols are used to get the reachability of all so-called valid states
in a DFA for concatenation.
The symbol~$b$ performs a contraction in each~$A_i$
and assures the distinguishability of all valid states
almost for free. However,
the proof requires that each~$A_i$ 
has at least three states.
With a slightly more complicated proof, we also solve
the case that   includes two-state automata.
Then we describe special binary witnesses
for the concatenation of two languages.
We combine our ideas used for the~$(k+1)$-letter alphabet and those for binary  witnesses
to describe witnesses for multiple concatenation
over a $k$-letter alphabet,
which solves an open problem stated by Caron et al. \cite{clp18}.
In the case of~$k=3$, we show that the ternary alphabet is optimal.
 
We also examine multiple concatenation on binary, 
ternary, and unary  languages.
We show that in the binary case, 
the lower bounds remain exponential in $n_2,n_3,\ldots,n_k$,
and in the ternary case, 
the trivial upper bound~$n_12^{n_2+n_3+\cdots+n_k}$
can be met up to some multiplicative constant
depending on~$k$.
For unary languages, 
we use Frobenius numbers to
get a tight upper bound for cyclic languages,
or languages recognized by automata that do not have final states in their tails. We also consider
the case with final states in tails,
and provide  upper and lower bounds for
multiple concatenation in such a case.

\section{Preliminaries}
\label{sec:prelim}

We assume that the reader is familiar with basic notions
in automata and formal language theory.
For details and all unexplained notions, 
we refer the reader to \cite{si12}.  %hu79,
The size of a finite set $S$ is denoted by $|S|$,
and the set of all its subsets by $2^S$.

For a finite non-empty alphabet of symbols $\Sigma$,
the set of all strings over~$\Sigma$,
including the empty string~$\eps$,
is denoted by $\Sigma^*$.
A language is any subset of~$\Sigma^*$.
The multiple concatenation of $k$ languages $L_1,L_2,\ldots,L_k$
is~$L_1 L_2\cdots L_k = 
\{ u_1 u_2\cdots u_k  \mid u_1\in L_1, u_2\in L_2,\ldots,u_k\in L_k\}$.

A \emph{deterministic finite automaton} (DFA) 
is a quintuple $A=(Q,\Sigma,\cdot,s,F)$ where
$Q$ is a non-empty finite \emph{set of states},
$\Sigma$ is a non-empty finite \emph{alphabet  of input symbols},
$\cdot\,\colon Q\times \Sigma \to Q$ is the \emph{transition function},
$s\in Q$ is the initial state,
and $F\subseteq Q$ is the set of \emph{final}  (accepting) states.
The transition function can be naturally extended to the domain $Q\times\Sigma^*$.
The \emph{language  recognized (accepted)} by the DFA~$A$ 
is the set of strings $L(A)=\{w\in\Sigma^*\mid s\cdot w\in F\}$.

All deterministic finite automata
in this paper are assumed to be complete;
that is, the transition function is a total function.

We usually omit~$\cdot$, and   write~$qa$ instead of~$q\cdot a$.
Next, for a subset $S$ of $Q$ and a string~$w$, let~$Sw=\{qw\mid q\in S\}$
and~$wS=\{q\mid qw\in S\}$.
Each input symbol~$a$ induces a transformation 
on~$Q=\{q_1,q_2,\ldots,q_n\}$
given by~$q\mapsto qa$. 
We denote by~$a\colon (q_1,q_2,\ldots,q_\ell)$ the transformation
that maps~$q_i$ to~$q_{i+1}$ for~$i=1, \ldots,\ell-1$, the state $q_\ell$ to~$q_1$,
and fixes any other state in~$Q$.
In particular,~$(q_1)$ denotes the identity.
Next, we denote by~$a\colon (q_1 \to q_2 \to \cdots \to q_\ell)$  the transformation
that maps~$q_i$ to~$q_{i+1}$ for~$i=1,2,\ldots,\ell-1$
and fixes any other state. 
Finally, we denote by~$a\colon (S \to  q_i)$  the transformation
that maps~each~$q\in S$ to~$q_{i}$  
and fixes any other state.

A state~$q\in Q$  is \emph{reachable}
in the DFA~$A$
if there is a string~$w\in\Sigma^*$
such that~$q=sw$.
Two states~$p$ and~$q$ are \emph{distinguishable}
if there is a string~$w$ such that
exactly one of the states~$pw$ and~$qw$ is final.
A~state~$q\in Q$ is a \emph{dead  state}
if~$qw\notin F$ for every string~$w\in\Sigma^*$.

A DFA is \emph{minimal} 
if all its states are reachable and pairwise distinguishable.
The \emph{state complexity} of a regular language~$L$, $\dsc(L)$,
is the number of states in the minimal DFA recognizing $L$.
The state complexity of a $k$-ary regular operation $f$
is a function from $\mathbb{N}^k$ to $\mathbb{N}$ given 
by~$(n_1, n_2, \ldots,n_k)\mapsto \max\{ \dsc(f(L_1,L_2, \ldots,L_k)) 
\mid \dsc(L_i) \le n_i \text{ for } i=1,2, \ldots, k \}$.

A \emph{nondeterministic finite automaton} (NFA)
is a quintuple~$N=(Q,\Sigma,\cdot,I,F)$
where~$Q,\Sigma$, and~$F$ are the same as for DFAs,
$I\subseteq Q$ is the set of initial states, 
and
$\cdot \,\colon Q\times(\Sigma\cup\{\eps\})\to 2^Q$
is the transition function.
A string~$w$ in~$\Sigma^*$ is \emph{accepted} by the NFA~$N$
if~$w=a_1a_2\cdots a_m$ where~$a_i\in\Sigma\cup\{\eps\}$
and 
a sequence of states~$q_0,q_1,\ldots,q_m$
exists in~$Q$
such that~$q_0\in I$, $q_{i+1}\in q_i\cdot a_{i+1}$
for~\mbox{$i=0, 1, \ldots,m-1$}, and~$q_m\in F$.
The \emph{language recognized} by the NFA~$N$ 
is~$L(N)=\{w\in\Sigma^*\mid \text{$w$ is accepted by~$N$}\}$.
For~$p,q\in Q$ and~$a\in\Sigma\cup\{\eps\}$, 
we say that a triple~$(p,a,q)$ is a \emph{transition} in $N$ if $q\in p\cdot a$.

Let~$N=(Q,\Sigma,\cdot,I,F)$ be an~NFA.
For a set~$S\subseteq Q$, let~$E(S)$ denote
\mbox{the~$\eps$-closure} of~$S$; that is, 
the set of states~$\{q\mid q$  is reached from a state in~$S$ 
through~0 or more~$\eps$-transitions$\}$.
The \emph{subset automaton} of the  NFA~$N$ 
is the~DFA~$\cD(N)=(2^Q,\Sigma,\cdot',E(I),F')$ 
 where~$F'=\{S\in 2^Q\mid S\cap F\ne\emptyset\}$  
and~$S\cdot' a = \cup_{q\in S} E(q\cdot a)$
for each~$S\in 2^Q$ and each~$a\in \Sigma$.
The subset automaton~$\cD(N)$
recognizes the language~$L(N)$. 
 
The \emph{reverse} of the NFA  $N$ is the NFA
$N^R=(Q,\Sigma,\cdot^R,F,I)$ 
where the transition function is defined by $q\cdot^R a=\{p\in Q \mid q\in p\cdot a\}$;
that is,~$N^R$ is obtained from~$N$
by swapping the roles of initial and final states,
and by reversing all transitions.

A subset $S$ of $Q$ is \emph{reachable} in $N$ 
if there is a string $w$ in $\Sigma^*$ such that~$S=I\cdot w$,
and it is \emph{co-reachable} in $N$ if it is reachable in the reverse $N^R$.
We use the following two simple observations
to prove distinguishability of states in subset automata.

\begin{lemma}
\label{le:dist}
 Let~$N=(Q,\Sigma,\cdot,I,F)$ be an NFA.
% without~$\eps$-transitions.
 Let~$S,T\subseteq Q$  and $q\in S\setminus T$.
 If  the  singleton set~$\{q\}$ is co-reachable in $N$,
 then $S$ and $T$ are distinguishable in
 the subset automaton~$\cD(N)$.
\end{lemma}

\begin{proof}
 Since the singleton set $\{q\}$ is co-reachable in $N$,
 there is a string~$w\in \Sigma^*$ which sends the set of final states~$F$
 to~$\{q\}$ in the reversed automaton~$N^R$.
 It follows that the string~$w^R$ is accepted by~$N$
 from the state~$q$, and it is rejected from any other state.
 Thus, the string~$w^R$ is accepted 
 by~$\cD(N)$ from~$S$ and rejected from~$T$. 
\end{proof}

\begin{corollary}
\label{cor}
 If for each state $q$ of an NFA $N$,
 the singleton set $\{q\}$ is co-reachable in $N$,
 then all states of the subset automaton $\cD(N)$
 are pairwise distinguishable.
 \qed
\end{corollary}

\section{Multiple Concatenation: Upper Bound}
\label{sec:upper}

In this section, we recall the constructions
of~$\eps$-NFAs and NFAs
for multiple concatenation,
as well as the known upper bounds.
We also provide a simple alternative method
to get upper bounds. 
In the last part of this section,
we consider the case when some of given automata
have just one state.

For~$i=1,2,\ldots,k$,
let~$A_i=(Q_i,\Sigma,\cdot_i,s_i,F_i)$
be a DFA,
and assume that~$Q_i\cap Q_j=\emp$ if~$i\ne j$.
Then the concatenation~$L(A_1)L(A_2)\cdots L(A_k)$
is recognized by an NFA
$
 N=(Q_1\cup Q_2\cup\cdots\cup Q_k,\Sigma,\cdot,s_1,F_k),
$
where for each~$i=1,2,\ldots,k$, each~$q\in Q_i$,
and each~$a\in\Sigma$, we have~$ q\cdot a =\{q\cdot_i a\}$
and for each~$i=1,2,\ldots,k-1$ and each~$q\in F_i$,
we have~$q\cdot\eps=\{s_{i+1}\}$, that is,
the NFA~$N$ is obtained from the DFAs~$A_1,A_2,\ldots,A_k$
by adding the~$\eps$-transition from each final state of~$A_i$ to the initial state~$s_{i+1}$
of~$A_{i+1}$ for~$i=1,2,\ldots,k-1$;
the initial state of~$N$ is~$s_1$,
and its set of final states is~$F_k$.

Since~$A_1$ is a complete DFA,
in the corresponding subset automaton~$\cD(N)$,
each reachable subset is of the 
form~$\{q\}\cup S_2\cup S_3\cup\cdots\cup S_k$
where~$q\in S_1$ and~$S_i\subseteq Q_i$
for~$i=2,3,\ldots,k$.  
We represent such a set by 
the~$k$-tuple~$(\{q\},S_2,S_3,\ldots,S_k)$,
or more often by~$(q,S_2,S_3,\ldots,S_k)$,
and with this representation,
it is not necessary to have the state sets disjoint.
Nevertheless, since we sometimes use
special properties of the NFA~$N$,
we keep in mind that this~$k$-tuple
represents the union of appropriate set of states
of the corresponding DFAs.
We usually denote all transition functions by~$\cdot$,
and simply write~$(qa,S_2,S_3,\ldots,S_k)$
or~$(q,S_2a,S_3,\ldots,S_k)$; that is,
applying~$a$ to the~$i$-th component
means that we use the transition function~$\cdot_i$.

It follows from the construction of the NFA~$N$ 
that if~$S_i\cap F_i\ne\emp$ then~$s_{i+1}\in S_{i+1}$,
and if~$S_i=\emp$, then~$S_{i+1}=\emp$
in any reachable state~$(S_1,S_2,\ldots,S_k)$
of the subset automaton~$\cD(N)$.
The states satisfying the above mentioned properties
are called valid in~\cite{clp18}; 
let us summarize the three properties in the next definition.

\begin{definition}
    A state~$(S_1,S_2,\ldots,S_k)$
    of the subset automaton~$\cD(N)$ is \emph{valid} if
    \begin{enumerate}
        \item $|S_1|=1$,
        \item if~$S_i=\emp$ and~$i\le k-1$, then~$S_{i+1}=\emp$,
        \item if~$S_i\cap F_i\ne\emp$ and~$i\le k-1$,
        then~$s_{i+1}\in S_{i+1}$.
    \end{enumerate}
\end{definition}

Since each reachable state  
of~$\cD(N)$ is valid,
we have the next observation.

\begin{proposition}
\label{prop:valid}
 An upper bound on 
 $\dsc(L(A_1)L(A_2)\cdots L(A_k))$
 is given by the number of valid states
 in the subset automaton~$\cD(N)$. 
\qed
\end{proposition}

Notice that, to reach as many valid states as possible,
each automaton~$A_i$ with~$i\le k-1$ should have exactly one final state~$f_i$,
that is, we have~$F_i=\{f_i\}$.
Moreover, if~$A_i$ has at least two states,
then we should have~$s_i\ne f_i$.
If this is the case for all~$A_i$, 
then we 
can construct an NFA~$N$ for the 
concatenation~$L(A_1)L(A_2)\cdots L(A_k)$
from the DFAs~$A_1,A_2,\ldots,A_k$ as follows:
for each~$i=1,2,\ldots,k-1$,
each state~$q\in Q_i$, and each symbol~$a\in\Sigma$ such that~$q\cdot_i a=f_i$,
we add the transition~$(q,a,s_{i+1})$;
the initial state of~$N$ is~$s_1$,
and its unique final state is~$f_k$.

For~$k=2$, an upper bound 
on the number of valid states
is~$(n_1-1)2^{n_2}+2^{n_2-1}$~\cite{yzs94},
which is the sum of  the number of states~$(q,S_2)$
with~$q\ne f_1$ and~$S_2\subseteq Q_2$
and the number of states~$(f_1,S_2)$
with~$s_2\in S_2$.
For~$k\ge3$, we have the following inequalities.

\begin{proposition}
\label{prop:ineq}
\label{page}
 Let~$k\ge3$ and $\# \tau_k$ denote the number of valid states.
 Then
 \begin{align*}
     \frac{1}{2^{k-1}} \,n_1 \, 2^{n_2+n_3+\cdots+n_k} \le \# \tau_k \le \frac{3}{4}\, n_1\, 2^{n_2+n_3+\cdots+n_k}.  
 \end{align*}
\end{proposition}

\begin{proof}
    Every state~$(S_1,S_2,\ldots,S_k)$ 
    with $s_i\in S_i$ for $i=2,3,\ldots,k$
    is a valid state. This gives the left inequality.
    On the other hand,
    every   state $(S_1,S_2,\ldots,S_k)$ 
    with~$f_2\in S_2$ and $s_3\notin S_3$
    is not valid,
    which gives the right inequality.
\end{proof}

\medskip

We now provide a simple alternative method
for obtaining an upper bound on the number of valid states.
To this aim let for~$i=2,3,\ldots,k$,
\begin{itemize}
    \item $U_i$ be the number of tuples $(S_i,S_{i+1},\ldots,S_k)$
such that for fixed~$S_1', \ldots,S_{i-1}'$  
with~\mbox{$f_{i-1}\notin S'_{i-1}$}
the state~$(S_1',\ldots,S_{i-1}',S_i,S_{i+1},\ldots,S_k)$ is valid,
\item $V_i$ be the number of tuples $(S_i,S_{i+1},\ldots,S_k)$
such that for a fixed~$S_1',\ldots,S_{i-1}'$  
with~$f_{i-1}\in S'_{i-1}$
the state~$(S_1',\ldots,S_{i-1}',S_i,S_{i+1},\ldots,S_k)$ is valid.
\end{itemize}
Then
we have the next result.

\begin{theorem}
\label{thm:upper}
    Let~$k\ge2$, ~$n_i\ge2$ for~$i=1,2,\ldots,k$,
    and~$A_i=(Q_i,\Sigma,\cdot,s_i,\{f_i\})$
    be an~$n_i$-state DFA with~$s_i\ne f_i$.
    Let~$U_i$ and~$V_i$ be as defined above,
    and~$\# \tau_k$ be the number of valid states
    in the subset automaton~$\cD(N)$
    accepting~$L(A_1)L(A_2)\cdots L(A_k)$.
    Then
    \begin{align}
    \label{u1}
        & U_k=2^{n_k} \text{ and }    V_k=2^{n_k-1}, 
    \end{align}
    and for~$i=2,3,\ldots, k-1$,
    \begin{align}
     \label{u2}
        & U_i = 1 + (2^{n_i-1}-1)U_{i+1}+2^{n_i-1}V_{i+1}, \\
    \label{u3}     
         &V_i=2^{n_i-2}(U_{i+1}+V_{i+1}). 
       \end{align}   
         Finally, we have   
    \begin{align}
     \label{u4} 
         & \# \tau_k=(n_1-1)U_2+V_2.
    \end{align}
\end{theorem}

\begin{proof}
    If~$f_{k-1}\notin S'_{k-1}$,
    then~$S_k$ may be an arbitrary subset of~$Q_k$.
    If~$f_{k-1}\in S'_{k-1}$,
    then~$S_k$ must contain~$s_k$.
       This gives~(\ref{u1}).

    Let~$f_{i-1}\notin S'_{i-1}$.
    We have just one tuple with~$S_i=\emp$,
    namely,~$(\emp,\emp,\ldots,\emp)$,
    then $(2^{n_i}-1)U_{i+1}$
    tuples with~$f_i\notin S_i$ and~$S_i$ non-empty,
    and~$2^{n_i-1}V_{i+1}$
    tuples with~$f_i\in S_i$ final. 
    This gives~(\ref{u2}). 

    Let~$f_{i-1}\in S'_{i-1}$.
    Then~$s_i\in S_i$.
    We have~$(2^{n_i}-2)U_{i+1}$
    tuples with~$s_i\in S_i$ and~$f_i\notin S_i$,
    and~$2^{n_i-2}V_{i+1}$
    tuples with~$s_i\in S_i$ and~$f_i\notin S_i$.
    This gives~(\ref{u3}).

    Finally, we have~$(n_1-1)$ possibilities
    for~$S'_1$ to be non-final singleton set,
    and one,  $S'_1=\{f_1\}$, to be final.
    This gives~(\ref{u4}).
   \end{proof}

Let us illustrate the above result in the 
following example.
  
\begin{example}\rm
    Let~$k=3$ and~$n_1,n_2,n_3\ge2$. 
    Then
    \begin{align*}
       U_3= &2^{n_3} \text{ and } V_3=2^{n_3-1}, \\
       U_2= &1+(2^{n_2-1}-1)U_3 +2^{n_2-1}V_3
        =1+(2^{n_2-1}-1)2^{n_3} +2^{n_2-1}2^{n_3-1},\\
       V_2=&2^{n_2-2}(U_3+V_3)
         =2^{n_2-2}(2^{n_3}+2^{n_3-1}) \\
       \# \tau_k  =&(n_1-1)U_2+V_2= \\
    & (n_1-1)(1+(2^{n_2-1}-1)2^{n_3} +2^{n_2-1}2^{n_3-1})+2^{n_2-2}(2^{n_3}+2^{n_3-1})= \\
     &  n_1(1+2^{n_2+n_3-1}-2^{n_3}+2^{n_2+n_3-2})
     -1 -2^{n_2+n_3-1}+2^{n_3}-2^{n_2+n_3-2} +  \\     
     & 2^{n_2+n_3-2} + 2^{n_2+n_3-3} =\\
    & n_1(1+\frac{3}{4}2^{n_2+n_3}-2^{n_3}) 
     -\frac{3}{8} 2^{n_2+n_3}+2^{n_3}-1,  
    \end{align*}
     which is the same as  
     in~\cite[Example~3.6]{clp18}.
     \qed
\end{example}

To conclude this section, 
let us consider
also the case when some automata have just one state.
If this state is non-final, then the resulting concatenation
is empty. Thus, assume that all one-state automata
recognize~$\Sigma^*$, so consist of one initial and final state~$f_i$.
We construct an NFA~$N$  accepting the language~$L(A_1)L(A_2)\cdots L(A_k)$
as described above.
Let~$\cD(N)$ be the corresponding subset automaton.
We represent its states by $k$-tuples~$(\{q\},S_2,S_3,\ldots,S_k)$
where~$q\in Q_1$ and~$S_i\subseteq Q_i$.
Moreover, if~$n_i=1$, then~$S_i=\{f_i\}$.
If~$n_i\ge2$ and~$i<k$, then to get maximum number of valid reachable sets,
we must have~$F_i=\{f_i\}$ and~$s_i\ne f_i$.
The next observation provides an upper bound
in the case when exactly one of given DFAs has one state. 

\begin{proposition}
\label{prop:jeden1-state}
    Let~$k\ge2$, $j\in\{1,2,\ldots,k\}$,
     $n_j=1$, and~$n_i\ge2$ if~$i\ne j$.
     For~$i=1,2,\ldots,k$,
     let~$A_i$ be an~$n_i$-state DFA
     and~$L=L(A_1)L(A_2)\cdots L(A_k)$.
     Let~$U_i$ and~$V_i$ be given by expressions~(\ref{u2})-(\ref{u3}).
     Then
     \[ \dsc(L)\le
     \begin{cases}
         V_2,   % & \\
         % \text{\ \ \  with~$U_{k}=2^{n_{k}}$ 
         % and~$V_{k}=2^{n_{k}-1}$},
         &\text{if~$j=1$}; \\
         n_1, & \text{if~$j=k=2$}; \\
          (n_1-1)U_2+V_2+1 &  \\
          \text{\ \ \  with~$U_{k-1}=2^{n_{k-1}-1}$ 
         and~$V_{k-1}=2^{n_{k-1}-2}$}, &\text{if~$j=k\ge3$}; \\
         (n_1-1)U_2+V_2+V_{i+1} & \\
    \text{\ \ \ with~$U_{j-1}=2^{n_{j-1}-1}$
     and~$V_{j-1}=2^{n_{i-1}-2}$}, 
     &\text{if~$2\le j \le k-1$}.        
     \end{cases}
     \]
\end{proposition}

\begin{proof}
    First, let~$j=1$. Then we have~$S_1=\{f_1\}$
    in each valid state~$(S_1,S_2,\ldots,S_k)$.
    It follows that the number of valid states
    is~$V_2$ with~$U_{k}=2^{n_{k}}$ 
         and~$V_{k}=2^{n_{k}-1}$. % in this case.

    Now, let~$j=k$.
    Then all states~$(S_1,S_2,\ldots,S_{k-1},\{f_k\})$
    are equivalent to a final sink state.
    If~$S_k=\emp$, then~$f_{k-1}\notin S_{k-1}$.
    This results in an upper bound~$(n_1-1)U_2+V_2+1$
    with~$U_{k-1}=2^{n_{k-1}-1}$ and~$V_{k-1}=2^{n_{k-1}-2}$
    if~$k\ge3$ and~$(n_1-1)+1$ if~$k=2$.
    
    Finally, let~$2\le j\le k-1$.
    Then all 
    states~$(S_1,S_2,\ldots,S_{i-1},\{f_i\},\{s_{i+1}\},\emp,\emp,\ldots,\emp)$
    are equivalent to the 
    state~$(\{s_1\},\{s_2\},\ldots,\{s_{i-1}\},\{f_i\},\{s_{i+1}\},\emp,\emp,\ldots,\emp)$
    since we have a loop on each input symbol
    in the state~$f_i$ and therefore every string 
    accepted by~$N$ from a state 
    in~$Q_1\cup Q_2 \cup \cdots\cup Q_{i-1}$
    is accepted also from~$f_i$.
    It follows that the reachable and pairwise 
    distinguishable valid states of the subset automaton~$\cD(N)$ are 
    either of the form~$(S_1,S_2,\ldots,S_{i-1},\emp,\emp,\ldots,\emp)$
    or of the form~$(\{s_1\},\{s_2\},\ldots,\{s_{i-1}\},\{f_i\}, S_{i+1},S_{i+2}, \ldots,S_k)$.
    If~$S_i=\emp$,
    then~$S_{i-1}$ does not contain~$f_i$,
    so the number of valid states
    of the first form is given 
    by~$(n_i-1)U_2+V_2$
    with~$U_{i-1}=2^{n_{i-1}-1}$
    and~$V_{i-1}=2^{n_{i-1}-2}$.
    The number of valid states of the second form
    is given by~$V_{i+1}$.    
\end{proof}

\begin{example}\rm
    Let~$k=4$, ~$n_3=1$, and~$n_1,n_2,n_4\ge2$.
    Then number of valid states~$(S_1,S_2,\emp,\emp)$
    is~$(n_1-1)U_2+V_2$
    where~$U_2=2^{n_2-1}$ and~$V_2=2^{n_2-2}$.
    The number of valid
    states~$(\{s_1\},\{s_2\},\{f_3\},S_4)$
    is~$V_4=2^{n_4-1}$.
    This gives an upper 
    bound~$(n_1-1)2^{n_2-1}+2^{n_2-2}+2^{n_4-1}$
    for concatenation of four languages,
    the third of which is~$\Sigma^*$.
    \qed
\end{example}

Finally, we provide an upper bound in a more general case.

\begin{proposition}
	\label{prop:more-one-state}
    Let~$k\ge3$. Let~$[i_1,j_1],[i_2,j_2],\ldots,[i_m,j_m]$
    with~$j_m<k$, ~$i_\ell\le j_\ell$, and~$i_{\ell+1}\ge j_\ell+2$      
    be intervals on which~$n_i\ge2$
    and with~$n_i=1$ outside these intervals.
    For~$i=1,2,\ldots,k$,
     let~$A_i$ be an~$n_i$-state DFA
     and~$L=L(A_1)L(A_2)\cdots L(A_k)$.
    Set~$U_{j_i}=2^{n_{j_i}-1}$ and~$V_{j_i}=2^{n_{j_i}-2}$
    for~$i=1,2,\ldots,m$.
    Then
    \[   \dsc(L)\le
 \begin{cases}
    V_{i_1}+V_{i_2}+V_{i_3}+ \cdots+ V_{i_m}, 
    &\text{if~$i_1\ge2$};\\
       (n_1-1)U_2+V_2+V_{i_2}+V_{i_3}+ \cdots+ V_{i_m}, &\text{if~$i_1=1$ and~$j_1\ge2$};\\
        (n_1-1)+V_{i_2}+V_{i_3}+ \cdots+ V_{i_m},
     &\text{if~$i_1=1$ and~$j_1=1$}.
 \end{cases}
    \]
\end{proposition}

\begin{proof}
   For each~$i$ with~$n_i=1$,
   the unique state~$f_i$ of~$Q_i$ has a loop on each input symbol. Therefore   a state~$(S_1,S_2,\ldots,S_{i-1},\{f_i\},S_{i+1},\ldots,S_k)$
   is equivalent to~$(\{s_1\},\{s_2\},\ldots,\{s_{i-1}\},\{f_i\},S_{i+1},\ldots,S_k)$.

   First, let~$i_1\ge2$.
   Then for each interval~$[i_\ell,j_\ell]$,
   we need to count the  valid 
  states of the form~$(\{s_1\},\{s_2\},\ldots,\{s_{i_{\ell}-2}\}, \{f_{i_{\ell}-1}\},    S_{i_\ell},S_{i_\ell+1},\ldots,S_{j_\ell},\emp,\ldots,\emp)$.
  The number of such states is given by~$V_{i_\ell}$
  where~$U_{j_\ell}$ and~$V_{j_\ell}$
  are given in the statement.

  If~$i_1=1$ and~$j_1\ge2$, then we
  need to consider also
      states 
  of the form~$(S_1,S_2,\ldots,S_{j_1},\emp,\emp,\ldots\emp)$,
  and there are~$(n_1-1)U_2+V_2$ of them.

  If~$i_1=j_1=1$, then we need to consider also the
  valid states~$(S_1,\emp,\ldots,\emp)$,
  and there are~$(n_1-1)$ of them.
\end{proof}

 \begin{example}\rm
     Let~$k=5$, $n_1=n_2=n_5=3$, and~$n_3=n_4=1$.
     Then the number of valid states~$(S_1,S_2,\emp,\emp,\emp)$
     is~$2\cdot 2^2+2^1$,
     and the number of valid 
     states~$(\{s_1\},\{s_2\},\{f_3\},\{f_4\},S_5)$
     is~$2^2$.
     The upper bound 14 is met
     by one-state automata~$A_3$ and~$A_4$
     and three-state automata~$A_1,A_2,A_5$ with the transitions: 

     in $A_1$: \quad  $a\colon (1,2,3)$;

     in $A_2$: \quad   $b\colon (1,2,3)$, $e\colon (1,2)$,  $f\colon (\{1,2\}\to 1)$;  

     in~$A_5$: \quad $c\colon (1,2,3)$,   $d\colon (2 \to 3\to 1)$;\\
     and all  symbols that are not shown perform the identities.     
 \qed\end{example}

 We conjecture that if all automata have one or at least three states,
 then upper bounds from Proposition~\ref{prop:more-one-state} are tight.
 However, it looks like these upper bounds can further
 be decreased if a two-state automaton precedes the one-state automaton.
 If all automata except for~$A_k$ have two states,
 and~$A_k$ has one,
 it looks like their concatenation is recognized by a~$k$-state DFA.
 We leave the problem of finding a tight upper bound
 for such cases open.
 In what follows, we assume that
 all automata have at least two states.% ~$n_i\ge2$

\section{Matching Lower Bound:     \textbf({\textit{k+1})}-letter Alphabet}

In   this section, we describe   witness languages
meeting the upper bound on the state complexity of multiple concatenation
 of~$k$ regular languages
over a~$(k+1)$-letter alphabet
with a significantly simpler proof than that  in \cite[Section~4, pp.~266--271]{clp18}.
We use these witnesses in the next section to 
describe witness languages over a $k$-letter alphabet.
Let us start with the following example.

\begin{example}\rm
Let~$n_1,n_2\ge3$. Consider DFAs~$A_1$ and~$A_2$ over~$\{a_1,a_2\}$ shown in Figure~$\ref{fig:ex1}$.
The symbol~$a_1$ performs the circular shift in~$A_1$,
and the identity in~$A_2$.
Symmetrically, the symbol~$a_2$ performs the identity in~$A_1$, and the circular shift in~$A_2$.

\begin{figure}[h!]
\centering 
\begin{tikzpicture}[>=stealth', initial text={},shorten >=1pt,auto,node distance=2cm]

\node[state,initial,initial text={$A_1$}  ]  (q1)  [label=center:{$s_1{=}1$}]{};
\node[state](q2) [right of=q1,label=center:{$2$}]{};
\node[state,draw=none] (q3)[right of=q2,label=center:{$\ldots$}]{};
\node[state] (q4) [right of=q3,label=center:{$n_1{-}1$}]{}; 
\node[state,accepting ] (q5)  [right of=q4,label=center:{$f_1{=}n_1$}]{}; 

\node[state,initial,initial text={$A_2$}  ]  (1) at (2,-3)  [ label=center:{$s_2{=}1$}]{};
\node[state](2) [right of=1,label=center:{$2$}]{};
\node[state,draw=none] (3)[right of=2,label=center:{$\ldots$}]{};
\node[state] (4) [right of=3,label=center:{$n_2{-}1$}]{}; 
\node[state,accepting ] (5)  [right of=4,label=center:{$f_2{=}n_2$}]{}; 

\draw[->]  (q1) to node{$a_1$} (q2);
\draw[->]  (q2) to node{$a_1$} (q3);
\draw[->]  (q3) to node{$a_1$} (q4);
\draw[->]  (q4) to node{$a_1$} (q5);
\draw[->]  (q5) [bend left=20] to node[above]{$a_1$}(q1);

\draw[->](q1)to[loop above]node {$a_2$}(q1);
\draw[->](q2)to[loop above]node {$a_2$}(q2);
\draw[->](q4)to[loop above]node {$a_2$}(q4);
\draw[->](q5)to[loop above]node {$a_2$}(q5);

\draw[->]  (1) to node{$a_2$} (2);
\draw[->]  (2) to node{$a_2$} (3);
\draw[->]  (3) to node{$a_2$} (4);
\draw[->]  (4) to node{$a_2$} (5);
\draw[->]  (5) [bend left=20] to node[above]{$a_2$}(1);

\draw[->](1)to[loop above]node {$a_1$}(1);
\draw[->](2)to[loop above]node {$a_1$}(2);
\draw[->](4)to[loop above]node {$a_1$}(4);
\draw[->](5)to[loop above]node {$a_1$}(5);
\end{tikzpicture}
\caption{DFAs  $A_1$  and~$A_2$ with all valid states reachable in~$\cD(N)$.}  
\label{fig:ex1}
\end{figure}
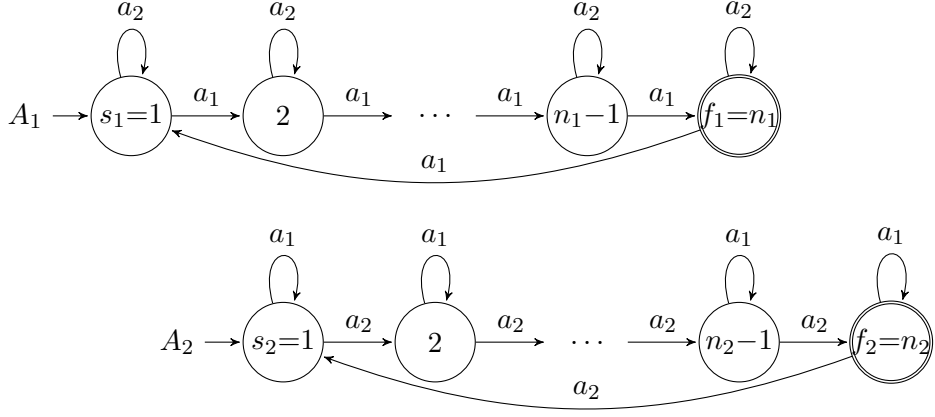

Construct the NFA~$N$ recognizing the language~$L(A_1)L(A_2)$
from the DFAs~$A_1$ and~$A_2$ by adding
the transitions~$(f_1,a_2,s_2)$ and~$(f_1-1,a_1,s_2)$,
by making the state~$f_1$ non-final and state~$s_2$ non-initial.
The NFA~$N$ is shown in Figure~\ref{fig:exN}.

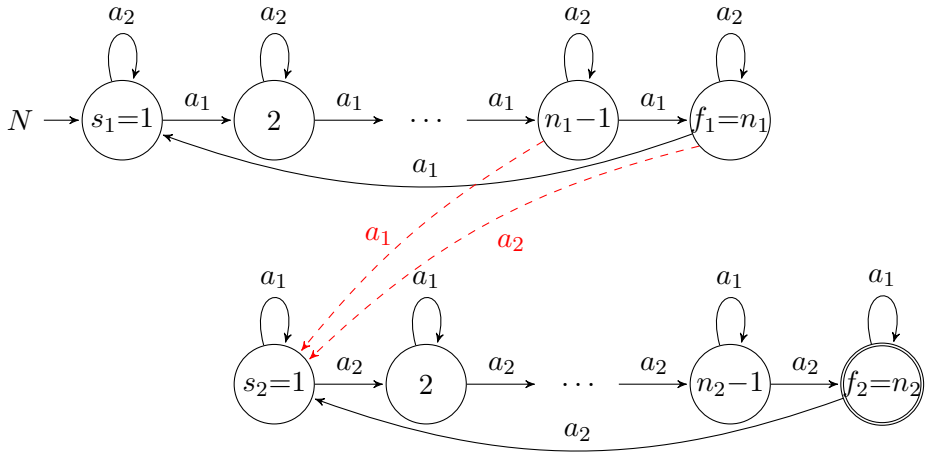
\begin{figure}[b]
\centering 
\begin{tikzpicture}[>=stealth', initial text={},shorten >=1pt,auto,node distance=2cm]

\node[state,initial,initial text={$N$}  ]  (q1)  [label=center:{$s_1{=}1$}]{};
\node[state](q2) [right of=q1,label=center:{$2$}]{};
\node[state,draw=none] (q3)[right of=q2,label=center:{$\ldots$}]{};
\node[state] (q4) [right of=q3,label=center:{$n_1{-}1$}]{}; 
\node[state] (q5)  [right of=q4,label=center:{$f_1{=}n_1$}]{}; 

\node[state]  (1) at (2,-3.5)  [ label=center:{$s_2{=}1$}]{};
\node[state](2) [right of=1,label=center:{$2$}]{};
\node[state,draw=none] (3)[right of=2,label=center:{$\ldots$}]{};
\node[state] (4) [right of=3,label=center:{$n_2{-}1$}]{}; 
\node[state,accepting ] (5)  [right of=4,label=center:{$f_2{=}n_2$}]{}; 

\draw[->]  (q1) to node{$a_1$} (q2);
\draw[->]  (q2) to node{$a_1$} (q3);
\draw[->]  (q3) to node{$a_1$} (q4);
\draw[->]  (q4) to node{$a_1$} (q5);
\draw[->]  (q5) [bend left=20] to node[above]{$a_1$}(q1);

\draw[->](q1)to[loop above]node {$a_2$}(q1);
\draw[->](q2)to[loop above]node {$a_2$}(q2);
\draw[->](q4)to[loop above]node {$a_2$}(q4);
\draw[->](q5)to[loop above]node {$a_2$}(q5);

\draw[->]  (1) to node{$a_2$} (2);
\draw[->]  (2) to node{$a_2$} (3);
\draw[->]  (3) to node{$a_2$} (4);
\draw[->]  (4) to node{$a_2$} (5);
\draw[->]  (5) [bend left=20] to node[above]{$a_2$}(1);

\draw[->](1)to[loop above]node {$a_1$}(1);
\draw[->](2)to[loop above]node {$a_1$}(2);
\draw[->](4)to[loop above]node {$a_1$}(4);
\draw[->](5)to[loop above]node {$a_1$}(5);

\draw[->, red, dashed]  (q4)[bend right=10,above,pos=0.6] to node{$a_1~~$} (1);
\draw[->, red, dashed]  (q5.220) [bend right=15]to node{$a_2$} (1.35);
\end{tikzpicture}
\caption{The NFA~$N$ recognizing the language~$L(A_1)L(A_2)$.}  
\label{fig:exN}
\end{figure}

Let us show that each valid state~$(j,S)$ is reachable
in the subset automaton~$\cD(N)$.
The proof is by induction on~$|S|$.
The basis, with~$|S|=0$, holds true
since each state~$(j,\emptyset)$ with~$j\le n_1-1$
is reached from the initial state~$(s_1,\emptyset)$
by~$a_1^{j-1}$.
Let~$|S|\ge1$. There are three cases to consider.

\medskip\noindent{\it Case~1:}
$j=f_1$. Then~$s_2\in S$ since~$(f_1,S)$ is valid.
Since~$a_1$ performs the circular shift in $A_1$,
and the identity in~$A_2$, we have
$
(n_1-1, S\setminus\{s_2\})
\xrightarrow{a_1}$ $(f_1,\{s_2\}\cup(S\setminus\{s_2\}))=(f_1,S),
$
where the leftmost state is reachable by induction.

\medskip\noindent{\it Case~2:}
$j=s_1$. Let~$m=\min S$. Then~$s_2\in a_2^{m-1}(S)$,
and~$|a_2^{m-1}(S)|=|S|$ since~$a_2$ performs a permutation on the state set of~$A_2$.
Since~$a_1$ performs the identity on the state set of~$A_2$, we have
\[
(f_1,a_2^{m-1}(S))
\xrightarrow{a_1}(s_1,a_2^{m-1}(S))
\xrightarrow{a_2^{m-1}}=(s_1,S),
\]
where the leftmost state is reachable
as shown in Case~1.

\medskip\noindent{\it Case~3:}
$2\le j \le n_1-1$. 
Then we have~$(s_1,S)\xrightarrow{a_1^{j-1}}(j,S)$,
where the left state is considered in Case~2.

\medskip
Thus, the two simple symbols~$a_1$ and~$a_2$ 
guarantee
the reachability of all valid states
in the subset automaton~$\cD(N)$.
However, since both these symbols perform
permutations on the state set~$Q_2$ of~$A_2$,
we have~$Q_2\cdot a_1=Q_2\cdot a_2=Q_2$.
It follows that in~$\cD(N)$, 
all states~$(i,Q_2)$ are equivalent 
to the final sink state.

To guarantee distinguishability,
we add one more input symbol~$b$
which performs the contractions~$s_1\to 2$
and~$s_2\to 2$, 
and denote the resulting automata~$A_1'$ and~$A_2'$,
respectively.
The NFA~$N'$ recognizing~$L(A_1')L(A_2')$
is shown in Figure~\ref{fig:exN'}.

\begin{figure}[h!]
\centering 
\begin{tikzpicture}[>=stealth', initial text={},shorten >=1pt,auto,node distance=2cm]

\node[state,initial,initial text={}  ]  (q1)  [label=center:{$s_1{=}1$}]{};
\node[state](q2) [right of=q1,label=center:{$2$}]{};
\node[state,draw=none] (q3)[right of=q2,label=center:{$\ldots$}]{};
\node[state] (q4) [right of=q3,label=center:{$n_1{-}1$}]{}; 
\node[state] (q5)  [right of=q4,label=center:{$f_1{=}n_1$}]{}; 

\node[state]  (1) at (2,-4)  [ label=center:{$s_2{=}1$}]{};
\node[state](2) [right of=1,label=center:{$2$}]{};
\node[state,draw=none] (3)[right of=2,label=center:{$\ldots$}]{};
\node[state] (4) [right of=3,label=center:{$n_2{-}1$}]{}; 
\node[state,accepting ] (5)  [right of=4,label=center:{$f_2{=}n_2$}]{}; 

\draw[->]  (q1) to node{$a_1,b$} (q2);
\draw[->]  (q2) to node{$a_1$} (q3);
\draw[->]  (q3) to node{$a_1$} (q4);
\draw[->]  (q4) to node{$a_1$} (q5);
\draw[->]  (q5) [bend left=20] to node[above]{$a_1$}(q1);

%\draw[->] (q1) [bend left] to node[above]{$b$}(q2);
\draw[->] (q1)to[loop above]node {$a_2$}(q1);
\draw[->] (q2)to[loop above]node {$a_2,b$}(q2);
\draw[->] (q4)to[loop above]node {$a_2,b$}(q4);
\draw[->] (q5)to[loop above]node {$a_2,b$}(q5);

\draw[->]  (1) to node{$a_2,b$} (2);
\draw[->]  (2) to node{$a_2$} (3);
\draw[->]  (3) to node{$a_2$} (4);
\draw[->]  (4) to node{$a_2$} (5);
\draw[->]  (5) [bend left=20] to node[above]{$a_2,b$}(1);

%\draw[->]  (1) [bend left] to node[above]{$b$}(2);
\draw[->](1)to[loop above]node {$a_1$}(1);
\draw[->](2)to[loop above]node {$a_1,b$}(2);
\draw[->](4)to[loop above]node {$a_1,b$}(4);
\draw[->](5)to[loop above]node {$a_1,b$}(5);

\draw[->,red,dashed]  (q4)[bend right=10,above,pos=0.6] to node{$a_1~~$} (1);
\draw[->,red,dashed]  (q5.220) [bend right=15,pos=0.4]to node{$a_2,b$} (1.35);
\end{tikzpicture}
\caption{The NFA recognizing the language~$L(A_1')L(A_2')$.}  
\label{fig:exN'}
\end{figure}
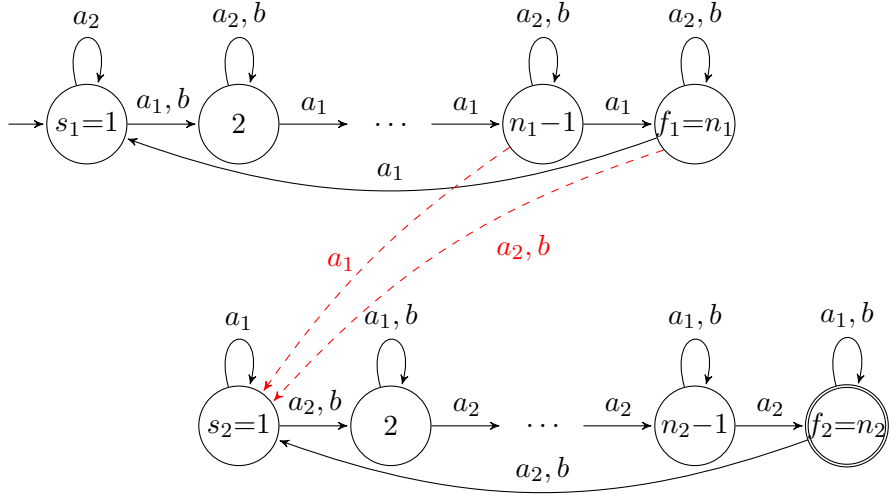

As shown above,
all valid states~$(j,S)$
are reachable
in the corresponding subset automaton~$\cD(N')$.
To get distinguishability, 
let us show that each singleton set
is co-reachable in~$N'$.
In the reversed automaton~$(N')^R$,
the initial set is~$\{f_2\}$,   
and 
 \[
\{f_2\}\xrightarrow{a_2}\{n_2-1\}
 \xrightarrow{a_2}\{n_2-2\}
 \xrightarrow{a_2}\cdots
 \xrightarrow{a_2} \{2\}
 \xrightarrow{a_2} \{s_2\}.
 \]
 Next, since~$n_1\ge3$, we have
 \[
  \{s_2\}\xrightarrow{b}\{f_1\}
 \xrightarrow{a_1}\{n_1-1\}
 \xrightarrow{a_1}\cdots
 \xrightarrow{a_1} \{2\}
 \xrightarrow{a_1} \{s_1\}; 
 \]
 notice that we need~$n_1\ge3$
 to get~$\{s_2\}\xrightarrow{b} \{f_1\}$,
 in the case of~$n_1=2$ we would 
 have~$\{s_2\}\xrightarrow{b} \{f_1,s_1\}$.
Hence each singleton set is co-reachable
in~$N'$.
By Corollary~\ref{cor}, all states 
of the subset automaton~$\cD(N')$
are pairwise distinguishable.
\qed
\end{example}

We use the ideas from the above example
to describe witnesses for multiple concatenation
over a~$(k+1)$-letter alphabet.
To this aim,
let~$k\ge2$ and~$n_i\ge3$ for~$i=1,2,\ldots,k$.

Let~$\Sigma=\{b,a_1,a_2,\ldots,a_k\}$
be an alphabet consisting of~$k+1$ symbols.
For~$i=1,2,\ldots,k$,
define  an~$n_i$-state   DFA $A_i=(Q_i,\Sigma,\cdot,s_i, \{f_i\})$, where
\begin{itemize}
    \item $Q_i=\{1,2,\ldots,n_i\}$,
    \item $s_i=1$,
    \item  $f_i=n_i$, 
    \item $a_i \colon (1,2,\ldots,n_i)$, \quad
           $a_j \colon (1)$ if~$j\ne i$, \quad
           $b\colon (1\to 2)$,
\end{itemize}  
that is, the symbol $a_i$ performs the circular shift on~$Q_i$, 
each symbol $a_j$ with $j\ne i$ performs the identity,
and the symbol~$b$ performs a contraction.
The DFA $A_i$ is shown in Figure~\ref{fig:Ai};
here~$\Sigma\setminus\{a_i\}$
on a loop means that there is a loop in the corresponding state on each symbol in~$\Sigma\setminus\{a_i\}$,
and the same for~$\Sigma\setminus\{a_i,b\}$.

\begin{figure}[h!]
\centering 
\begin{tikzpicture}[>=stealth', initial text={},shorten >=1pt,auto,node distance=2.5cm]

\node[state,initial,initial text={$A_i$}](1)[label=center:{$s_i{=}1$}]{};
\node[state](2)[right of=1,label=center:{$2$}]{};
\node[state,draw=none](3)[right of=2,label=center:{$\ldots$}]{};
\node[state](4)[right of=3,label=center:{$n_i{-}1$}]{};
\node[state,accepting] (5)[right of=4,label=center:{$f_i{=}n_i$}]{};  

\draw[->]  (1) to node{$a_i,b$} (2);
\draw[->]  (2) to node{$a_i$}   (3);
\draw[->]  (3) to node{$a_i$}     (4);
\draw[->]  (4) to node{$a_i$}   (5);
\draw[->]  (5)[bend left] to node[above]{$a_i$}(1);

\draw[->](1)to[loop above]node {$\Sigma\setminus\{a_i,b\}$}(1);
\draw[->](2)to[loop above]node {$\Sigma\setminus\{a_i\}$}(2);
\draw[->](4)to[loop above]node {$\Sigma\setminus\{a_i\}$}(4);
\draw[->](5)to[loop above]node {$\Sigma\setminus\{a_i\}$}(5);
\end{tikzpicture}
\caption{The witness  DFA  $A_i$  
over the $(k+1)$-letter alphabet $\{b, a_1,a_2,\ldots,a_k\}$; $n_i\ge3$.} 
\label{fig:Ai}
\end{figure}
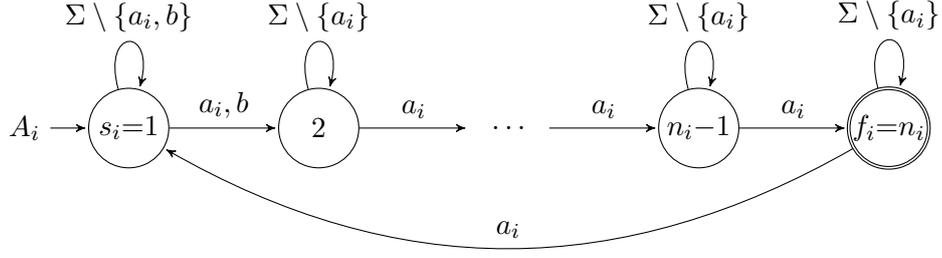 

\begin{figure}[h!]
	\centering 
	\begin{tikzpicture}[>=stealth', initial text={},shorten >=1pt,auto,node distance=2.1cm]

\node[state,initial,initial text={$N_i$}  ]  (q1)  [label=center:{$s_{i-1}$}]{};
\node[state](q2) [right of=q1,label=center:{$2$}]{};
\node[state,draw=none] (q3)[right of=q2,label=center:{$\ldots$}]{};
\node[state] (q4) [right of=q3,label=center:{\footnotesize{$n_{i-1}{-}1$}}]{}; 
\node[state] (q5)  [right of=q4,label=center:{$f_{i-1}$}]{}; 

\node[state]  (1) at (2,-4.6)  [ label=center:{$s_i$}]{};
\node[state](2) [right of=1,label=center:{$2$}]{};
\node[state,draw=none] (3)[right of=2,label=center:{$\ldots$}]{};
\node[state] (4) [right of=3,label=center:{$n_i{-}1$}]{}; 
\node[state,accepting ] (5)  [right of=4,label=center:{$f_i$}]{}; 

\draw[->]  (q1) to node{$a_{i-1},b$} (q2);
\draw[->]  (q2) to node{$a_{i-1}$} (q3);
\draw[->]  (q3) to node{$a_{i-1}$} (q4);
\draw[->]  (q4) to node{$a_{i-1}$} (q5);
\draw[->]  (q5) [bend left=20] to node[above]{$a_{i-1}$}(q1);

%\draw[->] (q1) [bend left] to node[above]{$b$}(q2);
\draw[->] (q1)to[loop above]node {$\Sigma\setminus\{a_{i-1},b\}$}(q1);
\draw[->] (q2)to[loop above]node {$\Sigma\setminus\{a_{i-1}\}$}(q2);
\draw[->] (q4)to[loop above]node {$\Sigma\setminus\{a_{i-1}\}$}(q4);
\draw[->] (q5)to[loop above]node {$\Sigma\setminus\{a_{i-1}\}$}(q5);

\draw[->]  (1) to node{$a_i,b$} (2);
\draw[->]  (2) to node{$a_i$} (3);
\draw[->]  (3) to node{$a_i$} (4);
\draw[->]  (4) to node{$a_i$} (5);
\draw[->]  (5) [bend left=20] to node[above]{$a_i,b$}(1);

%\draw[->]  (1) [bend left] to node[above]{$b$}(2);
\draw[->](1)to[loop above]node {$\Sigma\setminus\{a_i,b\}$}(1);
\draw[->](2)to[loop above]node {$\Sigma\setminus\{a_i\}$}(2);
\draw[->](4)to[loop above]node {$\Sigma\setminus\{a_i\}$}(4);
\draw[->](5)to[loop above]node {$\Sigma\setminus\{a_i\}$}(5);

\draw[->,red, dashed]  (q4)[bend right=10,above] to node{$a_{i-1}~~~~~$} (1);
\draw[->,red, dashed]  (q5.220) [bend right=19,pos=0.4]to node{$\Sigma\setminus\{a_i\}$} (1.40);
\end{tikzpicture}
	\caption{The NFA~$N_i$ recognizing
		the language~$L(A_{i-1})L(A-i)$.}
	\label{fig:NFANi}
\end{figure}
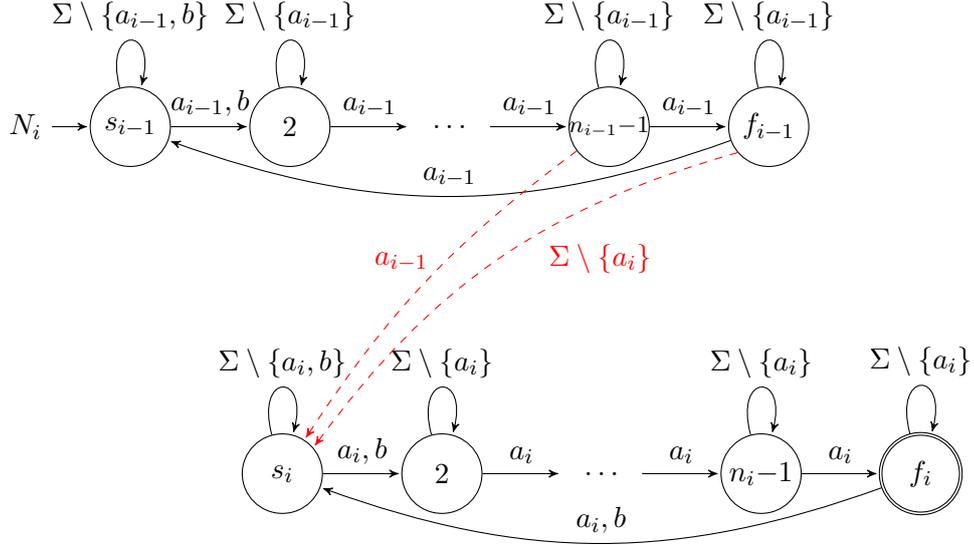

First, let us   consider the  
concatenation~$L(A_{i-1})L(A_{i})$ where $2\le i \le k$.
Construct an NFA~$N_i$ for this concatenation from DFAs $A_{i-1}$ and $A_{i}$
as shown in Figure~\ref{fig:NFANi},
that is,
by adding  the transitions~$(f_{i-1}{-}1 ,a_{i-1},s_{i})$
and~$(f_{i-1},\sigma,s_{i})$ with $\sigma\in \Sigma\setminus\{a_{i-1}\}$,
by making the state~$f_{i-1}$ non-final,
and the state~$s_i$ non-initial.

The next observation is crucial in what follows. It shows that in the  
subset automaton~$\cD(N_i)$,
each state~$( s_{i-1},S)$ with~$S\subseteq Q_i$ and~$S\ne\emp$
is reachable from~$( s_{i-1} ,\{s_i\})$.
Moreover, while reaching~$( s_{i-1},S)$ with~$f_{i}\notin S$,
the state~$f_{i}$ is never visited.
This is a very important property since, later, we do not wish to influence the~$(i+1)$st
component of a valid state while setting its~$i$th component.  

\begin{lemma}
\label{le:Ai-1Ai} 
Let~$2\le i \le k$ 
and~$N_i$ be the NFA for the language~$L(A_{i-1})L(A_i)$ 
described above.
For every non-empty subset~$S\subseteq Q_i$,
there exists a string~$w_S$ over
the alphabet~$\{a_{i-1}, a_i\}$
such that in the   subset automaton $\cD(N_i)$, we have
\begin{enumerate}
\item[(i)]  $( s_{i-1},\{s_i\}) \xrightarrow{w_S}( s_{i-1}, S)$;
\item[(ii)] if $f_i\notin S$,
 $u$ is a prefix of~$w_S$, and
 $( s_{i-1},\{s_i\}) \xrightarrow{u}( q,T)$,  
 then $f_i\notin T$.
\end{enumerate} 
\end{lemma}

\begin{proof} 
 The proof of both (i) and (ii) is by induction on~$|S|$.
 The basis, with~$|S|=1$,
 holds true since 
 for each~$j\in Q_i$\,, the state $(s_{i-1},\{j\})$ is reached from $(s_{i-1},\{s_i\})$
 by $a_i^{j-1}$. Moreover,   if~$j\ne f_i$, 
 then~$f_i$ is not visited while reading~$a_i^{j-1}$.
 
 Let~$|S|\ge2$. Let~$m=\min S$  
 and~$S'=a_i^{m-1} (S\setminus\{m\})$. Then~$|S'|=|S|-1$. By reading~$n_{i-1}$ times the symbol~$a_{i-1}$ and then the string~$a_i^{m-1}$ we get
 \[
 (s_{i-1}, S')
 \xrightarrow{a_{i-1}^{n_{i-1}}}(s_{i-1},\{s_i\}\cup 
S')
 \xrightarrow{a_i^{m-1}}(s_{i-1},\{m\}\cup(S\setminus\{m\}))
 =(s_{i-1},S),
 \]
 where the leftmost state is reached from the state~$(s_{i-1},\{s_i\})$
 by the string~$w_{S'}$ by induction,
 so we have~$w_S=w_{S'}a_{i-1}^{n_{i-1}}a_i^{m-1}$.
 Moreover, if~$f_i\notin S$,
 then~$S'\subseteq[2,f_i-m]$, so~$f_i\notin S'$.
 By~induction, the state~$f_i$ has not been visited
 while reading~$w_{S'}$ to reach~$(s_{i-1},S')$ 
 from~$(s_{i-1},\{s_i\})$. 
 Since in~$A_i$, the symbols~$a_{i-1}$ and~$a_i$ 
 perform the identity and circular shift, respectively,
 the state~$f_i$ is not visited either
 while reading the string~$a_{i-1}^{n_{i-1}}a_i^{m-1}$
  to reach~$(s_{i-1},S)$ from~$(s_{i-1},S')$. 
\end{proof} 

Now, construct the NFA~$N$
recognizing~$L(A_1)L(A_2)\cdots L(A_{k})$ from DFAs~$A_1, A_2, \ldots,A_k$
as follows:
First, for each~$i=1,2,\ldots,k-1$,
add the transitions~$(f_{i}{-}1,a_{i},s_{i+1})$
and~$(f_{i},\sigma,s_{i+1})$ with~$\sigma\in\Sigma\setminus\{a_{i}\}$.
Then, make states~$f_1,f_2,\ldots,f_{k-1}$
non-final, and states~$s_2,s_3,\ldots,s_{k}$
non-initial;
see Figure~\ref{fig:nfa3} for an illustration.

\begin{figure}[h!]
\centering
\begin{tikzpicture}[>=stealth', initial text={~},shorten >=1pt,auto,node distance=2cm]

\node[state,initial,initial text={$N$}  ]  (0)  [label=center:{$s_1  $}]{};
\node[state  ] (1)  [right of=0,label=center:{$ 2$}]{};
\node[state ] (3)  [right of=1,label=center:{$f_1$}]{};

\node[state] (s2) at (2,-3)[label=center:{$s_2$}]{};
\node[state] (22)  [right of=s2,label=center:{$2$}]{};
\node[state] (23)  [right of=22,label=center:{$3$}]{};
\node[state] (f2)  [right of=23,label=center:{$f_2$}]{};

\node[state]  (s3)  at (6,-6)[label=center:{$s_3$}]{};
\node[state] (32)  [right of=s3,label=center:{$2$}]{};
\node[state,accepting] (f3)  [right of=32,label=center:{$f_3$}]{};

\draw[->]  (s2) to node{$a_2,b$} (22);
\draw[->]  (22) to node{$a_2$} (23);
\draw[->]  (23) to node{$a_2$} (f2);
\draw[->] [bend left] (f2) to node[above]  {$a_2$}(s2);

\draw[->](22)to[loop above]node {$a_1,a_3,b$}(22);
\draw[->](23)to[loop above]node {$a_1,a_3,b$}(23);
\draw[->](f2)to[loop above]node {$a_1,a_3,b$}(f2);

\draw[->,red, dashed]  (3) to node [pos=0.1]{$a_2,a_3,b$} (s2.45);
\draw[->](s2)to[loop above]node {$a_1,a_3$}(s2);

\draw[->,red, dashed] [bend left] (1) to node[pos=0.4] {$a_1$}(s2);

\draw[->](s3)to[loop above]node {$a_1,a_2$}(s3);
\draw[->](32)to[loop above]node {$a_1,a_2,b$}(32);
\draw[->](f3)to[loop above]node {$a_1,a_2,b$}(f3);

\draw[->]  (s3) to node{$a_3,b$} (32);
\draw[->]  (32) to node{$a_3$} (f3);
\draw[->]  (f3) [bend left=40,above]to node  {$a_3$}(s3);

\draw[->,red, dashed]  (f2) to node [pos=0.1]{$a_1,a_3,b$} (s3.45);
\draw[->,red, dashed]   (23)[bend left] to node[pos=0.4] {$a_2$}(s3);

\draw[->]  (0) to node{$a_1,b$} (1);
\draw[->]  (1) to node{$a_1$} (3);
\draw[->] [bend left=40] (3) to node[above]{$a_1$}(0);

\draw[->](0)to[loop above]node {$a_2,a_3$}(0);
\draw[->](1)to[loop above]node {$a_2,a_3,b$}(1);
\draw[->](3)to[loop above]node {$a_2,a_3,b$}(3);
\end{tikzpicture}
\caption{The NFA~$N$ for~$L(A_1)L(A_2)L(A_3)$
with~$n_1=3$, $n_2=4$, and~$n_3=3$.}
\label{fig:nfa3}
\end{figure}
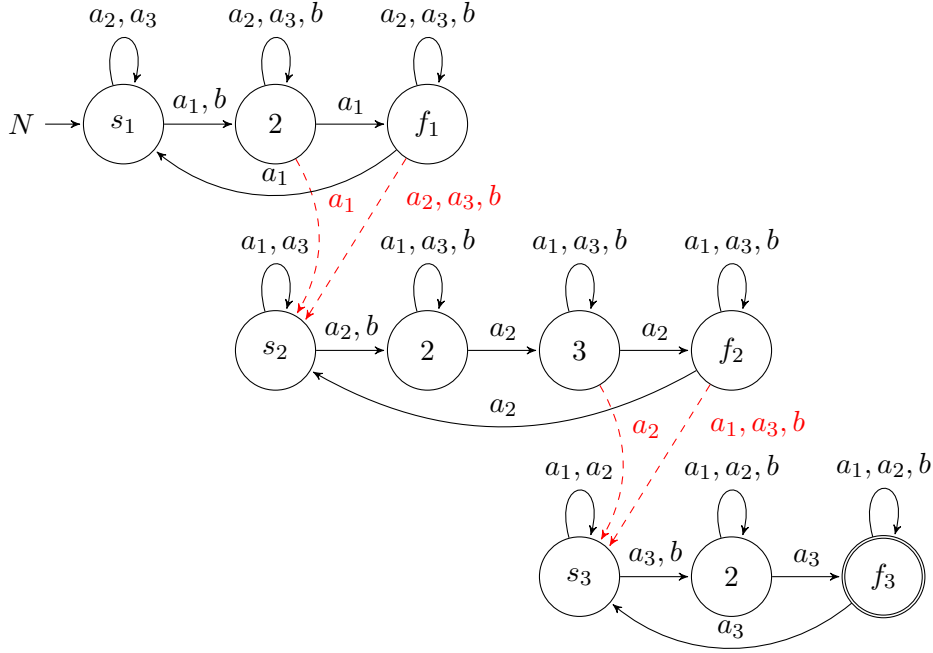
 
\begin{theorem}
\label{thm:k+1Reach}
\label{thm:k+1}
 Let~$k\ge2$ and~$n_i\ge3$ for~$i=1,2\ldots,k$.
 Let~$A_i$ be the~$n_i$-state DFA from Figure~\ref{fig:Ai}.
 Let~$N$ be the NFA for~$L(A_1)L(A_2)\cdots L(A_k)$
 described above.  
 Then  all valid states 
 are reachable and pairwise distinguishable
 in the subset automaton  $\cD(N)$.  
\end{theorem}

\begin{proof}
 We first prove reachability.
 Let~$q=(j,S_2,S_3,\ldots,S_k)$
 be a valid state.
 If~$S_2=\emp$,
 then the state~$q=(j,\emp,\emp,\ldots,\emp)$
 is reached from the initial state~$(s_1,\emp,\emp,\ldots,\emp)$
 by the string~$a_1^{j-1}$.
 Next, let~$\ell=\max\{i\ge2\mid S_i\ne \emp\}$.
 Then~$q=$~$(j,S_2,S_3,\ldots,S_\ell, \emp,\emp,\ldots,\emp)$
 where 
 $2\le \ell \le k$, \ 
 $S_i\subseteq Q_i$  and~$S_i\ne\emp$ for~$i= 2,3,\ldots,\ell$.  
  Since each~$a_i$ performs the circular shift
  in~$A_i$ and the identity in~$A_j$ with~$j\ne i$, 
  the string~$a_1^{n_1} a_2^{n_2} \cdots a_{\ell-1}^{n_{\ell-1}}$
  sends the initial state~$(s_1,\emp,\emp,\ldots,\emp)$
  to 
    \[
  (s_1, \{s_2\}, \{s_3\},\ldots, \{s_{\ell-1}\},\{s_\ell\},
  \emptyset, \emptyset,\ldots,\emptyset).
  \]
Now, we  
set the corresponding
components to sets~$S_i$, 
starting with~$S_\ell$,
continuing with~$S_{\ell-1}, \ldots,$ and ending with~$S_3$ and~$S_2$.
By Lemma~\ref{le:Ai-1Ai} applied to the NFA~$N_\ell$
recognizing~$L(A_{\ell-1})L(A_\ell)$, there is a string~$w_{S_\ell}$
 over~$\{a_{\ell-1},a_\ell\}$ 
 which sends~$(s_{\ell-1},\{s_\ell\})$
 to~$(s_{\ell-1},S_\ell)$ in the subset automaton~$\cD(N_\ell)$.
 Moreover, since~$q$ is valid,
 we have~$f_\ell \notin S_\ell$,
 which means that the state~$f_\ell$ is not
 visited while reading~$w_{S_\ell}$.
 Since both~$a_{\ell-1}$ and~$a_\ell$
 perform identities on~$Q_1,Q_2,\ldots,Q_{\ell-2}$,
  in 
  $\cD(N)$ we have 
 \begin{align*}
     &(s_1, \{s_2\},  \{s_3\},\ldots, \{s_{\ell-1}\},\{s_\ell\},
  \emptyset, \emptyset,\ldots,\emptyset) 
\xrightarrow{w_{S_\ell} \text{ over }\{a_{\ell-1},a_\ell\}}   \\
  &  (s_1, \{s_2\},  \{s_3\}, \ldots, \{s_{\ell-1}\}, ~~S_\ell~,
  \emptyset, \emptyset,\ldots,\emptyset).
 \end{align*}
 
 Next, Lemma~\ref{le:Ai-1Ai} applied to~$N_{\ell-1}$ gives a string~$w_{S_{\ell-1}}$ over~$\{a_{\ell-2},a_{\ell-1}\}$
 which sends~$(s_{\ell-2},\{s_{\ell-1}\})$
 to~$(s_{\ell-2},S_{\ell-1})$
 in~$\cD(N_{\ell-1})$,
 and moreover if~$f_{\ell-1}\notin S_{\ell-1}$,
 then~$f_{\ell-1}$ is not visited
 while reading this string.
 Since both symbols~$a_{\ell-2}$ and~$a_{\ell-1}$ perform identities on~$Q_1,Q_2,\ldots,Q_{\ell-3}$,
 as well as on~$Q_\ell$,
 in~$\cD(N)$ 
 we have
  \begin{align*}
     &(s_1, \{s_2\}, \{s_3\},\ldots, \{s_{\ell-2}\},\{s_{\ell-1}\},S_\ell,
  \emptyset, \emptyset,\ldots,\emptyset)
\xrightarrow{w_{S_{\ell-1}} \text{ over }\{a_{\ell-2},a_{\ell-1}\}} \\
    &(s_1, \{s_2\}, \{s_3\},\ldots, \{s_{\ell-2}\},~~S_{\ell-1}, ~S_\ell,
  \emptyset, \emptyset,\ldots,\emptyset).
 \end{align*}

Now, for~$i=2,3,\ldots,\ell-2$,
  let $w_{S_i}$  be the string   
  over $\{a_{i-1}, a_i\}$ given by  Lemma~\ref{le:Ai-1Ai}
  that sends~$( s_{i-1},\{s_i\})$
  to~$(s_{i-1},S_i)$.
  Moreover, $f_i\notin S_i$ implies
  that~$f_i$ is never visited while
  reading~$w_{S_i}$,
  which in turn implies that 
  $ s_{i+1}$ is never added to the~$(i+1)$th component in such a  case.
  If $f_i\in S_i$ and~$i\le k-1$, then the state $s_{i+1}$ is included in $S_{i+1}$
  since the state~$q$  is  valid,
  and~$s_{i+1}$ is sent to itself by both~$a_{i-1}$ and~$a_{i}$.
  Next, there is a loop on both~$a_{i-1}$ and~$a_i$
  in the states~$s_1,s_2,\ldots,s_{i-2}$,
  as well as in all states of automata~$A_{i+1}, A_{i+2},\ldots, A_\ell$. 
  Then in~$\cD(N)$   we have
  \begin{align*}
   & (s_1, \{s_2\}, \{s_3\},\ldots,
 \{s_{\ell-2}\},~~S_{\ell-1} ,~~~S_\ell,~
  \emptyset, \emptyset,\ldots,\emptyset) 
\xrightarrow{w_{S_{\ell-2}}w_{S_{\ell-3}}\cdots\, w_{S_3}w_{S_2}}\\
     & (s_1, ~~S_2, ~~~S_3,\ldots,
 ~~S_{\ell-2},~~~S_{\ell-1} ,~~~S_\ell,~
  \emptyset, \emptyset,\ldots,\emptyset),
  \end{align*} 
  and the resulting state is sent to~$q$ by the string~$a_1^{j-1}$.
    Hence~$q=(j,S_2, S_3, \ldots,S_\ell, \emp,\emp,\ldots,\emp)$
  is reached from the initial 
  state~$(s_1,\emp,\emp,\ldots,\emp)$ by the string
  $
      a_1^{n_1} a_2^{n_2} \cdots a_{\ell-1}^{n_{\ell-1}}
       w_{S_\ell} w_{S_{\ell-1}} \cdots  w_{S_3} w_{S_2} a_1^{j-1}       
  $.

\medskip 
 To get distinguishability,
 let us show that each singleton set is co-reachable in~$N$.
 First, for an example, consider the NFA from  Figure~\ref{fig:nfa3}.
 In its reversed automaton, 
 the initial set is~$\{f_3\}$, and we have  
\begin{align*}
  \{f_3\}\xrightarrow{a_3}\{2\}
 \xrightarrow{a_3}\{s_3\}
 \xrightarrow{b}\{f_2\}
 \xrightarrow{a_2}\{3\}
 \xrightarrow{a_2}\{2\}
 \xrightarrow{a_2}\{s_2\}
 \xrightarrow{b}\{f_1\}
\xrightarrow{a_1}\{2\}
 \xrightarrow{a_1}\{s_1\}.   
\end{align*} 
 In the general case, the initial set of~$N^R$   is~$\{f_k\}$.
 Next, for each~$i=1,2,\ldots,k$, 
 each singleton set~$\{j\}$ such that~$j\in Q_i$ 
 is reached from~$\{f_i\}$ via a string in~$a_i^*$.
 Finally, for each~$i=2,3,\ldots,k$,
 the singleton set~$\{{f_{i-1}}\}$   is reached from~$\{s_{i}\}$ by~$b$ 
 since~$n_{i-1}\ge3$.
 Thus, for every state~$q$ of~$N$,
 the singleton set~$\{q\}$ is co-reachable in the NFA~$N$.
 By Corollary~\ref{cor},
 all states of the subset automaton~$\cD(N)$
 are pairwise distinguishable.
 \end{proof}

Notice that all automata in the previous theorem,
as well as witness automata from~\cite{clp18},
are required to have at least three states.
We conclude this section by describing the witnesses
for multiple concatenation also in the case
where some of given automata have two states.
The idea is to use symbols~$a_k$ and~$b$ to guarantee
co-reachability of singleton sets.
However, then we should be careful with reachability.

To this aim, let~$k\ge2$,~$n_i\ge2$ for~$i=1,2,\ldots,k$,  
and~$\Sigma=\{b,a_1,a_2,\ldots,a_k\}$.
Let
\begin{align*}
    I&=\{i\mid 1\le i\le k-1 \text{ and } i\bmod 2= k\bmod 2\} \\
    J&=\{i\mid 1\le i\le k-1 \text{ and } i\bmod 2\ne k\bmod 2\}, 
\end{align*}
that is, the set~$I$ contains the indexes that have the same parity as~$k$, and the set~$J$ the others.

Consider the~$n_i$-state   DFAs $A_i=(Q_i,\Sigma,\cdot,s_i, \{f_i\})$,
see Figure~\ref{fig:dfaAk_dvojky}, 
where~$Q_i=\{1,2,\ldots,n_i\}$,
 ~$s_i=1$, ~$f_i=n_i$, 
and the transitions are as follows: \\
 if~$i\in I$,
        \,then~\,$a_i\colon (1,2,\ldots,n_i)$,  
        \,$a_k \colon (1\to 2 \to \cdots \to n_i)$,
        \ \,and~$\sigma\colon (1)$ if~$\sigma\in\Sigma\setminus\{a_i,a_k\}$,\\
 if~$i\in J$,
        then~$a_i\colon (1,2,\ldots,n_i)$,   
        $~b\,\, \colon (1\to 2 \to \cdots \to n_i)$,
        \,and~$\sigma\colon (1)$ if~$\sigma\in\Sigma\setminus\{a_i,b\}$,\\ 
 if~$i=k$,
        \,then \,$b\,\colon (1,2,\ldots,n_k)$,   
        $a_k \colon (1\to 2 \to \cdots \to n_k)$,
        \,and~$\sigma\colon (1)$ if~$\sigma\in\Sigma\setminus\{a_k,b\}$,  
that is,
\begin{itemize}
	\item 
	each~$a_i$ with~$1\le i \le k-1$
	performs the circular shift on~$Q_i$,
	and the identity on~$Q_j$ with~$j\ne i$;  
	\item 
	$a_k$ performs
	the transformation~$(1\to 2\to 3 \to \cdots \to n_i)$
	on~$Q_i$ with~$i\in I$ or~$i=k$,
	and the identity on~$Q_i$ with~$i\in J$,  
	\item 
	$b$ performs
	the transformation~$(1\to 2\to 3 \to \cdots \to n_i)$
	on~$Q_i$ with~$i\in J$,
	the circular shift on~$Q_k$,
	and the identity on~$Q_i$ with~$i\in I$.
\end{itemize}

\begin{figure}[h!]
\centering
\tikzset{every state/.style={minimum size=27}}
\begin{tikzpicture}[>=stealth', initial text={~},shorten >=1pt,auto,node distance=2.1cm]

\node[state,initial,initial text={$A_i\ (i\in I)$}  ]  (si)  [label=center:{$s_i{=}1 $}]{};
\node[state  ] (2)  [right of=si,label=center:{$2$}]{};
%\node[state  ] (2)  [right of=1,label=center:{$2$}]{};
\node[state,draw=none  ] (3)  [right of=2,label=center:{$\ldots$}]{};
\node[state  ] (4)  [right of=3,label=center:{$ n_i{-}2 $}]{}; 
\node[state,accepting  ] (5)  [right of=4,label=center:{$f_i{=}n_i$}]{}; 

\draw[->]  (si) to node{$a_i,a_k$} (2);
\draw[->]  (2) to node{$a_i,a_k$} (3);
%\draw[->]  (2) to node{$a$} (3);
\draw[->]  (3) to node{$a_i,a_k$} (4);
\draw[->]  (4) to node{$a_i,a_k$} (5);
\draw[->] [bend left=25] (5) to node{$a_i$}(si);

\draw[->](si)to[loop above]node {$\Sigma\setminus{\{a_i,a_k}\}$}(si);
\draw[->](2)to[loop above]node {$\Sigma\setminus{\{a_i,a_k}\}$}(2);
\draw[->](4)to[loop above]node {$\Sigma\setminus{\{a_i,a_k}\}$}(4);
\draw[->](5)to[loop above]node {$\Sigma\setminus{\{a_i,a_k}\}$}(5);
\draw[->](5)to[loop right]node {$a_k$}(5);
\end{tikzpicture}

\vskip10pt
\begin{tikzpicture}[>=stealth', initial text={~},shorten >=1pt,auto,node distance=2.1cm]

\node[state,initial,initial text={$A_i\ (i\in J)$}  ]  (sj)  [label=center:{$s_i{=}1 $}]{};
\node[state  ] (2)  [right of=sj,label=center:{$2$}]{};
%\node[state  ] (2)  [right of=1,label=center:{$2$}]{};
\node[state,draw=none  ] (3)  [right of=2,label=center:{$\ldots$}]{};
\node[state  ] (4)  [right of=3,label=center:{$ n_i{-}2 $}]{}; 
\node[state,accepting  ] (5)  [right of=4,label=center:{$f_i{=}n_i$}]{}; 

\draw[->]  (sj) to node{$a_i,b$} (2);
\draw[->]  (2) to node{$a_i,b$} (3);
%\draw[->]  (2) to node{$a$} (3);
\draw[->]  (3) to node{$a_i,b$} (4);
\draw[->]  (4) to node{$a_i,b$} (5);
\draw[->] [bend left=25] (5) to node{$a_i$}(sj);

\draw[->](sj)to[loop above]node {$\Sigma\setminus{\{a_i,b}\}$}(sj);
\draw[->](2)to[loop above]node {$\Sigma\setminus{\{a_i,b}\}$}(2);
\draw[->](4)to[loop above]node {$\Sigma\setminus{\{a_i,b}\}$}(4);
\draw[->](5)to[loop above]node {$\Sigma\setminus{\{a_i,b}\}$}(5);
\draw[->](5)to[loop right]node {$b$}(5);
\end{tikzpicture}

\begin{tikzpicture}[>=stealth', initial text={~},shorten >=1pt,auto,node distance=2.1cm]

\node[state,initial,initial text={$~~~~~A_k~~~~~$}  ]  (sk)  [label=center:{$s_k{=}1 $}]{};
\node[state  ] (2)  [right of=sk,label=center:{$2$}]{};
%\node[state  ] (2)  [right of=1,label=center:{$2$}]{};
\node[state,draw=none  ] (3)  [right of=2,label=center:{$\ldots$}]{};
\node[state  ] (4)  [right of=3,label=center:{$ n_k{-}2 $}]{}; 
\node[state,accepting  ] (5)  [right of=4,label=center:{$f_k{=}n_k$}]{}; 

\draw[->]  (sk) to node{$a_k,b$} (2);
\draw[->]  (2) to node{$a_k,b$} (3);
%\draw[->]  (2) to node{$a$} (3);
\draw[->]  (3) to node{$a_k,b$} (4);
\draw[->]  (4) to node{$a_k,b$} (5);
\draw[->] [bend left=25] (5) to node{$b$}(sk);

\draw[->](sk)to[loop above]node {$\Sigma\setminus{\{a_k,b}\}$}(sk);
\draw[->](2)to[loop above]node {$\Sigma\setminus{\{a_k,b}\}$}(2);
\draw[->](4)to[loop above]node {$\Sigma\setminus{\{a_k,b}\}$}(4);
\draw[->](5)to[loop above]node {$\Sigma\setminus{\{a_k,b}\}$}(5);
\draw[->](5)to[loop right]node {$a_k$}(5);

% \draw[->](sk)to[loop above]node {$\Sigma\setminus\{b,a_k\}$}(sk);
% \draw[->](2)to[loop above]node {$\Sigma\setminus\{b,a_k\}$}(2);
% \draw[->](4)to[loop above]node {$\Sigma\setminus\{b,a_k\}$}(4);
% \draw[->](5)to[loop above]node {$\Sigma\setminus\{b\}$}(5);
%\draw[->] [bend right=25] (1.135) to node[above]{$b$}(0.45);

\end{tikzpicture}
\vskip-10pt
\caption{The DFAs  $A_i$ with $i\in I$ (top), 
     $A_i$ with~$i\in J$ (middle), and~$A_k$ (bottom).}
\label{fig:dfaAk_dvojky}
\end{figure}
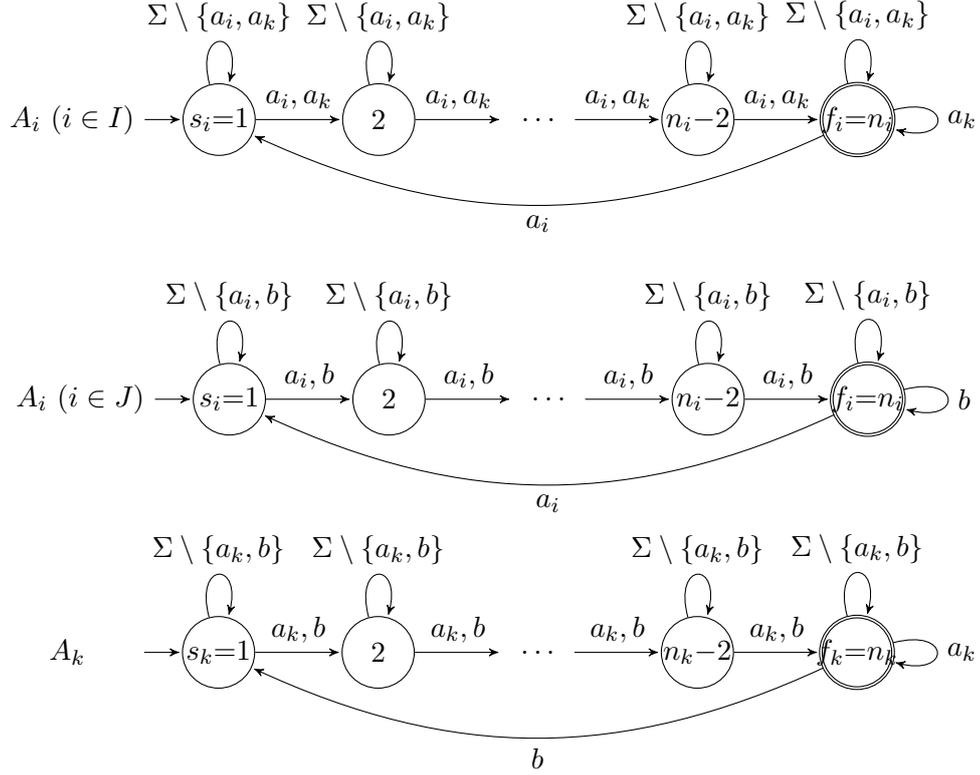 

Construct an NFA~$N$ for the language~$L(A_1)L(A_2)\cdots L(A_k)$ from the DFAs~$A_1, A_2, \ldots,A_k$ as follows
(see Figure~\ref{fig:nfaN23} for an illustration):
For~$i=1,2,\ldots,k-1$, add
the transitions~$(f_i{-}1,a_i,s_{i+1})$
and~$(f_i,\sigma,s_{i+1})$ for each~$\sigma\in\Sigma\setminus\{a_i\}$,
and moreover, if~$i\in I$, 
then add the transition~$(f_i{-}1,a_k,s_{i+1})$,
and if~$i\in J$, then add the transition~$(f_i{-}1,b,s_{i+1})$.
The initial state of~$N$ is~$s_1$,
and its unique final state is~$f_k$.

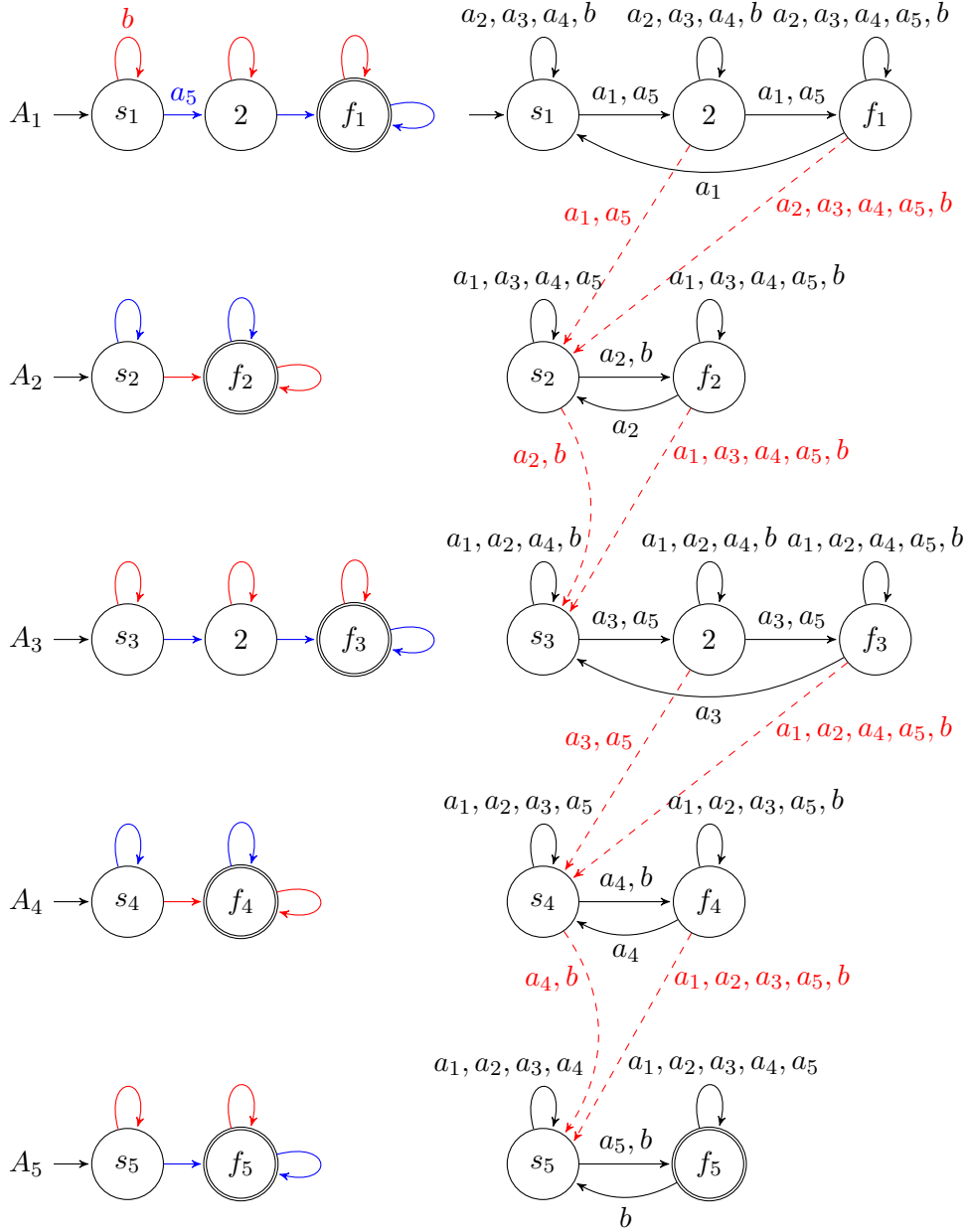
\begin{figure}[h!]
\centering
\tikzset{every state/.style={minimum size=27}}
\begin{tikzpicture}[>=stealth', initial text={~},shorten >=1pt,auto,node distance=2.2cm]

	\begin{scope}[node distance=1.5cm]
\node[state,initial,initial text={$A_1$}]  (q0)   [label=center:{$s_1  $}]{};
\node[state  ] (q1)  [right of=q0,label=center:{$ 2$}]{};
\node[state,accepting ] (q3)  [right of=q1,label=center:{$f_1$}]{};

\node[state,initial,initial text={$A_2$}] (qs2) at (0,-3.5)[label=center:{$s_2$}]{};
\node[state,accepting] (qf2)  [right of=qs2,label=center:{$f_2$}]{};

\node[state,initial,initial text={$A_3$}] (qs3) at (0,-7)[label=center:{$s_3$}]{};
 \node[state] (q32)  [right of=qs3,label=center:{$2$}]{};
% \node[state] (23)  [right of=22,label=center:{$3$}]{};
\node[state,accepting] (qf3)  [right of=q32,label=center:{$f_3$}]{};

\node[state,initial,initial text={$A_4$}] (qs4) at (0,-10.5)[label=center:{$s_4$}]{};
\node[state,accepting] (qf4)  [right of=qs4,label=center:{$f_4$}]{};

\node[state,initial,initial text={$A_5$}] (qs5) at (0,-14)[label=center:{$s_5$}]{};
\node[state,accepting] (qf5)  [right of=qs5,label=center:{$f_5$}]{};

\draw[->,blue]  (q0) to node{$a_5$} (q1);
\draw[->,blue]  (q1) to node{} (q3);
 \draw[->,blue]  (q3) to [loop right] node{} (q3);
\draw[->,red]  (q1) to [loop above] node{} (q1);
\draw[->,red]  (q0) to [loop above] node{$b$} (q0);
\draw[->,red]  (q3) to [loop above] node{} (q3);

\draw[->,red]  (qs2) to node{} (qf2);
 \draw[->,blue]  (qs2) to [loop above] node{} (qs2);
\draw[->,blue]  (qf2) to [loop above] node{} (qf2);
\draw[->,red]  (qf2) to [loop right] node{} (qf2);

\draw[->,blue]  (qs3) to node{} (q32);
\draw[->,blue]  (q32) to node{} (qf3);
\draw[->,blue]  (qf3) to [loop right] node{} (qf3);
\draw[->,red]  (qs3) to [loop above] node{} (qs3);
\draw[->,red]  (q32) to [loop above] node{} (q32);
\draw[->,red]  (qf3) to [loop above] node{} (qf3);

\draw[->,red]  (qs4) to node{} (qf4);
\draw[->,blue]  (qs4) to [loop above] node{} (qs4);
\draw[->,blue]  (qf4) to [loop above] node{} (qf4);
\draw[->,red]  (qf4) to [loop right] node{} (qf4);

\draw[->,blue]  (qs5) to node{} (qf5);
\draw[->,blue]  (qf5) to [loop right] node{} (qf5);
\draw[->,red]  (qs5) to [loop above] node{} (qs5);
\draw[->,red]  (qf5) to [loop above] node{} (qf5);
\end{scope}

\node[state,initial,initial text={}]  (0) at (5.5,0)   [label=center:{$s_1  $}]{};
\node[state  ] (1)    [right of=0, label=center:{$ 2$}]{};
\node[state ] (3)   [right of=1,label=center:{$f_1$}]{};

\node[state] (s2) at (5.5,-3.5)[label=center:{$s_2$}]{};
\node[state] (f2)  [right of=s2,label=center:{$f_2$}]{};

\node[state]  (s3)  at (5.5,-7)[label=center:{$s_3$}]{};
\node[state] (32)  [right of=s3,label=center:{$2$}]{};
\node[state] (f3)  [right of=32,label=center:{$f_3$}]{};

\node[state]  (s4)  at (5.5,-10.5)[label=center:{$s_4$}]{};
\node[state] (f4)  [right of=s4,label=center:{$f_4$}]{};

\node[state]  (s5)  at (5.5,-14)[label=center:{$s_5$}]{};
\node[state,accepting] (f5)  [right of=s5,label=center:{$f_5$}]{};

%,looseness=16

\draw[->]  (0) to node{$a_1,a_5$} (1);
\draw[->]  (1) to node{$a_1,a_5$} (3);
\draw[->] [bend left] (3) to node{$a_1$}(0);

\draw[->](0)to[loop above] %,looseness=16]
node {$a_2,a_3,a_4,b~~~$}(0);
\draw[->](1)to[loop above]node {$a_2,a_3,a_4,b~~~$}(1);
\draw[->](3)to[loop above]node {$a_2,a_3,a_4,a_5,b~~~$}(3);

\draw[->,red,dashed]  (3.220) to  node [pos=0.2]{\hglue-10pt$a_2,a_3,a_4,a_5,b$} (s2);
\draw[->,red,dashed]  (1) to node[above,pos=0.45] {$a_1,a_5~~~~~~~$}(s2);

\draw[->]  (s2) to node{$a_2,b$} (f2);
\draw[->] [bend left] (f2) to node   {$a_2$}(s2);
\draw[->](f2)to[loop above]node {\hglue25pt$~~~a_1,a_3,a_4,a_5,b$}(f2);
\draw[->](s2)to[loop above]node {$a_1,a_3,a_4,a_5~~~$}(s2);

\draw[->,red,dashed]  (f2) to node [pos=0.1]{\hglue-6pt $a_1,a_3,a_4,a_5,b$} (s3.45);
\draw[->,red,dashed]   (s2)[bend left] to node[pos=0.35,above] {$a_2,b~~~~~~~~~~$}(s3);

\draw[->](s3)to[loop above]node {$a_1,a_2,a_4,b~~~~~~$}(s3);
\draw[->](32)to[loop above]node {$a_1,a_2,a_4,b$}(32);
\draw[->](f3)to[loop above]node {$a_1,a_2,a_4,a_5,b$}(f3);

\draw[->]  (s3) to node{$a_3,a_5$} (32);
\draw[->]  (32) to node{$a_3,a_5$} (f3);
\draw[->]  (f3) [bend left]to node  {$a_3$}(s3);

\draw[->,red,dashed]  (f3.220) to node [pos=0.2]{\hglue-10pt$a_1,a_2,a_4,a_5,b$} (s4);
\draw[->,red,dashed]   (32)[above] to node[pos=0.45] {$a_3,a_5~~~~~~~$}(s4.55);

\draw[->]  (s4)   to node{$a_4,b$} (f4);
\draw[->]  (f4) [bend left] to node{$a_4$} (s4);
\draw[->](s4)to[loop above]node {$a_1,a_2,a_3,a_5~~~~~$}(s4);
\draw[->](f4)to[loop above]node {\hglue35pt $a_1,a_2,a_3,a_5,b$}(f4);

\draw[->,red,dashed]  (f4) to node [pos=0.1]{\hglue-6pt$a_1,a_2,a_3,a_5,b$} (s5.35);
\draw[->,red,dashed]   (s4)[bend left=35,above] to node[pos=0.35] {$a_4,b~~~~~~~~~$}(s5);

\draw[->]  (s5)  to node{$a_5,b$} (f5);
\draw[->]  (f5) [bend left] to node{$b$} (s5);
\draw[->](s5)to[loop above]node {$a_1,a_2,a_3,a_4~~~~~~~$}(s5);
\draw[->](f5)to[loop above]node {$~~~a_1,a_2,a_3,a_4,a_5$}(f5);

\end{tikzpicture}
\caption{The DFAs~$A_1,A_2,A_3,A_4,A_5$: transitions on~$a_5$ and~$b$ (left) and the NFA~$N$ for~$L(A_1)L(A_2)L(A_3)L(A_4)L(A_5)$ (right)
with~$n_1=n_3=3$ and $n_2=n_4=n_5=2$.}
\label{fig:nfaN23}
\end{figure}

\begin{theorem}
\label{thm:aj_dvojky}
    Let~$k\ge2$ and~$n_i\ge2$ for~$i=1,2,\ldots,k$.
    Let~$A_1,A_2,\ldots,A_k$
    be the DFAs shown in Figure~\ref{fig:dfaAk_dvojky},
    and~$N$ be the NFA for~$L(A_1)L(A_2)\cdots L(A_k)$
    described above.
    Then all valid states are reachable
    and pairwise distinguishable in~$\cD(N)$.
\end{theorem}

\begin{proof}
  First, notice that Lemma~\ref{le:Ai-1Ai}
  still holds for automata~$A_1,A_2,\ldots,A_{k-1}$
  since the transitions on~$a_1,a_2,\ldots,a_{k-1}$
  are the same. 
  Thus, for each non-empty subset~$S$ of~$Q_i$
  with~$i\le k-1$,
  let~$w_S$ be the string over~$\{a_{i-1},a_i\}$
  given By Lemma~\ref{le:Ai-1Ai}.

  Let~$(\{j\},S_2,S_3,\ldots,S_k)$
  be a valid state. 
  If~$S_k=\emp$,
  then~$(j,S_2,S_3,\ldots,S_{k-1},\emp)$
  is reachable as shown in the proof 
  of Theorem~\ref{thm:k+1}.

  Now, let~$S_k\ne\emp$.
  Then the state~$(s_1,\{s_2\},\{s_3\},\ldots,\{s_k\})$
  is reached from the initial state by~$a_1^{n_1} a_2^{n_2} \cdots a_{k-1}^{n_{k-1}}$.
  Next, notice that
  Lemma~\ref{le:Ai-1Ai}
  still holds for~$N_k$
  even if~$a_k$ fixes~$f_k$
  instead of sending it to~$s_k$
  since the out-transition in~$f_k$
  on~$a_k$ is not used in the proof of the lemma. Hence, there is a string~$w(S_k)$
  over~$\{a_{k-1},a_k\}$  which 
  sends the state~$(s_{k-1},\{s_k\})$
  to~$(s_{k-1},S_k)$  
  in the subset automaton~$\cD(N_{k})$.
    However, each~$a_k$ sends each state~$s_i$ with~$i\in I$
    to~$s_i+1$,  and we must then read the string~$u_i=(a_i)^{n_i-1}$
    to send~$s_i+1$ back to~$s_i$ while fixing
     the states in all the remaining components.
    Let~$u= \prod_{i\in I} u_i$.
    Now, let~$w'(S_k)$ be the string
    obtained from~$w(S_k)$ by inserting~$u$ after each~$a_k$.
    Since before reading each~$a_k$
    in~$w_{S_k}$ we have~$s_{k-1}$
    in the~$(k-1)$st component,
     the state~$(s_1,\{s_2\},\ldots,\{s_{k-1}\},\{s_k\})$
    is sent to~$(s_1,\{s_2\},\ldots,\{s_{k-1}\},S_k)$
    by~$w'_{S_k}$,
    and then to~$(j,S_2,S_3,\ldots,S_{k-1},S_k)$
    by~$w_{S_{k-1}}w_{S_{k-2}}\cdots w_{S_{3}}w_{S_{2}}a_1^{j-1}$.

  To prove distinguishability,
  let us show that all singleton sets
  are co-reachable in the NFA~$N$.
  First, as an example, consider the NFA~$N$
  from Figure~\ref{fig:nfaN23}. In the reversed automaton~$N^R$,
  we have
  \begin{align*}
       \{f_5\}
      \xrightarrow{b}\{s_5\}
    &  \xrightarrow{a_5}\{f_4\}
      \xrightarrow{a_4}\{s_4\}
      \xrightarrow{b}\{f_3\}
      \xrightarrow{a_3}\{2\}
      \xrightarrow{a_3}\{s_3\}
      \xrightarrow{a_5}\{f_2\}  
       \xrightarrow{a_2}\{s_2\}
       \xrightarrow{b}\{f_1\}
        \xrightarrow{a_1}\{2\}
         \xrightarrow{a_1}\{s_1\}.
  \end{align*}
 In the general case, the initial set of the reversed automaton~$N^R$ 
   is~$\{f_k\}$,
  and each set~$\{q\}$ with~$q\in Q_k$ is reached from~$\{f_k\}$
  by a string in~$b^*$.
  Next each~$\{f_i\}$ with~$i\in J$
  is reached from~$\{s_{i+1}\}$ by~$a_k$,
  while each~$\{f_i\}$ with~$i\in I$
  is reached from~$\{s_{i+1}\}$ by~$b$.
  Finally, each~$\{q\}$ with~$q\in Q_{i}$, where~$1\le i \le k-1$,
  is reached from~$\{f_i\}$
  by a string in~$a_i^*$.
  It follows that all singleton sets are co-reachable in~$N$. By Corollary~\ref{cor},
  all states of~$\cD(N)$ are 
  pairwise distinguishable.  
\end{proof}  
 
\section{Matching Lower Bound: \textbf{\textit{k}}-letter Alphabet}
\label{sec:k-symbols}

The aim of this section is to describe
witnesses for multiple concatenation
over a~$k$-letter alphabet.
Let us start with the following example.

\begin{example}\rm
\label{thm:two_nase}
\label{ex:dva_nase}
    Let~$n_1,n_2\ge1$ and~$A$ and~$B$
    be the binary DFAs shown in Figure~\ref{fig:nase_dva}.
    Let us show that the languages~$L(A)$ and~$L(B)$
    are witnesses for concatenation
    of two regular languages.

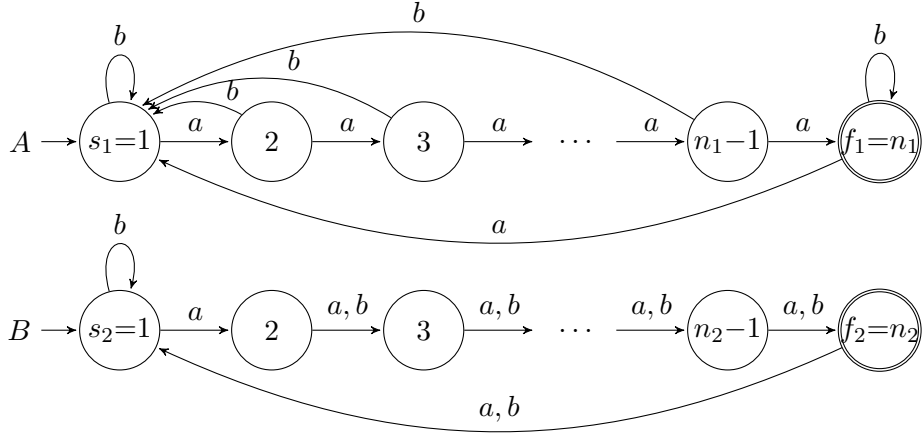
\begin{figure}[h!]
\centering
\tikzset{every state/.style={minimum size=30}}
\begin{tikzpicture}[>=stealth', initial text={~},shorten >=1pt,auto,node distance=2.0cm]

\node[state,initial,initial text={$A$}  ]  (0)  [label=center:{$s_1{=}1$}]{};
\node[state  ] (1)  [right of=0,label=center:{$2$}]{};
\node[state  ] (2)  [right of=1,label=center:{$3$}]{};
\node[state,draw=none  ] (3)  [right of=2,label=center:{$\ldots$}]{};
\node[state  ] (4)  [right of=3,label=center:{$n_1{-}1$}]{}; 
\node[state,accepting  ] (5)  [right of=4,label=center:{$f_1{=}n_1$}]{}; 

\draw[->]  (0) to node{$a$} (1);
\draw[->]  (1) to node{$a$} (2);
\draw[->]  (2) to node{$a$} (3);
\draw[->]  (3) to node{$a$} (4);
\draw[->]  (4) to node{$a$} (5);
\draw[->](5)to[bend left=25]node[above]{$a$}(0);

\draw[->](0)to[loop above]node{$b$}(0);
\draw[->](1)to[bend right,above,pos=0.1]node{$b$}(0.40);
\draw[->](2)to[bend right,above,pos=0.4]node{$b$}(0.50);
\draw[->](4)to[bend right,above]node{$b$}(0.60);
\draw[->](5)to[loop above]node{$b$}(5); 

\node[state,initial,initial text={$B$}] at (0,-2.5)  (q0) [label=center:{$s_2{=}1$}]{};
\node[state](q1)[right of=q0,label=center:{$2$}]{};
\node[state](q2)[right of=q1,label=center:{$3$}]{};
\node[state,draw=none](q3)[right of=q2,label=center:{$\ldots$}]{};
\node[state](q4)[right of=q3,label=center:{$n_2{-}1$}]{}; 
\node[state,accepting](q5)[right of=q4,label=center:{$f_2{=}n_2$}]{}; 

\draw[->]  (q0) to[loop above]node{$b$}(q0);
\draw[->]  (q0) to node{$a$} (q1);
\draw[->]  (q1) to node{$a,b$} (q2);
\draw[->]  (q2) to node{$a,b$} (q3);
\draw[->]  (q3) to node{$a,b$} (q4);
\draw[->]  (q4) to node{$a,b$} (q5); 
\draw[->]  (q5) to[bend left=25]node[above]{$a,b$}(q0);
\end{tikzpicture}
\caption{The binary witnesses for concatenation; $n_1,n_2\ge1$.}
\label{fig:nase_dva}
\end{figure}

    First, let~$n_2=1$. Then~$L(B)=\{a,b\}^*$
    and the concatenation~$L(A)\{a,b\}^*$
    is recognized by the minimal~$n_1$-state DFA
    obtained from~$A$ by replacing the transition~$(f_1,a,s_1)$
    with the transition~$(f_1,a,f_1)$.
    An upper bound is~$n_1$
    by Proposition~\ref{prop:jeden1-state}.

    Now, let~$n_1=1$ and~$n_2\ge2$.
    Then~$s_1=f_1$.
    Construct an NFA~$N$ for~$L(A)L(B)$
    from the DFAs~$A$ and~$B$
    by adding the transitions~$(f_1,a,s_2)$
    and~$(f_1,b,s_2)$,
    and by making the state~$s_1$ non-final.
    Let us show that all valid states~$(f_1,S)$
    are reachable in~$\cD(N)$.
    Since~$(f_1,S)$ is valid, we have~$s_2\in S$.
    The proof is by induction on~$|S|$.
    The basis,~$|S|=1$, that is, $S=\{s_2\}$,
    holds true since~$(f_1,\{s_2\})$
    is the initial state.
    Let~$|S|\ge2$ and $s_2\in S$.   
    Let~$m=\min(S\setminus\{s_2\})$
    and~$S'=S\setminus\{s_2,m\}$.
    Then~$ab^{m-2}(S')\subseteq[2,n_2-m+1]$ and
    \[
    (f_1,\{s_2\}\cup ab^{m-2}(S')) \xrightarrow{a}(f_1,\{s_2,2\}\cup b^{m-2}(S'))
    \xrightarrow{b^{m-2}}(f_1,\{s_2,m\} \cup S')=(f_1,S),
    \]
    where the leftmost valid state is reachable by induction.
    This proves the reachability of~$2^{n_2-1}$
    valid states.
    All these states are pairwise
    distinguishable by Lemma~\ref{le:dist}
    since all singletons~$\{q\}$, where~$q$
    is a state of~$B$, are co-reachable in~$N$. 
     By Proposition~\ref{prop:jeden1-state}, an upper bound is~$V_2=2^{n_2-1}$.
     
    Finally, let~$n_1,n_2\ge2$.
    Construct an NFA~$N$ for~$L(A)L(B)$
    from the DFAs~$A$ and~$B$
    by adding the transitions~$(f_1{-}1,a,s_2)$
    and~$(f_1,b,s_2)$, 
    by making the state~$f_1$ non-final
    and the state~$s_2$ non-initial.
    Let us show that in the subset automaton~$\cD(N)$,
    each valid state~$(j,S)$ is reachable.
    The proof is by induction on~$|S|$.
    The basis, with~$|S|=0$, holds true
    since each valid state~$(j,\emp)$
    is reached from
    the initial state is~$(s_1,\emp)$
    by~$a^{j-1}$. Let~$|S|\ge1$.
    There are three cases to consider.

    \medskip\noindent{\it Case~1:} $j=f_1$.
    Then~$s_2\in S$ since~$(f_1,S)$ is valid.
    We have
    \[ (f_1{-}1,a(S\setminus\{s_2\}))
    \xrightarrow{a}(f_1,\{s_2\})
    \cup(S\setminus\{s_2\})=(f_1,S)
    \]
    where the leftmost valid state is reachable by induction.

    \medskip\noindent{\it Case~2:} $j=s_1$.

    \smallskip\noindent{\it Case~2.a:} $2\in S$.
    Then~$s_2\in a(S)$ and~$(s_1,S)$ is reached
    from~$(f_1,a(S))$ by~$a$,
    where the latter valid state is considered in Case~1.       

    \smallskip\noindent{\it Case~2.b:} $2\notin S$
    and~$S=\{s_2\}$. Then we have
    $
    (f_1,\{s_2\})
    \xrightarrow{a}(s_1,\{2\})
    \xrightarrow{b^{n_2}}(s_1,\{s_2\}),
    $
    where the leftmost state is considered in Case~1.
    
    \smallskip\noindent{\it Case~2.c:} $2\notin S$
    and~$S\ne\{s_2\}$.
    Let~$m=\min(S\setminus\{s_2\})$
    and~$S'=S\cap\{s_2\}$.
    Then~$2\in b^{m-2}(S\setminus\{s_2\})$
    and~$(s_1,S)$ is reached 
    from~$(s_1,S'\cup b^{m-2}(S\setminus\{s_2\}))$
    by~$b^{m-2}$ where the latter state is considered in Case~2.a.
   
    \medskip\noindent{\it Case~3:} $2\le j\le n_1-1$.
    Then~$(j,S)$
    is reached from~$(s_1,a^{j-1}(S))$ by~$a^{j-1}$,
    and the latter set is considered in Case~2.

    \medskip
    This proves the reachability of~$(n_1-1)2^{n_2}+2^{n_2-1}$ states.
    To get distinguishability,
    let~$(i,S)$ and~$(j,T)$
    be two distinct valid states.
    There are two cases to consider.

    \medskip\noindent{\it Case~1:} $S\ne T$.
    The the two states are distinguishable by
    Lemma~\ref{le:dist} since
    all singletons~$\{q\}$,
    where~$q$ is a state of~$B$,
    are co-reachable in~$N$.

    \medskip\noindent{\it Case~2:} $S=T$ and~$i<j$.
    First, let~$S=\emp$. Since~$n_1\ge2$, the string~$a^{n_1-j}$
    sends the two states to states that
    differ in~$s_2$.
    The resulting states are distinguishable
    as shown in Case~1.
        Now, let~$S\ne\emp$. Then the two states
    are sent to~$(s_1,\{s_2\})$
    and~$(f_1,\{s_2\})$
    by~$a^{n_1-j}b^{n_2}$.
    Let us show that the resulting states
    are sent to states that differ in~$s_2$
    by~$a^{n_1}$ if~$s_2a^{n_1}\ne s_2$,
    and by~$a^{n_1-1}ba^{n_1-1}$ otherwise.

    First, notice that both strings~$a^{n_1}$ 
    and~$a^{n_1-1}ba^{n_1-1}$
    send the state~$f_1$ to itself in~$A$.
    It follows that~$(f_1,\{s_2\})$
    is sent to a state containing~$s_2$
    in its second component by both these strings.

    Now, let~$s_2a^{n_1}\ne s_2$.
    Then we have
    \[
    (s_1,\{s_2\})
    \xrightarrow{a^{n_1-1}}(f_1,\{s_2, s_2a^{n_1-1}\})
    \xrightarrow{a}(s_1,\{s_2a,s_2a^{n_1}\}),
    \]
    where~$s_2a\ne s_2$ since~$n_2\ge2$.
    Thus, in this case, the string~$a^{n_1}$
    sends the state~$(s_1,\{s_2\})$ to a state
    which does not have~$s_2$ in its second component.

    Finally, let~$s_2a^{n_1}= s_2$. Then~$s_2a^{n_1-1}=f_2$
    and since~$s_2b=f_2b=s_2$, we have
    \[
    (s_1,\{s_2\})
    \xrightarrow{a^{n_1-1}}(f_1,\{s_2, f_2\})
    \xrightarrow{b}(f_1,\{s_2\})
    \xrightarrow{a^{n_1-1}}(f_1{-}1,\{f_2\}),
    \]
    where~$f_2\ne s_2$
    since~$n_2\ge2$. Hence, this time
    the string~$a^{n_1-1}ba^{n_1-1}$
    sends~$(s_1,\{s_2\})$ to a state
    which does not contain~$s_2$
    in its second component.

     This proves distinguishability,
     and concludes our proof since by Theorem~\ref{thm:upper}, 
     a (known) upper bound   
     is~$(n_1-1)U_2+V_2=(n_1-1)2^{n_2}+2^{n_2-1}$
     in this case.
\qed
\end{example}

Hence the above example provides a two-letter witnesses
for the concatenation of two regular languages (even in the case
then automata may have one or two states).
Therefore, in what follows we assume that~$k\ge3$.

We use our previous results 
to describe witnesses for the concatenation of~$k$ languages
over the~$k$-letter alphabet $\{b,a_1,a_2,\ldots,a_{k-1}\}$.
The idea is as follows.
The transitions on input symbols $a_1,a_2,\ldots,a_{k-1}$ in automata~$A_1,A_2,\ldots,A_{k-1}$
are the same as in our $(k+1)$-letter witnesses from Theorem~\ref{thm:k+1},
while~$A_{k-1}$ and~$A_k$ over~$\{a_{k-1},b\}$ are the same as automata~$A$ and~$B$
in Example~\ref{thm:two_nase}.
The input symbol~$b$ performs 
the transformation~$(\{2,3,\ldots,n_i-1\}\to s_i)$
in each~$A_i$ except for~$A_k$, and it is used to get reachability as well as
distinguishability. 

To this aim,
let~$k\ge3$ and~$\Sigma=\{b,a_1,a_2, \ldots,a_{k-1}\}$
be a~$k$-letter alphabet.
Let~$n_1,n_k\ge2$ and~$n_i\ge3$ for~$i=2,3,\ldots,k-1$.
For~$i=1,2,\ldots,k$, define an~$n_i$-state 
DFA~$A_i=(Q_i,\Sigma,\cdot,s_i,\{f_i\})$,
see Figure~\ref{fig:nase_k}, where 
$Q_i=\{1,2,\ldots,n_i\}$,
$s_i=1$,
$f_i=n_i$,  and the transitions are as follows:
\begin{itemize}
    \item if~$i\le k-1$, then~\\
$a_i\colon (1,2, \ldots,n_i)$,
$b \colon (\{2,3,\ldots,n_i-1\}\to s_i)$,
and~$\sigma\colon (1)$ if~$\sigma\in\Sigma\setminus\{a_i,b\}$,
    \item if~$i=k$, then~\\
$a_{k-1}\colon (1,2, \ldots,n_k)$,
$b \colon (2\to 3\to\cdots\to n_k \to 1)$,
and~$\sigma\colon (1)$ if~$\sigma\in\Sigma\setminus\{a_{k-1},b\}$.
\end{itemize}

\begin{figure}[h!]
\centering
\tikzset{every state/.style={minimum size=30}}
\begin{tikzpicture}[>=stealth', initial text={~},shorten >=1pt,auto,node distance=2.1cm]

\node[state,initial,initial text={$\underset{(i<k)}{A_i}$}  ]  (0)  [label=center:{$s_i{=}1$}]{};
\node[state  ] (1)  [right of=0,label=center:{$2$}]{};
\node[state  ] (2)  [right of=1,label=center:{$3$}]{};
\node[state,draw=none  ] (3)  [right of=2,label=center:{$\ldots$}]{};
\node[state  ] (4)  [right of=3,label=center:{$n_i{-}1$}]{}; 
\node[state,accepting  ] (5)  [right of=4,label=center:{$f_i{=}n_i$}]{}; 

\draw[->]  (0) to node{$a_i$} (1);
\draw[->]  (1) to node{$a_i$} (2);
\draw[->]  (2) to node{$a_i$} (3);
\draw[->]  (3) to node{$a_i$} (4);
\draw[->]  (4) to node{$a_i$} (5);
\draw[->](5)to[bend left=22]node[above]{$a_i$}(0);

\draw[->](0)to[loop above]node{$b$}(0);
\draw[->](1)to[bend right,above,pos=0.1]node{$b$}(0.40);
\draw[->](2)to[bend right,above,pos=0.4]node{$b$}(0.50);
\draw[->](4)to[bend right,above]node{$b$}(0.60);
\draw[->](5)to[loop above]node{$b$}(5); 

\node[state,initial,initial text={$A_k$}] at (0,-2.5)  (q0) [label=center:{$s_k{=}1$}]{};
\node[state](q1)[right of=q0,label=center:{$2$}]{};
\node[state](q2)[right of=q1,label=center:{$3$}]{};
\node[state,draw=none](q3)[right of=q2,label=center:{$\ldots$}]{};
\node[state](q4)[right of=q3,label=center:{$n_k{-}1$}]{}; 
\node[state,accepting](q5)[right of=q4,label=center:{$f_k{=}n_k$}]{}; 

\draw[->]  (q0) to[loop above]node{$b$}(q0);
\draw[->]  (q0) to node{$a_{k-1}$} (q1);
\draw[->]  (q1) to node{$a_{k-1},b$} (q2);
\draw[->]  (q2) to node{$a_{k-1},b$} (q3);
\draw[->]  (q3) to node{$a_{k-1},b$} (q4);
\draw[->]  (q4) to node{$a_{k-1},b$} (q5); 
\draw[->]  (q5) to[bend left=22]node[above]{$a_{k-1},b$}(q0);
\end{tikzpicture}
\vskip-10pt
\caption{The DFA~$A_i$ with~$i<k$ (top):
transitions on~$a_i$ and~$b$,
and the DFA~$A_k$ (bottom): 
transitions on~$a_{k-1}$ and~$b$;
all the remaining symbols 
in both automata perform identities; $n_1,n_k\ge2$ and~$n_i\ge3$ for~$i=2,3,\ldots,k-1$.}
\label{fig:nase_k}
\end{figure}
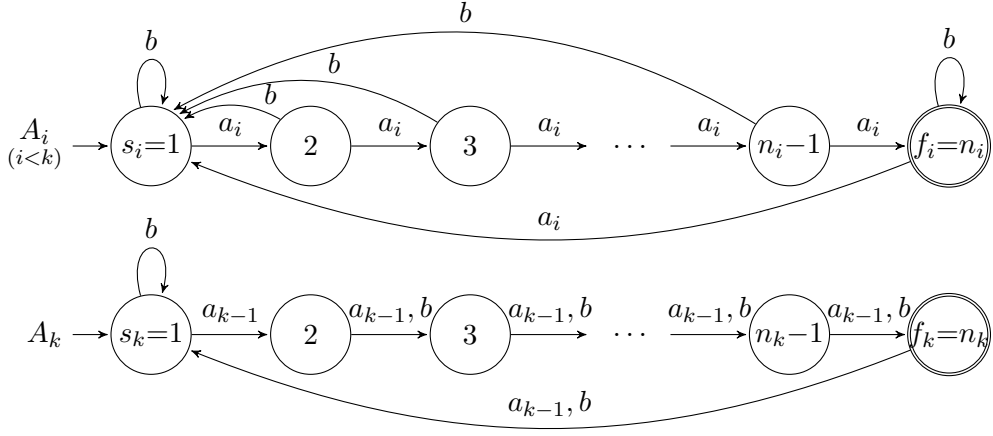

Construct an NFA $N$  
for $L(A_1)L(A_2)\cdots L(A_k)$
from DFAs $A_1,A_2,\ldots,A_k$
by adding the transitions
$(f_{i}{-}1,a_{i},s_{i+1})$,
$(f_{i}, a_j,s_{i+1})$ for $j\ne i$, and
$(f_{i},b,s_{i+1} )$
for~$i=1, 2,\ldots,k-1$;
the initial state of $N$ is $s_1$, and the  final state  is $f_k$.
The next theorem shows that all valid states are
reachable and pairwise distinguishable in~$\cD(N)$.
The proof of reachability   is based on our   results concerning $(k+1)$-letter witnesses
as well as our binary witnesses from Example~\ref{ex:dva_nase}.
The proof of distinguishability is not for free   this time.

\begin{theorem}
\label{thm:k}
 Let $k\ge3$, $n_1,n_k\ge2$, and~$n_i\ge3$ for $i=2,3,\ldots,k-1$.
 Let $A_1, A_2,\ldots,A_k$ 
 be DFAs shown in Figure~\ref{fig:nase_k}
 over the~$k$-letter alphabet~$\{b,a_1,a_2,\ldots,a_{k-1}\}$. Let~$N$ be the NFA for~$L(A_1)L(A_2)\cdots L(A_k)$ described above.
 Then all valid states are reachable and pairwise distinguishable in~$\cD(N)$.
\end{theorem}

\begin{proof}
 Consider a valid state $q=(j,S_2,\ldots,S_{k-1},S_k)$.
 First, let $S_k=\emp$.
 Since the transitions on~$a_1, a_2, \ldots, a_{k-1}$
 in $A_1,A_2,\ldots,A_{k-1}$
 are the same as in automata in Theorem~\ref{thm:k+1},
  the valid state~$(j,S_2,\ldots,S_{k-1},\emp)$  
 is reachable exactly the same way as in the proof of 
 this theorem. 
 
 Now let $S_k\ne\emp$.
 Notice that the transitions on~$a_{k-1}$ and~$b$
 in~DFAs~$A_{k-1}$ and~$A_k$ are the same as those on~$a$ and~$b$
 in DFAs~$A$ and~$B$ in Example~\ref{ex:dva_nase}.
  As shown in this example, % we have shown 
 for each $S\subseteq Q_k$, 
 there is a string  $w_{S}$ over $\{a_{k-1},b\}$
 which sends~$(s_{k-1},\emp)$
 to~$(s_{k-1},S)$
 in the subset automaton for $L(A_{k-1})L(A_k)$.
 Since we have 
 a loop on both $a_{k-1}$ and $b$ in 
 all states~$s_1,s_2,\ldots,s_{k-2}$,
 we reach~$(s_1,\{s_2\}, \{s_3\},\ldots,\{s_{k-2}\},\{s_{k-1}\},S)$ from the initial state by~$a_1^{n_1}a_2^{n_2}\cdots a_{k-2}^{n_{k-2}}w_{S}$.
Next, let~$w_{S_{i}}$
be the string over~$\{a_{i-2},a_{i-1}\}$
given by Lemma~\ref{le:Ai-1Ai}
which sends~$(s_{i-1},\{s_{i}\})$
to~$(s_{i-1},S_{i})$. 
Recall that~$f_{i}\notin S_{i}$ implies
 that the state $f_{i}$ is not visited while reading $w_{S_i}$.
 Moreover, a closer look at the proof of the lemma shows that if~$f_{i}\in S_{i}$
  then~$f_{i}$ 
 is visited for the first time immediately
 after reading the last~$a_i$ in~$w_{S_i}$.
 Now, let~$m$ be the number of occurrences
 of the symbol~$a_{k-1}$ in the string~$w_{S_{k-1}}$.
 Then the state~$(s_1,\{s_2\}, \{s_3\},\ldots,\{s_{k-2}\},\{s_{k-1}\},a^m_{k-1}(S_k))$
 is reachable as shown above,
 and it is sent to~$(s_1,\{s_2\}, \{s_3\},\ldots,\{s_{k-2}\},S_{k-1},S_k)$ by~$w_{S_{k-1}}$.
 The resulting state is sent to~$q$
 by the string~$w_{S_{k-2}}w_{S_{k-3}}\cdots w_{S_3}w_{S_2}a_1^{j-1}$.
 
 To get distinguishability, let~$p=(S_1,S_2, S_3,\ldots,S_k)$
 and~$q=(T_1,T_2,T_3,\ldots,T_k)$
 be two distinct valid states.
 If~$S_k\ne T_k$,
 then~$p$ and~$q$ are distinguishable
 by Lemma~\ref{le:dist} since each singleton subset
 of~$Q_k$ is co-reachable in~$N$ 
 via a string in~$a_{k-1}^*$.

 Let~$S_i\ne T_i$ for some~$i$ with~$1\le i \le k-1$,
 and~$S_j=T_j$ for~$j=i+1,i+2,\ldots,k$.
 Let us show that there is a string
 that sends~$p$ and~$q$ to two states
 which differ in~$s_{i+1}$.
 
 Without loss of generality, we have~$s\in S_i\setminus T_i$.
 First, we read the string~$w=a_i^{f_i-s}$ 
 which sends~$s$ to~$f_i$ in~$A_i$ and fixes all states in all~$A_j$ with~$j\ne i$
 to get states
 \begin{align*}
  &(S'_1,S_2',S_3',\ldots,S'_{i-1},S'\cup (S_i\cdot w),S'_{i+1},\ldots,S'_k) \\  
 &(T'_1,T_2',T_3',\ldots,T'_{i-1},\,T'\cup (T_i\cdot w),T'_{i+1},\ldots,T'_k)
 \end{align*}
 where~$S',T'\subseteq[1,f_i-s]$ and~$f_i\in (S_i\cdot w)\setminus (T_i\cdot w)$,
 that is, the~$i$th components of the resulting states
 differ in the state~$f_i$. 
 If~$S'_{i+1}\ne T'_{i+1}$,
 then we have the desired result.
 Otherwise, since~$s_{i+1}\in S'_{i+1}$, 
 both~$S'_{i+1}$ and~$T'_{i+1}$ are non-empty,
 which means that all~$S'_1,S'_2,\ldots,S'_{i}$
 and all~$T'_1,T'_2,\ldots,T'_i$ are non-empty.
 Now, 
 the string~$b$ sends all states 
 of~$Q_j$ with~$2\le j\le k-1$,
 either to~$s_j$ or to~$f_j$,
 and then~$a_jb$ sends~$f_j$ to~$s_j$ and~$s_j$ to itself %if~$2\le j\le k-1$
 since~$n_j\ge3$. 
 Thus after reading 
 the string~$b (a_2b) (a_3b) \cdots (a_{i-1}b)$ and 
 if~$T'_1=\{f_1\}$, then also~$(a_1b)$, 
 we get states
 \begin{align*}
    &(~\{q\},\{s_2\},\{s_3\},\ldots,\{s_{i-1}\},S''\cup\{f_i\},S''_{i+1},\ldots,S''_k) \\
    &(\{s_1\},\{s_2\},\{s_3\},\ldots,\{s_{i-1}\}, ~~~\{s_i\}~~~~,
    T''_{i+1},\ldots,T''_k)  
 \end{align*}
 where~$q\in\{s_1,f_1\}$, ~$S''\subseteq\{s_i\}$,
 and~$S''_j,T''_j\subseteq\{s_j,f_j\}$ for~$j=i+1,i+2,\ldots,k-1$.
 There are two cases to consider.

  \medskip\noindent{\it Case~1:} $1\le i \le k-2$.
  Then~$2 \le i+1\le k-1$ and~$n_{i+1}\ge3$ which means that
  the string~$a_{i+1}b$ sends both~$f_{i+1}$ and~$s_{i+1}$
  to~$s_{i+1}$. 
  Thus after reading~$a_{i+1}b$, we get states
  \begin{align*}
    &(~\{q\},\{s_2\},\{s_3\},\ldots,\{s_{i-1}\},S''\cup\{f_i\},
    \{s_{i+1}\},S'''_{i+2},\ldots,S'''_k) \\
    &(\{s_1\},\{s_2\},\{s_3\},\ldots,\{s_{i-1}\}, ~~~\{s_i\}~~~~,
    \{s_{i+1}\},T'''_{i+2},\ldots,T'''_k). 
 \end{align*}
  Finally, the string~$a_{i+1}$, 
 which performs the identity on~$Q_j$ with~$j\ne i+1$
 and the circular shift on~$Q_{i+1}$,
 sends the resulting states to states
 which differ in~$s_{i+1}$.
 
 \medskip\noindent{\it Case~2:} 
 $i=k-1$. Then the string~$b^{n_k}$ sends all states of~$Q_k$ to~$s_k$,
 while it fixes~$s_j$ and~$f_{j}$  
 for~$j=1,2,\ldots,k-1$.
 Thus after reading~$b^{n_k}$ we get states
      $ (\{q\},\{s_2\}, \ldots,\{s_{k-2}\}, S''\cup \{f_{k-1}\},\{s_k\})$ and  
     $ (\{s_1\},\{s_2\}, \ldots,\{s_{k-2}\}, \{s_{k-1}\},\{s_k\})$.
 Now, in the same way as in Example~\ref{ex:dva_nase}
 we show that
 either~$a_{k-1}^{n_k}$
 or~$a_{k-1}^{n_k-1}ba_{k-1}^{n_k-1}$
 sends the resulting states
 to two states  
 which differ in~$s_k$.
\end{proof} 

Since the number of valid states provides an upper bound on the state complexity of multiple concatenation, we get our main result.

\begin{corollary}
    The DFAs~$A_1,A_2,\ldots,A_k$ shown in Figure~\ref{fig:nase_k}
    defined over a~$k$-letter alphabet
    are witnesses for multiple concatenation of~$k$
    languages.
    \qed
\end{corollary}

We conjecture that~$k$ symbols are necessary for describing witnesses
for concatenation of~$k$ languages. 
The next observation shows that our conjecture
holds for~$k=3$.

\begin{theorem}
	The ternary alphabet used 
	to describe witnesses for the concatenation of three languages
	in Theorem~\ref{thm:k} is optimal.
\end{theorem}

\begin{proof}
To get a contradiction, let $\Sigma=\{a,b\}$ and~$n_i\ge2$ for $i=1,2,3$. 
Let us consider binary 
DFAs~$A_i=(Q_i,\Sigma,\cdot,s_i,\{f_i\})$ 
where~$Q_i=\{1, 2,\ldots,n_i\}$,
~$s_i=1$,~$f_i\ne1$ for~$i=1,2,3$;
notice that to meet the upper bound 
for multiple concatenation,
each $A_1, A_2,\ldots,A_{k-1}$ must have one final state, and it must be different from the initial state.
	
	Construct the NFA~$N$ for~$L(A_1)L(A_2)L(A_3)$ 
    from DFAs~$A_1,A_2,A_3$
    as follows: for~$i=1,2$,
    each state~$q\in Q_i$ and each symbol~$\sigma\in\{a,b\}$ such that~$q\sigma=f_i$,
    add the transition~$(q,\sigma,s_{i+1})$;
    the initial state of~$N$ is~$s_1$ 
    and its unique final state is~$f_3$.
	Our aim is to show that either some valid state is unreachable in
	the subset automaton~$\cD(N)$ or some valid states are equivalent to each other.
	
Notice that to reach the valid state~$(s_1,Q_2,\{s_3\})$, 
we must have an input symbol that  performs a permutation on~$Q_2$,
and to reach the valid sta\-te~$(s_1,\{s_2\},Q_3)$,
we must have an input symbol that performs a permutation on~$Q_3$.
	
If both symbols perform  a  permutation  on~$Q_3$,
then the  sta\-tes~$(s_1,\{s_2\},Q_3)$ and~$(s_1,\{2\},Q_3)$ are equivalent
since all strings are accepted from both of  them.  
 
If both  symbols  perform a  permutation  on~$Q_2$,
then the    sta\-tes~$(s_1, Q_2,\{s_3\})$ and~$(2,Q_2,\{s_3\})$ are equivalent 
since if a string~$w$ is accepted by~$N$ from the state $s_1$ in $A_1$  
	through a computation 
	~$s_1 \xrightarrow{w'} s_2\xrightarrow{w''}f_3$ 
	with~$w=w'w''$,
	then it is   accepted through a computation~$w's_2\xrightarrow{w'}s_2 \xrightarrow{w''}f_3$ where~$w's_2\in Q_2$, so it is accepted from $(2,Q_2,\{s_3\})$;
	and vice versa. 

	Hence to meet the upper bound,
	we must have one permutation and one non-permutation input symbol
	 in both~$A_2,A_3$.
	
	Next, while reaching the valid  state~$(s_1,Q_2\setminus\{f_2\},\emp)$,
	we cannot visit    state~$f_2$. This means that
	there must be an input that maps~$Q_2\setminus\{f_2\}$ onto~$Q_2\setminus\{f_2\}$.
	Without loss of generality, let this input be~$a$.
	Since~$f_2$ must be reachable in~$A_2$,
	there must exist  a state ~$p$  in ~$Q_2\setminus\{f_2\}$ with~$p  b=f_2$.
	Moreover,~$f_2b\ne f_2$ because otherwise either~$f_2$ would have loops on both symbols,
	or both~$a$ and~$b$ would be non-permutation symbols in~$A_2$.
	We have two cases:
    
	(1) Let~$b$ be a non-permutation symbol in~$A_2$.
	Then~$a$ is a permutation on~$Q_2$, so~$f_2a=f_2$.
	This situation is depicted in   Fig.~\ref{fig:case1}.
	Moreover, there is a state in~$Q_2\setminus\{f_2\}$
	with no in-transition on~$b$.
	Therefore the valid state~$(s_1,Q_2\setminus\{f_2\},Q_3)$
	must be reached from some valid state on~$a$, and consequently~$a$ is a permutation on~$Q_3$.
	Next, since $f_2b\ne f_2$, the valid state~$(s_1,\{f_2b\}, Q_3)$  must be reached 
	from a valid state~$(j,\{f_2\}\cup S,Q_3)$ on~$b$ since
	to get $Q_3$ in the third component, we must visit $f_2$, and only reading~$b$
	eliminates the state~$f_2$. It follows that~$b$ is a permutation on~$Q_3$.
	Hence both~$a$ and~$b$ perform permutations on~$Q_3$, 
	thus resulting in  a contradiction.
	
	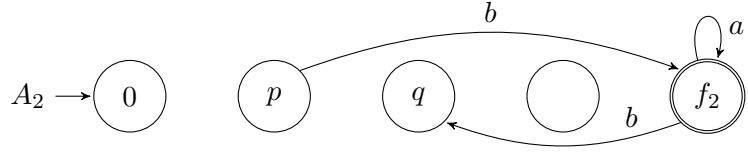
\begin{figure}[h!]
		\centering 
		\tikzset{every state/.style={minimum size=27}}
\begin{tikzpicture}[>=stealth', initial text={~},shorten >=1pt,auto,node distance=1.9cm]

\node[state,initial,initial text={$A_2$}  ]  (0)  [label=center:{$0$}]{};
\node[state  ] (1)  [right of=0,label=center:{$p$}]{};
%\node[state  ] (2)  [right of=1,label=center:{$2$}]{};
\node[state   ] (3)  [right of=1,label=center:{$q$}]{};
\node[state  ] (4)  [right of=3,label=center:{$  $}]{}; 
\node[state,accepting  ] (5)  [right of=4,label=center:{$f_2$}]{}; 

%\draw[->]  (0) to node{$a_i,b$} (1);
%\draw[->]  (1) to node{$a_i$} (3);
%\draw[->]  (2) to node{$a$} (3);
%\draw[->]  (3) to node{$a_i$} (4);
%\draw[->]  (4) to node{$a_i$} (5);
\draw[->] [bend left=20] (1.45) to node[above]{$b$}(5.135);
\draw[->] [bend left=20] (5.-135) to node[above,pos=0.2]{$b$}(3.-45);

%\draw[->](0)to[loop above]node {$\Sigma\setminus\{a_i,b\}$}(0);
%\draw[->](1)to[loop above]node {$\Sigma\setminus\{a_i\}$}(1);
%\draw[->](4)to[loop above]node {$\Sigma\setminus\{a_i\}$}(4);
\draw[->](5)to[loop above]node [right, near end] {$a$}(5);
%\draw[->] [bend right=25] (1.135) to node[above]{$b$}(0.45);

\end{tikzpicture}

%\begin{tikzpicture}[>=stealth', initial text={~},shorten >=1pt,auto,node distance=2cm]
%
%\node[state,initial] [label=center:{$0$}] (0){\phantom{0}};
%\node[state] [right of=0, label=center:{$1$}] (1) {\phantom{0}};
%%\node[state] [right of=1, label=center:{$2$}] (2){\phantom{0}};
%\node[state,draw=none] [right of=1] (x) {\dots};
%\node[state] [right of=x, label=center:{$n_i{-}2$}] (3){\phantom{0}};
%\node[state,accepting] [right of=3, label=center:{$n_i-1$}] (4){\phantom{0}};
%
%%%%%%%%%%%%%%%%%%%%%%%%%%%%%%%%
%
%\draw [->] (0) to  node{$a_i,b$} (1);
%\draw [->] (0) to [loop above]  node  %[pos=0.75,anchor=west]
%{$\Sigma\setminus\{a_i,b\}$ } (0);
%
%\draw [->] (1) to  node{$a_i$} (x);
%\draw [->] (1) to [loop above]node %[pos=0.75,anchor=west]
%{$\Sigma\setminus\{a_i\}$} (1);
%
%\draw [->] (x) to  node{$a_i$} (3);
%
%\draw [->] (3) to  node{$a_i$} (4);
%\draw [->] (3) to [loop above] node %[pos=0.75,anchor=west]
% {$\Sigma\setminus\{a_i\}$} (3);
%
%\draw [->] (4) to [loop above] node %[pos=0.75,anchor=west]
%{$\Sigma\setminus\{a_i\}$} (4);
%
%\draw [to path={
%.. controls +(-1,-1) and +(1,-1) .. (\tikztotarget) \tikztonodes}]
% [->] (4) to node[above] {$a_i$} (0);
%\end{tikzpicture}  
			\caption{Case 1:~$a$ maps~$Q_2\setminus\{f_2\}$ onto~$Q_2\setminus\{f_2\}$ and
			~$b$ is not a permutation on~$Q_2$.}
		% ~$f_2\cdot a= f_2$.}
	\label{fig:case1}
\end{figure} 

(2) Let~$b$ be a permutation symbol in~$A_2$.
Then~$a$ is not a permutation on~$Q_2$, so~$f_2a\ne f_2$,
and therefore~$f_2\notin Q_2a$,
so each state containing~$f_2$ in its second component must be reached by~$b$.
This situation is illustrated in   Fig.~\ref{fig:case2}.
It~follows that every valid state~$(j,Q_2,\{s_3\})$
must be reached   on~$b$, so~$b$ is a permutation on~$Q_1$,

Next, the valid state~$(s_1,\{f_2\},Q_3)$ must be reached on~$b$ as well.
Therefore each state in~$Q_3\setminus  \{s_3\}$ has an in-transition on~$b$.
Moreover, the state~$(f_1b, Q_2,\{s_3\})$ must be reached by~$b$ from
a valid state~$(f_1,Q_2,\{s_3\}\cup T)$;
recall that~$b$ is a permutation on~$Q_1$.
This means that~$s_3b=s_3$. Hence~$b$ is a permutation on~$Q_3$.
Let~$r\in Q_2\setminus\{s_2b,f_2\}$.
Then the valid state~$(f_1b,\{r\},Q_3)$ cannot be reached on~$b$
because otherwise it would be reached from~$(f_1,\{s_2\}\cup S,T)$
and would contain~$s_2b$ in its second component.
It follows that~$a$ is a permutation on~$Q_3$.
Thus both~$a$ and~$b$ perform a permutation in~$A_2$,
which is a contradiction.
\end{proof}

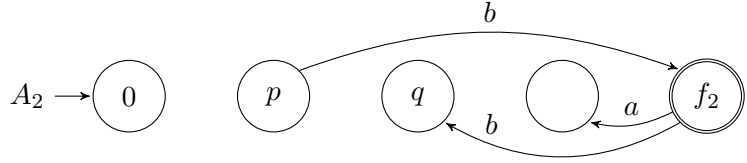
\begin{figure}[h!]
	\centering 
	\tikzset{every state/.style={minimum size=27}}
\begin{tikzpicture}[>=stealth', initial text={~},shorten >=1pt,auto,node distance=1.9cm]

\node[state,initial,initial text={$A_2$}  ]  (0)  [label=center:{$0$}]{};
\node[state  ] (1)  [right of=0,label=center:{$p$}]{};
%\node[state  ] (2)  [right of=1,label=center:{$2$}]{};
\node[state   ] (3)  [right of=1,label=center:{$q$}]{};
\node[state  ] (4)  [right of=3,label=center:{$  $}]{}; 
\node[state,accepting  ] (5)  [right of=4,label=center:{$f_2$}]{}; 

%\draw[->]  (0) to node{$a_i,b$} (1);
%\draw[->]  (1) to node{$a_i$} (3);
%\draw[->]  (2) to node{$a$} (3);
%\draw[->]  (3) to node{$a_i$} (4);
%\draw[->]  (4) to node{$a_i$} (5);
\draw[->] [bend left=20] (1.45) to node[above]{$b$}(5.135);
\draw[->] [bend left=30] (5.-135) to node[above,pos=0.8]{$b$}(3.-45);
\draw[->] [bend left=20] (5.-155) to node[above]{$a$}(4.-45);
%\draw[->](0)to[loop above]node {$\Sigma\setminus\{a_i,b\}$}(0);
%\draw[->](1)to[loop above]node {$\Sigma\setminus\{a_i\}$}(1);
%\draw[->](4)to[loop above]node {$\Sigma\setminus\{a_i\}$}(4);
%\draw[->](5)to[loop above]node [right, near end] {$a$}(5);
%\draw[->] [bend right=25] (1.135) to node[above]{$b$}(0.45);

\end{tikzpicture}

%\begin{tikzpicture}[>=stealth', initial text={~},shorten >=1pt,auto,node distance=2cm]
%
%\node[state,initial] [label=center:{$0$}] (0){\phantom{0}};
%\node[state] [right of=0, label=center:{$1$}] (1) {\phantom{0}};
%%\node[state] [right of=1, label=center:{$2$}] (2){\phantom{0}};
%\node[state,draw=none] [right of=1] (x) {\dots};
%\node[state] [right of=x, label=center:{$n_i{-}2$}] (3){\phantom{0}};
%\node[state,accepting] [right of=3, label=center:{$n_i-1$}] (4){\phantom{0}};
%
%%%%%%%%%%%%%%%%%%%%%%%%%%%%%%%%
%
%\draw [->] (0) to  node{$a_i,b$} (1);
%\draw [->] (0) to [loop above]  node  %[pos=0.75,anchor=west]
%{$\Sigma\setminus\{a_i,b\}$ } (0);
%
%\draw [->] (1) to  node{$a_i$} (x);
%\draw [->] (1) to [loop above]node %[pos=0.75,anchor=west]
%{$\Sigma\setminus\{a_i\}$} (1);
%
%\draw [->] (x) to  node{$a_i$} (3);
%
%\draw [->] (3) to  node{$a_i$} (4);
%\draw [->] (3) to [loop above] node %[pos=0.75,anchor=west]
% {$\Sigma\setminus\{a_i\}$} (3);
%
%\draw [->] (4) to [loop above] node %[pos=0.75,anchor=west]
%{$\Sigma\setminus\{a_i\}$} (4);
%
%\draw [to path={
%.. controls +(-1,-1) and +(1,-1) .. (\tikztotarget) \tikztonodes}]
% [->] (4) to node[above] {$a_i$} (0);
%\end{tikzpicture}  
	\caption{Case 2:~$a$ maps~$Q_2\setminus\{f_2\}$ onto $Q_2\setminus\{f_2\}$ and 
		$b$ is a permutation on $Q_2$.}
\label{fig:case2}
\end{figure}

Notice that all our~$k$-letter witness DFAs
from Theorem~\ref{thm:k}, except for the first 
and last one, are assumed to have
at least three states.
However, our witnesses over a~$(k+1)$-letter
alphabet from Theorem~\ref{thm:aj_dvojky}
cover also the cases when some of given DFAs
have two states. Although, we are not able
to cover such cases by using just~$k$ letters,
we can do it
providing that all automata have two states.
 
Let~$\Sigma=\{b,c,a_2,a_3,\ldots a_{k-1}\}$ be a~$k$-letter alphabet.
   For~$i=1,2,\ldots,k$,
   let~$A_i=(Q_i,\Sigma,s_i,\cdot,f_i)$
   be a two-state DFA with~$Q_i=\{1,2\}$,
   ~$s_i=1$, ~$f_i=2$,
   and the transitions defined as follows
   (see Figure~\ref{fig:2stateDFAs} for an illustration):  
    \begin{itemize}
        \item  
     $a_i$ with~$i=2,3,\ldots,k-2$ performs the cycle on~$Q_i$
    and the identity on~$Q_j$ with~$j\ne i$;
    \item  $a_{k-1}$
    performs the cycle on~$Q_{k-1}$ and~$Q_k$,
    and the identity on~$Q_1,Q_2,\ldots,Q_{k-2}$;
    \item   $b$ performs the cycle on~$Q_1$, the identity on~$Q_i$ 
    if~$i$ is even, and the contraction~$(f_i\to s_i)$
    on~$Q_i$ if~$i\ge3$ is odd;
    \item  $c$ performs the identity on~$Q_i$
    if~$i$ is odd, and the contraction~$(f_i\to s_i)$
    otherwise. 
    \end{itemize}   

    \begin{figure}[h!]
    \centering
    \vskip-10pt
    \tikzset{every state/.style={minimum size=27}}
\begin{tikzpicture}[>=stealth', initial text={~},shorten >=1pt,auto,node distance=2.6cm]

\node[state,initial,initial text={$A_1$}  ]  (s1)  [label=center:{$s_1$}]{};
\node[state,accepting] (f1)  [right of=s1,label=center:{$f_1$}]{};
 
\node[state,initial,initial text={$A_2$}   ] (s2)  at (6,0) [label=center:{$s_2$}]{};
\node[state, accepting  ] (f2)  [right of=s2,label=center:{$f_2$}]{}; 

\node[state,initial,initial text={$A_3$}   ] (s3)  at (0,-2.5) [label=center:{$s_3$}]{};
\node[state, accepting  ] (f3)  [right of=s3,label=center:{$f_3$}]{}; 

\node[state,initial,initial text={$A_4$}   ] (s4)  at (6,-2.5) [label=center:{$s_4$}]{};
\node[state, accepting  ] (f4)  [right of=s4,label=center:{$f_4$}]{}; 

\node[state,initial,initial text={$A_5$}   ] (s5)  at (0,-5) [label=center:{$s_5$}]{};
\node[state, accepting  ] (f5)  [right of=s5,label=center:{$f_5$}]{};

\node[state,initial,initial text={$A_6$}   ] (s6)  at (6,-5) [label=center:{$s_6$}]{};
\node[state, accepting  ] (f6)  [right of=s6,label=center:{$f_6$}]{};

% \node[state,initial,initial text={$A_7$}   ] (s7)  at (0,-9) [label=center:{$s_7$}]{};
% \node[state, accepting  ] (f7)  [right of=s7,label=center:{$f_7$}]{};

\draw[->,red] [bend left] (s1) to node{$b$}(f1);
\draw[->,red] [bend left] (f1) to node{$b$}(s1);

\draw[->] [bend left] (s2) to node{$a_2$}(f2);
\draw[->] [bend left] (f2) to node{$a_2$}(s2);
\draw[->,blue]   (f2) to node[above]{$c$}(s2);

\draw[->] [bend left] (s3) to node{$a_3$}(f3);
\draw[->] [bend left] (f3) to node{$a_3$}(s3);
\draw[->,red]   (f3) to node{$b$}(s3);

\draw[->] [bend left] (s4) to node{$a_4$}(f4);
\draw[->] [bend left] (f4) to node{$a_4$}(s4);
\draw[->,blue]   (f4) to node{$c$}(s4);

\draw[->] [bend left] (s5) to node{$a_5$}(f5);
\draw[->] [bend left] (f5) to node{$a_5$}(s5);
\draw[->,red]   (f5) to node{$b$}(s5);

\draw[->] [bend left] (s6) to node{$a_5$}(f6);
\draw[->] [bend left] (f6) to node{$a_5$}(s6);
\draw[->,blue]  (f6) to node{$c$}(s6);

% \draw[->] [bend left] (s7) to node{$a_6$}(f7);
% \draw[->] [bend left] (f7) to node{$a_6$}(s7);
% \draw[->,red]  (f7) to node{$b$}(s7);

\draw[->,blue](s1)to[loop above]node {$c$}(s1);
\draw[->,blue](f1)to[loop above]node {$c$}(f1);

% \draw[->,blue](s7)to[loop above]node {$c$}(s7);
% \draw[->,blue](f7)to[loop above]node {$c$}(f7);

\draw[->,red](s2)to[loop above]node {$b$}(s2);
\draw[->,red](f2)to[loop above]node {$b$}(f2);
\draw[->,red](s4)to[loop above]node {$b$}(s4);
\draw[->,red](f4)to[loop above]node {$b$}(f4);
\draw[->,red](s6)to[loop above]node {$b$}(s6);
\draw[->,red](f6)to[loop above]node {$b$}(f6);

\draw[->,blue](s3)to[loop above]node {$c$}(s3);
\draw[->,blue](f3)to[loop above]node {$c$}(f3);

\draw[->,blue](s5)to[loop above]node {$c$}(s5);
\draw[->,blue](f5)to[loop above]node {$c$}(f5);
\end{tikzpicture}
    \caption{Two-state DFAs; $k=6$.
    In each DFA,  the remaining symbols
    perform identities.}
    \label{fig:2stateDFAs}
\end{figure}
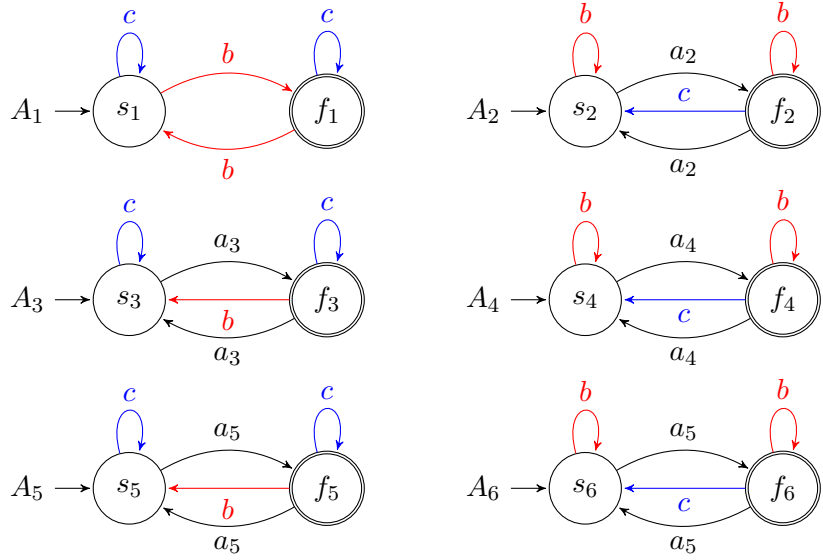

Construct an NFA~$N$ for~$L(A_1)L(A_2)\cdots L(A_k)$
from the DFAs~$A_1,A_2,\ldots,A_k$ as follows:
for each~$i=1,2,\ldots,k-1$,
each~$q\in Q_i$ and~$\sigma\in\Sigma$
such that~$q\cdot \sigma=f_i$ in~$A_i$,
add the transition~$(q,\sigma,s_{i+1})$;
the initial state of~$N$ is~$s_1$
and its final state is~$f_k$.

We prove reachability and distinguishability
of states of the subset automaton~$\cD(N)$
in a similar way as before,
but we have take into account that
to reach a state~$p=(f_1,T_2,T_3,\ldots,T_k)$
from a state~$q=(s_1,S_2,S_3,\ldots,S_k)$,
the symbol~$b$ has to be read.
However, although~$b$ sends~$s_1$ to~$f_1$,
it also sends each non-empty subset~$S_i$
with~$i\ge3$ and~$i$ odd to~$\{s_i\}$.
Then, we have to carefully return~$\{s_i\}$     back to~$S_i$.
This is done in the following claim.
 
\begin{claim}
\label{claim}
   For each valid state~$(s_1,S_2,S_3,\ldots,S_k)$
   with~$s_2\in S_2$
    there is a string~$w_b(S_2,S_3,\ldots,S_k)$
     in~$\{bb,a_2,a_3,\ldots,a_{k-1}\}^*$ 
        such that
    \begin{align}
    \label{c1}
        & (s_1,S_2,S_3b,S_4b,\ldots,S_kb) 
           \xrightarrow{w_b(S_2,S_3,\ldots,S_k)}(s_1,S_2,S_3,\ldots,S_k).
   \end{align}  
   Next, for  valid 
   states~$(f_1,\{s_2\},\{s_3\},S_4, S_5,\ldots,S_k)$
   and~$(f_1,\{s_2\},S_3,S_4,S_5,\ldots,S_k)$,
   there exist strings~$v_b(S_4, S_5,\ldots,S_k)$
   and~$v_c(S_3,S_4,\ldots,S_k)$,
   respectively,
   such that
   \begin{align}
          \label{c2}
     &       (f_1,\{s_2\},\{s_3\},S_4b,S_5b,\ldots,S_kb) 
       \xrightarrow{v_b(S_4, S_5,\ldots,S_k)}(f_1,\{s_2\},\{s_3\},S_4, \ldots,S_k) \\
       \label{c3}
&       (f_1,\{s_2\},S_3,S_4c,S_5c,\ldots,S_kc) 
       \xrightarrow{v_c(S_4,S_5,\ldots,S_k)}(f_1,\{s_2\},S_3,S_4,S_5,\ldots,S_k)
    \end{align}

\end{claim}
  
\begin{proof}
    We prove~(\ref{c1});
    the proofs of~(\ref{c2}) and~(\ref{c3}) are similar.
    
    For~$i\ge2$, we have~$S_ib=S_i$ if~$i$ is even, and~$S_ib=\{s_i\}$ if~$i$ is odd.
    Let~$j$ be the maximum odd number with~$S_j\ne\emp$.
    If~$j=k$, set
    \begin{align*}
        w_k=\begin{cases}
            \eps,             &\text{if~$S_k=\{s_k\}$}; \\
            a_{k-1}a_{k-1},   &\text{if~$S_k=\{s_k,f_k\}$ 
                                and~$s_{k-1}\in S_{k-1}$}, \\
           a_{k-1}bba_{k-1}, &\text{otherwise}.                    
           \end{cases}      
    \end{align*}
    Since~$s_2\in S_2$, \, $w_k$ sends~$(s_1,S_2,S_3b,\ldots,\{s_{k-2}\}, S_{k-1},\{s_{k}\})$   
     to~$(s_1,S_2,S_3b, \ldots,\{s_{k-2}\},S_{k-1}, S_k)$.
     
     Next, for each odd~$i$ with~$3\le i\le j$ and~$i\ne k$, if the string~$w_i$ is defined by
    \[w_i=
        \begin{cases}
         \eps,           &\text{if~$S_i=\{s_i\}$}; \\
         a_i,            &\text{if~$S_i=\{f_i\}$}; \\
         a_i(a_{i-1})^2, &\text{if~$S_i=\{s_i,f_i\}$},        
    \end{cases}
    \]
then we have
     \begin{align*}
        &(s_1,S_2,S_3b,\ldots,\{s_{i-2}\},S_{i-1},\{s_i\},S_{i+1},\ldots,S_k) \xrightarrow{w_i}  \\
     &(s_1,S_2,S_3b, \ldots,\{s_{i-2}\},S_{i-1},~~S_i~, S_{i+1},\ldots,S_k).
    \end{align*}      
    Now, it is enough to set~$w_b(S_2,S_3,\ldots,S_k)=w_jw_{j-2}\cdots w_5w_3$.  
\end{proof}

We are now ready to prove that the two-state
DFAs over a~$k$-letter alphabet described above are witnesses for multiple concatenation
in the case of~$n_1=n_2=\cdots=n_k=2$.

\begin{proposition}
   Let~$A_1,A_2,\ldots,A_k$ be the two state DFAs
   over~$\{b,c,a_2,a_3,\ldots,a_k\}$
   defined above and~$N$ be the NFA for~$L(A_1)L(A_2)\cdots L(A_k)$.
     Then all valid states are reachable and pairwise distinguishable 
  in the subset automaton~$\cD(N)$.
\end{proposition}

\begin{proof}
 For~$j=2,3\ldots,k$,
 let $\sigma_j\in\{b,c\}$ 
 be the symbol that performs the contraction in~$A_j$
 and the identity in~$A_{j-1}$.
 First, consider a valid state~$(s_1,S_2,S_3,\ldots,S_k)$.
Let~$\ell=\max\{i\mid S_i\ne\emp\}$.
    If~$\ell=k$, set
    \begin{align*}
        w_k=\begin{cases}
            \eps,             &\text{if~$S_k=\{s_k\}$}; \\
            a_{k-1}\sigma_k^2a_{k-1}, &\text{if~$S_k=\{f_k\}$}; \\ 
            a_{k-1}^2,   &\text{if~$S_k=\{s_k,f_k\}$}.
        \end{cases}      
    \end{align*}
    Then the string~$w_k$
    sends~$(s_1,\{s_2\},\{s_3\},\ldots,\{s_{k-1}\},\{s_k\})$ to~$(s_1,\{s_2\},\{s_3\},\ldots,\{s_{k-1}\},S_k)$.
    Next, for each~$i$ with~$3\le i \le j$ and~$i\ne k$, set
    \begin{align*}
         w_i=
        \begin{cases}
         \eps,           &\text{if~$S_i=\{s_i\}$}; \\
         a_i,            &\text{if~$S_i=\{f_i\}$}; \\
           a_ia_{i-1}^2, &\text{if~$S_i=\{s_i,f_i\}$}.        
    \end{cases}
    \end{align*}
    Then~$w_i$
    sends the three 
    components~$(\{s_{i-2}\},\{s_{i-1}\},\{s_i\})$ 
    to~$(\{s_{i-2}\},\{s_{i-1}\},\{S_i\})$,
    and fixes all the remaining components.
Thus we have
 \begin{align*}
  (s_1,\emp,\emp,\ldots,\emp)
   \xrightarrow{b^{2}a_2^{2}a_3^{2}\cdots a_{\ell-1}^{2}}
  & (s_1,\{s_2\},\{s_3\},\ldots,\{s_\ell\},\emp,\emp,\ldots,\emp) \\
   \xrightarrow{~~w_\ell w_{\ell-1}\cdots w_4 w_3~~}
  & (s_1,\{s_2\},S_3,S_4,\ldots,~S_\ell,\emp,\emp,\ldots,\emp) \\
  \xrightarrow{a_2}
   & (s_1,\{f_2\},S_3,S_4,\ldots,~S_\ell,\emp,\emp,\ldots,\emp) \\
   \xrightarrow{bb}
   & (s_1,\{s_2,f_2\},S_3b,S_4b,\ldots,~S_\ell b,\emp,\emp,\ldots,\emp) \\
   \xrightarrow{w_b(S_3,S_4,\ldots,S_\ell,\emp,\emp,\ldots,\emp)}
   & (s_1,\{s_2,f_2\},S_3,S_4,\ldots,~S_\ell,\emp,\emp,\ldots,\emp),   
 \end{align*}
 where~$w_b(S_3,S_4,\ldots,S_\ell,\emp,\emp,\ldots,\emp)$
 is the string given by Claim~\ref{claim}(\ref{c1}).
 This proves the reachability of all valid states~$(s_1,S_2,S_3,\ldots,S_k)$.
 The valid state~$(f_1,\{s_2\},\emp,\emp,\ldots,\emp)$
 is reached from the valid state~$(s_1,\{s_2\},\emp,\emp,\ldots,\emp)$ by~$b$.
 Next, we have
 \begin{align*}
      (s_1,\{s_2\},\{s_3\},S_4,..., S_k)\xrightarrow{b}
&(f_1,\{s_2\},\{s_3\},S_4b,S_5b,..., S_kb) \\
\xrightarrow{v_b(S_4,S_5,\ldots,S_k)}   
     &(f_1,\{s_2\},\{s_3\},S_4, S_5,\ldots,S_k) \\
     \xrightarrow{a_3}  
  & (f_1,\{s_2\},\{f_3\},S_4, S_5,..., S_k) \\
  \xrightarrow{a_2}  
  &(f_1,\{s_2,f_2\},\{s_3,f_3\},S_4, S_5,..., S_k)\\
  \xrightarrow{c}  
   &(f_1,\{s_2\},\{s_3,f_3\},S_4c,S_5c,..., S_kc) \\\xrightarrow{v_c(S_4,S_5,\ldots,S_k)} 
 & (f_1,\{s_2\},\{s_3,f_3\},S_4,S_5,..., S_k), 
 \end{align*}
 where~$v_b(S_4,S_5,\ldots,S_k)$ and~$v_c(S_4,S_5,\ldots,S_k)$
 are strings given by Claim~\ref{claim}(\ref{c2}) and~(\ref{c3}),
 respectively.
 Finally, we have
\begin{align*}
    &(f_1,\{s_2\},\{s_3\},S_4,S_5,..., S_k)\xrightarrow{a_2}  
    (f_1,\{s_2,f_2\},\{s_3\},S_4,S_5,..., S_k),
\end{align*}
  which proves the reachability of all valid states~$(f_1,S_2,S_3,\ldots,S_k)$. 
   Thus, all valid states are reachable.
 
 Now we prove distinguishability.
 Recall that
 for~$j=2,3\ldots,k$,
 \, $\sigma_j\in\{b,c\}$ 
 is the symbol that performs the contraction in the DFA~$A_j$
 and the identity in the DFA~$A_{j-1}$.
 Consider two distinct valid states~$p=(S_1,S_2,S_3,\ldots,S_k)$ and~$q=(T_1,T_2,T_3,\ldots,T_k)$.
  First, let~$S_k\ne T_k$.
 If~$f_k\in S_k\setminus T_k$,
 then~$p$ is final, while~$q$ is non-final. Otherwise, the string~$a_{k-1}$
 is accepted from~$p$ and rejected
 from~$q$.

 Now, let~$S_i\ne T_i$
 for some~$i\in\{1,2,\ldots,k-1\}$
 and~$S_j=T_j$ for~$j=i+1,i+2,\ldots,k$.
 Let us show that there is a string 
 that sends~$p$ and~$q$
 to two states which differ in
 their~$(i+1)$st component.
 Then, distinguishability
 follows by induction.
 
 If~$i=1$, then~$p=(\{f_1\},S_2,S_3,\ldots,S_k)$
 and~$q=(\{s_1\},S_2,S_3,\ldots,S_k)$, 
 so~$S_1$ and~$T_1$ differ in~$f_1$.
 Since~$s_2\in S_2$, the symbol~$\sigma_2$
 sends~$p$ and~$q$
 to~$(f_1,\{s_2\},S_3',\ldots,S'_k)$
 and~$(s_1,\{s_2\},S_3', \ldots,S_k')$,
 respectively, and then~$a_2$
 sends the resulting states
 to~$(f_1,\{s_2,f_2\},S_3'',\ldots,S_k'')$
 and~$(s_1,\{f_2\},S_3'', \ldots,S_k'')$
 which differ in~$s_2$.
 
 Let~$i\ge2$
 and~$s\in S_i\setminus T_i$. Then the string~$a_i^{2-s}$
 sends~$p$ and~$q$
 to states~$p'$ and~$q'$
 which differ in~$f_i$.
 If the resulting states 
 differ in their~$(i+1)$st
 components,
 we have the desired result.
 Otherwise, the string~$\sigma_{i+1}$
 sends~$p'$ and~$q'$ to two states which 
 have~$\{s_{i+1}\}$
 in their~$(i+1)$st components.
 Then, the resulting states
 are sent to states which differ
 in~$s_{i+1}$
 by~$a_{i+1}$ if~$i\le k-2$,
 and by~$a_{k-1}\sigma_{k}a_{k-1}$
 if~$i=k-1$; remind that~$\sigma_k$
 sends the~$(k-2)$th component to~$\{s_{k-2}\}$. 
 
 This proves distinguishability, and concludes the proof.
\end{proof}
 
\section{Binary and Ternary Languages}

In this section, we examine the state complexity of multiple concatenation
on binary and ternary languages.
Our aim is to show that in the binary case,
the resulting complexity is still exponential in $n_2,n_3,\ldots,n_k$,
and in the ternary case, it is the same as in the general case,
up to a multiplicative constant
depending on~$k$.
Let us start with the following example.
 
\begin{example}\rm
\label{ex:binary}
 Let $n\ge3$ and $N$ be the NFA shown in Figure~\ref{fig:Abinary} 
 that  recognizes the language of strings over~$\{a,b\}$
 which have an~$a$ in the~$(n-1)$st position from the end.
 
  Let us show that
 each sub\-set~$S\subseteq [1,n]$ with $1\in S$
 is reachable in the subset automaton $\cD(N)$. 
 The proof is by induction on~$|S|$.
 The basis, with~$|S|=1$, holds true since~$\{1\}$ is 
 the initial state. % of~$\cD(N)$. % is $\{0\}$.
 Let~$|S|\ge2$ and~$1\in S$.
 Let~$m=\min(S\setminus\{1\})$.
 Set~$S'=ab^{m-2}(S\setminus\{1,m\})$.
 Then~$S'\subseteq[2,n-m+1]$ and~$|S'|=|S|-2$.
 We have
 $$
 \{1\}\cup S' 
 \xrightarrow{a} \{1,2\}\cup b^{s-2}(S\setminus \{1,s\})
 \xrightarrow{b^{s-2}} \{1,s\} \cup (S\setminus \{1,s\})=S,
 $$
 where the leftmost set of size~$|S|-1$ is reachable by induction. 
\qed\end{example}

\begin{figure}[h!]
\centering
\tikzset{every state/.style={minimum size=27}}
\begin{tikzpicture}[>=stealth', initial text={~},shorten >=1pt,auto,node distance=2cm]

\node[state,initial,initial text={$N$}](p0)[label=center:{$1$}]{};
\node[state] (p1)[right of=p0,label=center:{$2$}]{};
\node[state] (p2)  [right of=p1,label=center:{$3$}]{};
\node[state,draw=none] (p3)  [right of=p2,label=center:{$\ldots$}]{};
\node[state] (p4)  [right of=p3,label=center:{$n {-}1$}]{}; 
\node[state,accepting](p5)[right of=p4,label=center:{$n$}]{}; 

\draw[->]  (p0) to node{$a$} (p1);
\draw[->]  (p1) to node{$a,b$} (p2);
\draw[->]  (p2) to node{$a,b$} (p3);
\draw[->]  (p3) to node{$a,b$} (p4);
\draw[->]  (p4) to node{$a,b$} (p5);
\draw[->](p0)to[loop above]node{$a,b$}(p0); 
\end{tikzpicture}
 \caption{A binary 
 NFA $N$ such that every set $\{1\}\cup S$
 is reachable in $\cD(N)$.}
\label{fig:Abinary}
\end{figure}
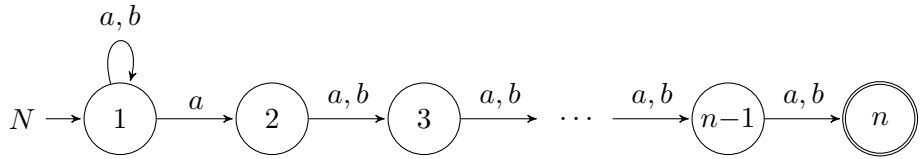

We now use the result from the above example
to get a lower bound on the state complexity
of multiple concatenation on binary languages.
The idea is to describe binary DFAs
in such a way that the NFA for their concatenation
would accept, except for a finite set,
the set of strings having an~$a$
in an appropriate position from the end.

\begin{theorem}
\label{thm:binary}
 Let $k\ge3$, $n_1\ge3$, $n_2\ge4$, and $n_i\ge3$ for $i=3,4,\ldots,k$. 
 Let~$A_1, A_2, \ldots,A_k$ be the binary DFAs shown in Figure~\ref{fig:binary}.
 Then every DFA for the language~$L(A_1) L(A_2)  \cdots L(A_k)$ 
 has at least $ n_1-1 + (1/2^{2k-2})\, 2^{n_2+n_3+\cdots+n_k}$ states.
\end{theorem}

 \begin{proof}
 Construct an NFA   for~$L(A_1)L(A_2)\cdots L(A_k)$
 from the DFAs~$A_1, A_2,\ldots,A_k$
 by adding the transitions~$(f_1{-}1,b,s_2)$,
 $(f_1,a,s_2),(f_1,b,s_2)$,
 and~$(f_i{-}1,\sigma,s_{i+1})$ for~$i=2,3,\ldots,k-1$ and~$\sigma\in\{a,b\}$,
 by making states~$f_1,f_2,\ldots,f_{k-1}$ non-final,
 and states~$s_2,s_3,\ldots,s_k$ non-initial.
 In this NFA, the  states 
 $f_i$ and~$f_i{+}1$ with~$2\le i\le k{-}1$,
 as well as the state~$f_k+1$ are dead, so we can omit them.
 Let~$N$ be the resulting~NFA;
 see Figure~\ref{fig:binaryNFA} for an illustration.
 
 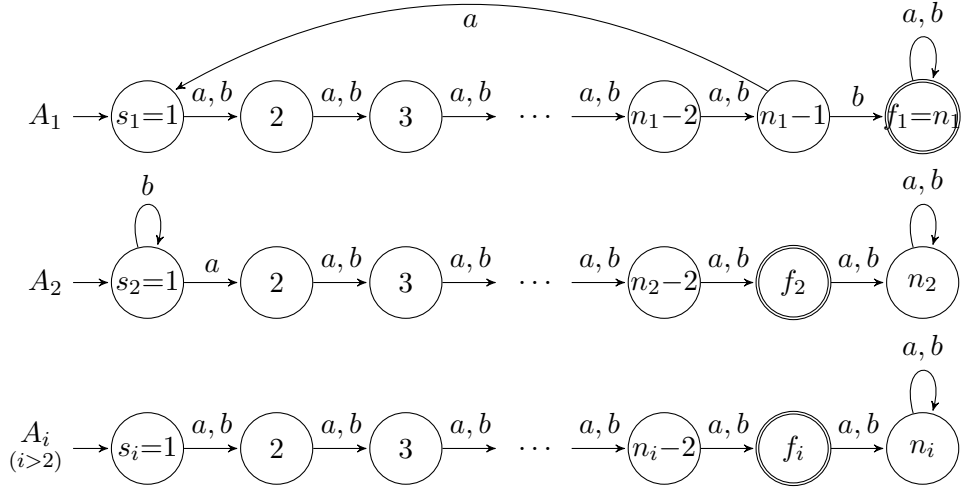
\begin{figure}[t]
 \centering
\tikzset{every state/.style={minimum size=27}}
\begin{tikzpicture}[>=stealth', initial text={~},shorten >=1pt,auto,node distance=1.7cm]

\node[state,initial,initial text={$A_1$}](p0)[label=center:{$s_1{=}1$}]{};
\node[state](p1)[right of=p0,label=center:{$2$}]{};
\node[state](p2)[right of=p1,label=center:{$3$}]{};
\node[state,draw=none](p3)[right of=p2,label=center:{$\ldots$}]{};
\node[state](p4)[right of=p3,label=center:{$n_1{-}2$}]{}; 
\node[state](p45)[right of=p4,label=center:{$n_1{-}1$}]{}; 
\node[state,accepting](p5)[right of=p45,label=center:{$f_1{=}n_1$}]{}; 

\draw[->]  (p0) to node{$a,b$} (p1);
\draw[->]  (p1) to node{$a,b$} (p2);
\draw[->]  (p2) to node{$a,b$} (p3);
\draw[->]  (p3) to node{$a,b$} (p4);
\draw[->]  (p4) to node{$a,b$} (p45);
\draw[->]  (p45) to node{$b$}   (p5);
\draw[->]  (p45.135) [bend right] to node{$a$}(p0.45);
\draw[->]  (p5)to[loop above]node{$a,b$}(p5);

\node[state,initial,initial text={$A_2$}](0) at (0,-2.2) [label=center:{$s_2{=}1$}]{};
\node[state](1)[right of=0,label=center:{$2$}]{};
\node[state](2)[right of=1,label=center:{$3$}]{};
\node[state,draw=none] (3)  [right of=2,label=center:{$\ldots$}]{};
\node[state](3a)[right of=3,label=center:{$n_2{-}2$}]{};
\node[state,accepting](4)[right of=3a,label=center:{$f_2$}]{}; 
\node[state](5)[right of=4,label=center:{$ n_2$}]{}; 

\draw[->] (0) to node{$a$}   (1);
\draw[->] (1) to node{$a,b$} (2);
\draw[->] (2) to node{$a,b$} (3);
\draw[->] (3) to node{$a,b$} (3a);
\draw[->] (3a)to node{$a,b$} (4);
\draw[->] (4) to node{$a,b$} (5);
\draw[->] (0) to[loop above]node{$b$}(0);
\draw[->] (5) to[loop above]node{$a,b$}(5);

\node[state,initial,initial text={$\underset{(i>2)}{A_i}$}] (q0) at (0,-4.4)[label=center:{$s_i{=}1$}]{};
\node[state]           (q1) [right of=q0,label=center:{$2$}]{};
\node[state]           (q12)[right of=q1,label=center:{$3$}]{};
\node[state,draw=none] (q2)  [right of=q12,label=center:{$\ldots$}]{};
\node[state]           (q3)  [right of=q2,label=center:{$n_i{-}2$}]{};
\node[state,accepting] (q4)  [right of=q3,label=center:{$f_i$}]{}; 
\node[state]           (q5)  [right of=q4,label=center:{$n_i$}]{}; 

\draw[->]  (q0) to node{$a,b$} (q1);
\draw[->]  (q1) to node{$a,b$} (q12);
\draw[->]  (q12)to node{$a,b$} (q2);
\draw[->]  (q2) to node{$a,b$} (q3);
\draw[->]  (q3) to node{$a,b$} (q4);
\draw[->]  (q4) to node{$a,b$} (q5); 
\draw[->]  (q5) to[loop above]node{$a,b$}(q5);
\end{tikzpicture}
\vskip5pt
 \caption{Binary DFAs $A_1,A_2,$ and  $A_i$ for $i=3,4,\ldots,k$
 	($n_2\ge4$ and~$n_i\ge3$ if~$i\ne2$)
         meeting the lower 
         bound $n_1-1 + 
        (1/2^{2k-1}) 2^{n_2+n_3+\cdots +n_k}$
        for multiple concatenation; ~$f_i=n_i-1$ for~$i\ge2$.}
 \label{fig:binary}
 \end{figure}

  \begin{figure}[h!]
 \centering
 \tikzset{every state/.style={minimum size=27}}
\begin{tikzpicture}[>=stealth', initial text={~},shorten >=1pt,auto,node distance=1.9cm]

\node[state,initial,initial text={$N$}](s1)[label=center:{$s_1{=}1$}]{};
\node[state](12)[right of=s1,label=center:{$2$}]{};
\node[state](13)[right of=12,label=center:{$3$}]{};
\node[state](f1)[right of=13,label=center:{$f_1{=}4$}]{}; 

\node[state](s2) at (0,-2.2)[label=center:{$s_2{=}1$}]{};
\node[state](22)[right of=s2,label=center:{$2$}]{};
\node[state](23)[right of=22,label=center:{$3$}]{};
\node[state](24)[right of=23,label=center:{$4$}]{};
\node[state,dashed](f2)[right of=24,label=center:{$f_2{=}5$}]{};
\node[state,dashed](d2)[right of=f2,label=center:{$6$}]{};

\node[state](s3) at (0,-4.4)[label=center:{$s_3{=}1$}]{};
\node[state](32)[right of=s3,label=center:{$2$}]{};
\node[state](33)[right of=32,label=center:{$3$}]{};
\node[state,dashed](f3)[right of=33,label=center:{$f_3{=}4$}]{};
\node[state,dashed](d3)[right of=f3,label=center:{$5$}]{};

\node[state](s4) at (0,-6.6)[label=center:{$s_4{=}1$}]{};
\node[state](42)[right of=s4,label=center:{$2$}]{};
\node[state](43)[right of=42,label=center:{$3$}]{};
\node[state,accepting](f4)[right of=43,label=center:{$f_4{=}4$}]{};
\node[state,dashed](d4)[right of=f4,label=center:{$5$}]{};

\draw[->](s1)to node{$a,b$}(12);
\draw[->](12)to node{$a,b$}(13);
\draw[->](13)to node{$b$}(f1);
\draw[->] (13) to [bend right=40] node  {$a$}(s1);
\draw[->](f1)to[loop above]node {$a,b$}(f1);

\draw[->](s2)to[loop above]node {$b$}(s2);
\draw[->](s2)to node{$a$}(22);
\draw[->](22)to node{$a,b$}(23);
\draw[->](23)to node{$a,b$}(24);
\draw[->,dashed](24)to node{$a,b$}(f2);
\draw[->,dashed](f2)to node{$a,b$}(d2);
\draw[->,dashed](d2)to[loop above]node {$a,b$}(d2);

\draw[->](s3)to node{$a,b$}(32);
\draw[->](32)to node{$a,b$}(33);
\draw[->,dashed](33)to node{$a,b$}(f3);
\draw[->,dashed](f3)to node{$a,b$}(d3);
\draw[->](d3)to[loop above]node {$a,b$}(d3);

\draw[->](s4)to node{$a,b$}(42);
\draw[->](42)to node{$a,b$}(43);
\draw[->](43)to node{$a,b$}(f4);
\draw[->,dashed](f4)to node{$a,b$}(d4);
\draw[->](d4)to[loop above]node {$a,b$}(d4);

\draw[->,red,dashed](13)to[above,pos=0.6] node{$b$}(s2.60);
\draw[->,red,dashed](f1.210)to[pos=0.4] node{$a,b$}(s2.45);
\draw[->,red,dashed](24.210)to node[pos=0.4]{$a,b$}(s3.45);
\draw[->,red,dashed](33)to node[pos=0.4]{$a,b$}(s4.45);

\end{tikzpicture}
\vskip10pt
 \caption{A binary NFA for $L(A_1)L(A_2)L(A_3)L(A_4)$
         where $n_1=4$, $n_2=6$, $n_3=n_4=5$.}
 \label{fig:binaryNFA}
 \end{figure}
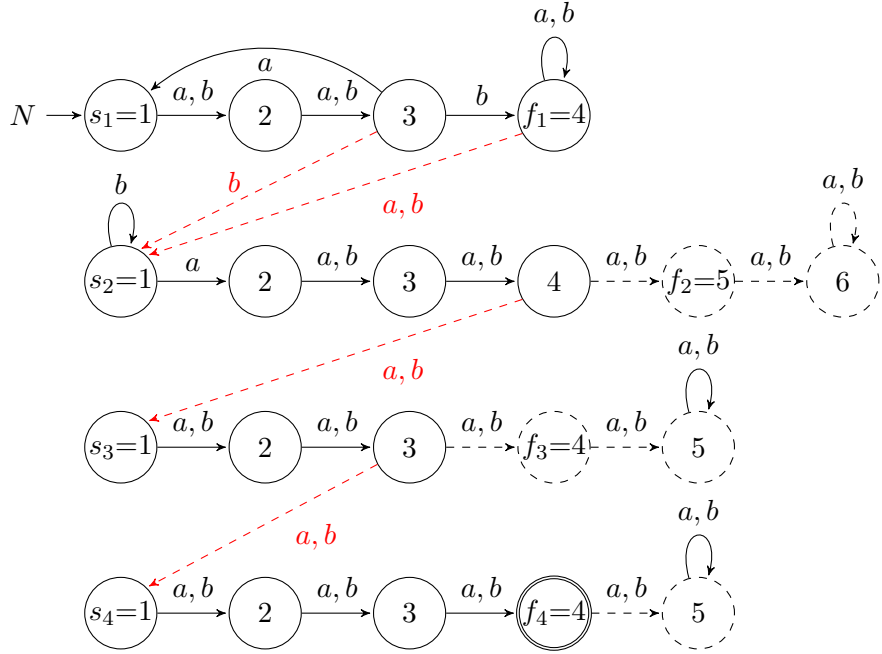

 In the subset automaton~$\cD(N)$,
 each state~$(j,\emp,\emp,\ldots,\emp)$ 
 with $1\le j \le f_1-1$ 
 is reached from the initial state~$(s_1,\emp,\emp,\ldots,\emp)$
 by~$b^{j-1}$,
 and~$(f_1,\{s_2\},\emp,\emp,\ldots,\emp)$ is reached 
 from~$(f_1{-}1,\emp,\emp,\ldots,\emp)$ by~$b$.
 Starting in the state~$f_1$,
 the NFA $N$ accepts all strings having an $a$
 in    position~$n_2-2+n_{3}-2+
 \cdots+ n_{k-1}-2+n_{k}-1$
 from~the end. As shown in Example~\ref{ex:binary},
 every state~$(f_1 \{s_2\}\cup S_2,S_3,\ldots,S_k)$
 with~$S_2\subseteq \{2,3,\ldots,n_2-2\}$,
 $S_i\subseteq \{1,2,\ldots,n_i-2\}$ for~$i=3,4, \ldots,k-1$,  
 and~$S_k\subseteq\{1,2,\ldots,n_k-1\}$
 is~reachable.
 This gives 
 $$n_1-1+2^{n_2-3+n_3-2+n_4-2+\cdots+n_{k-1}-2+n_k-1}
 =n_1-1+  
 (1/2^{2k-2})2^{n_2+n_3+\cdots+n_k}$$
 reachable states.
 
 Moreover,
 each singleton set is co-reachable in~$N$  via a string in~$a^*$,
 except for~$\{q\}$ where~$q$ is a non-final  state of~$A_1$.
 By Lemma~\ref{le:dist},   
 the reachable states~$(i, S_2, S_3,\ldots,S_k)$ and~$(j, T_2, T_3,\ldots,T_k)$ 
 are distinguishable if they differ in a state of~$A_i$ with~$i\ge2$ or in~$f_1$.
 Next, the states~$(i,S_2,S_3,\ldots,S_k)$ and~$(j,S_2,S_3,\ldots,S_k)$
 with~$1\le i < j < f_1$ are sent to states that differ in~$f_1$ 
 by~$b^{f_1-j}$.
\end{proof}
 
Our next result shows
that a trivial upper bound~$n_12^{n_2+n_3+\cdots+n_k}$
can be met, 
up to a multiplicative constant
depending on~$k$,
by the concatenation of~$k$ ternary languages.
Thus, this trivial upper bound
is asymptotically tight
in the ternary case.

\begin{theorem}
 Let $k\ge2$, $n_1\ge3$, $n_2\ge4$, and $n_i\ge3$ for $i=3,4,\ldots,k$. 
 There exist ternary DFAs~$A_1, A_2,\ldots,A_k$ 
 such that  every DFA recognizing 
 the concatenation~$L(A_1) L(A_2)  \cdots L(A_k)$ 
 has at least $(1/2^{2k-2})\, n_1  2^{n_2+n_3+\cdots +n_k}$ states.
\end{theorem}

\begin{proof}
 Let us add the transitions on symbol~$c$
 to the binary automata shown in Figure~\ref{fig:binary} 
 as follows: $c\colon (1,2,\ldots,n_1)$ in $A_1$,
 $c\colon (f_i\to f_i+1)$ in~$A_i$ with~$2 \le i \le k-1$,
 and $c\colon (1)$ in~$A_k$.
 Construct the NFA~$N$ for~$L(A_1) L(A_2)  \cdots L(A_k)$ with
 omitted dead states as in the binary case; see Figure~\ref{fig:NFA_ternary}
 for an illustration.
  As shown in the proof of Theorem~\ref{thm:binary},
 the subset automaton $\cD(N)$ has~$(1/2^{2k-2})2^{n_2+n_3+\cdots+n_k}$
 reachable states of the form~$(f_1,S_2,S_3,\ldots,S_k)$.
 Each such state is sent to the state~$(j,S_2,S_3,\ldots,S_k)$  
 with~$1\le j \le f_1-1$   by the string~$c^{j}$. 
 Moreover, in the NFA $N$,   
 each singleton set is co-reachable via a string in $a^*c^*$.
 By Corollary~\ref{cor}, all states of $\cD(N)$
 are pairwise distinguishable.
 This gives the desired lower bound.
\end{proof}  

  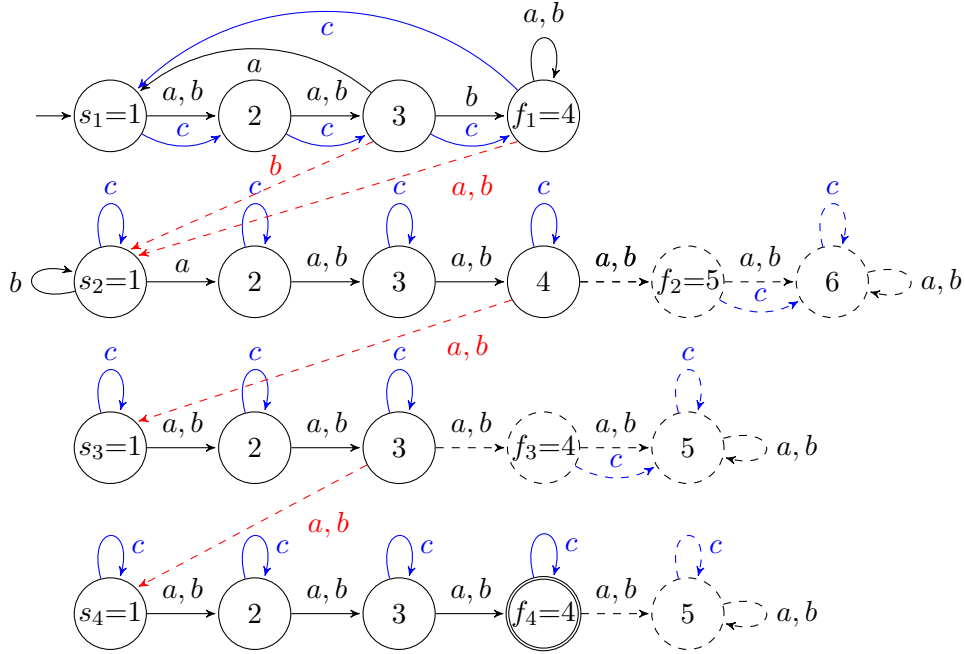
\begin{figure}[h!]
 \centering 
 \scalebox{1}{\tikzset{every state/.style={minimum size=27}}
\begin{tikzpicture}[>=stealth', initial text={~},shorten >=1pt,auto,node distance=1.9cm]

\node[state,initial,initial text={}](s1)[label=center:{$s_1{=}1$}]{};
\node[state](12)[right of=s1,label=center:{$2$}]{};
\node[state](13)[right of=12,label=center:{$3$}]{};
\node[state](f1)[right of=13,label=center:{$f_1{=}4$}]{}; 

\node[state](s2) at (0,-2.2)[label=center:{$s_2{=}1$}]{};
\node[state](22)[right of=s2,label=center:{$2$}]{};
\node[state](23)[right of=22,label=center:{$3$}]{};
\node[state](24)[right of=23,label=center:{$4$}]{};
\node[state,dashed](f2)[right of=24,label=center:{$f_2{=}5$}]{};
\node[state,dashed](d2)[right of=f2,label=center:{$6$}]{};

\node[state](s3) at (0,-4.4)[label=center:{$s_3{=}1$}]{};
\node[state](32)[right of=s3,label=center:{$2$}]{};
\node[state](33)[right of=32,label=center:{$3$}]{};
\node[state,dashed](f3)[right of=33,label=center:{$f_3{=}4$}]{};
\node[state,dashed](d3)[right of=f3,label=center:{$5$}]{};

\node[state](s4) at (0,-6.6)[label=center:{$s_4{=}1$}]{};
\node[state](42)[right of=s4,label=center:{$2$}]{};
\node[state](43)[right of=42,label=center:{$3$}]{};
\node[state,accepting](f4)[right of=43,label=center:{$f_4{=}4$}]{};
\node[state,dashed](d4)[right of=f4,label=center:{$5$}]{};

\draw[->](s1)to node{$a,b$}(12);
\draw[->](12)to node{$a,b$}(13);
\draw[->](13)to node{$b$}(f1);
\draw[->] (13) to [bend right=40] node  {$a$}(s1);
\draw[->](f1)to[loop above]node {$a,b$}(f1);

%[right,near end]

\draw[->,blue](s1)to [bend right] node{$c$}(12);
\draw[->,blue](12)to[bend right] node{$c$}(13);
\draw[->,blue](13)to[bend right] node{$c$}(f1);
\draw[->,blue] (f1) to [bend right=45] node  {$c$}(s1);

\draw[->,blue](s2)to[loop above]node {$c$}(s2);
\draw[->](s2)to[loop left, ]node {$b$}(s2);
\draw[->](s2)to node{$a$}(22);
\draw[->](22)to node{$a,b$}(23);
\draw[->](23)to node{$a,b$}(24);
\draw[->,dashed](24)to node{$a,b$}(f2);
\draw[->,dashed](f2)to node{$a,b$}(d2);
% \draw[->,dashed](d2)to[loop above,right,near end]node {$a,b$}(d2);

\draw[->,blue](22)to[loop above] node{$c$}(22);
\draw[->,blue](23)to [loop above] node{$c$}(23);
\draw[->,blue](24)to[loop above] node{$c$}(24);
\draw[->,dashed](24)to node{$a,b$}(f2);
\draw[->,blue,dashed](f2)to[bend right] node{$c$}(d2);
\draw[->,dashed](d2)to[loop right]node {$a,b$}(d2);
\draw[->,blue,dashed](d2)to[loop above]node {$c$}(d2);

\draw[->](s3)to node{$a,b$}(32);
\draw[->](32)to node{$a,b$}(33);
\draw[->,dashed](33)to node{$a,b$}(f3);
\draw[->,dashed](f3)to node{$a,b$}(d3);
\draw[->,dashed](d3)to[loop  right]node {$a,b$}(d3);

\draw[->,blue](s3)to[loop  above]node {$c$}(s3);
\draw[->,blue](32)to[loop  above]node {$c$}(32);
\draw[->,blue](33)to[loop  above]node {$c$}(33);
\draw[->,dashed,blue](f3)to[bend right] node{$c$}(d3);
\draw[->,dashed,blue](d3)to[loop  above]node {$c$}(d3);

\draw[->](s4)to node{$a,b$}(42);
\draw[->](42)to node{$a,b$}(43);
\draw[->](43)to node{$a,b$}(f4);
\draw[->,dashed](f4)to node{$a,b$}(d4);
\draw[->,dashed](d4)to[loop right]node {$a,b$}(d4);

\draw[->,blue](s4)to[loop above,right, near end]node {$c$}(s4);
\draw[->,blue](42)to[loop above,right, near end]node {$c$}(42);
\draw[->,blue](43)to[loop above,right, near end]node {$c$}(43);
\draw[->,blue](f4)to[loop above,right, near end]node {$c$}(f4);
\draw[->,blue,dashed](d4)to[loop above,right, near end]node {$c$}(d4);

\draw[->,red,dashed](13.225)to[above,pos=0.4] node{$b$}(s2.60);
\draw[->,red,dashed](f1.225)to[pos=0.2] node{$a,b$}(s2.45);
\draw[->,red,dashed](24.210)to node[pos=0.2]{$a,b$}(s3.45);
\draw[->,red,dashed](33)to node[pos=0.3]{$a,b$}(s4.45);

\end{tikzpicture}}
\tikzset{every state/.style={minimum size=30}}
\vskip10pt
 \caption{A ternary NFA for $L(A_1)L(A_2)L(A_3)L(A_4)$
         where $n_1=4$, $n_2=6$, $n_3=n_4=5$.}
 \label{fig:NFA_ternary}
 \end{figure}

\section{Unary Languages}   
 
The upper bound on the state complexity of concatenation of two unary languages  
is~$n_1n_2$,
and this upper bound can be met by cyclic unary languages 
if~$\gcd(n_1,n_2)=1$~\cite[Theorems~5.4 and~5.5]{yzs94}.
This gives a trivial upper bound~$n_1n_2\cdots n_k$ 
for concatenation of~$k$ unary languages.
Here we show that a tight upper bound for concatenation 
of~$k$ cyclic unary languages is much smaller.
Then we continue our study by investigating
the concatenation of languages of the form~$a^{\mu_i}Y_i$
where~$Y_i$ is a~$\lambda_i$-cyclic.
In both cases, we provide tight upper bounds.
Finally, we consider the case,
when automata may have final states in their tails.

Recall that the state set of a unary automaton
of size~$(\lambda,\mu)$ consists of 
a tail~$q_0,q_1,\ldots,q_{\mu-1}$
and a cycle~$p_0,p_1,\ldots,p_{\lambda-1}$
(with~$p_0=q_0$ if~$\mu=0$),
and its transitions are~$q_0\to q_1\to\cdots\to q_{\mu-1}\to p_0
\to p_1 \to\cdots\to p_{\lambda-1}\to p_0$; cf.~\cite{ps02}.

Let~$n_1,n_2,\ldots,n_k$ be positive integers 
with~$\gcd(n_1,n_2,\ldots,n_k)=1$.
Then~$g(n_1,n_2,\ldots,n_k)$
denotes the Frobenius number, that is, 
the largest integer that cannot be expressed 
as $x_1n_1+x_2n_2+\cdots+x_kn_k$
for some non-negative integers~$x_1,x_2,\ldots,x_k$.
Let us start with the following observation. 
 
\begin{lemma}
	\label{le:largest_cannot}
	  Let~$n_1,n_2,\ldots,n_k$ be positive integers
        with~$\gcd(n_1,n_2,\ldots,n_k)=d$.
	Then each number of the form~$x_1 n_1 + x_2 n_2 + \cdots + x_k n_k$,
	with $x_1,x_2,\ldots,x_k\ge0$,
	is~a~multiple of $d$.
	Furthermore, the largest multiple of $d$
	that cannot be represented as $x_1 n_1 + x_2 n_2 + \cdots + x_k n_k$,
	with~$x_1,x_2,\ldots,x_k\ge0$,
	is~$d\cdot g(\frac{n_1}{d},\frac{n_2}{d},\ldots,\frac{n_k}{d})$.
\end{lemma}

\begin{proof}
	The first claim follows from the fact that each $n_i$
	is~a multiple of $d$.
	Since~$\gcd(\frac{n_1}{d}, \ldots,\frac{n_k}{d})=1$,
	the largest integer that cannot be represented
	as~$x_1 \frac{n_1}{d}+x_2 \frac{n_2}{d}+\cdots+x_k\frac{n_k}{d}$,
	with~$x_1,x_2,\ldots,x_k\ge0$, 
        is~$g(\frac{n_1}{d},\frac{n_2}{d},\ldots,\frac{n_k}{d})$.
		Multiplying by $d$, we get the second claim. 
\end{proof}
 
Let~$f(n_1, n_2,\ldots,n_k)=g(n_1,n_2, \ldots,n_k)+n_1+ n_2+\cdots+n_k$
be the modified Frobenius number, that is, the largest integer
which is  not representable by positive integer linear combinations.
Using this notation, we have the following result.

\begin{theorem}  
	\label{a:thm:cyclic_unary}
	Let~$A_1, A_2,\ldots,A_k$ be unary cyclic  automata
	with~$n_1, n_2,\ldots,n_k$ states, respectively.
	Let~$d=\gcd(n_1,n_2, \ldots,n_k)$.
	Then~$L(A_1) L(A_2) \cdots L(A_k)$
	is recognized by a DFA 
    of size~$(\lambda,\mu)$,
    where~$\lambda=d$ and~$\mu=d\cdot f(\frac{n_1}{d}, \frac{n_2}{d},\ldots,\frac{n_k}{d})-k+1$, and this upper bound is tight.    
\end{theorem}

\begin{proof}
 Denote~$L_i=L(A_i)$ and~$L=L_1L_2\cdots L_k$.
 We show that~$L$ is recognized by a unary DFA of size~$(\lambda,\mu)$.
 By \cite[Theorem~2]{ps02}, it is enough to show that for 
 every~$m\ge d\cdot f(\frac{n_1}{d}, \frac{n_2}{d}, \ldots,\frac{n_k}{d})-k+1$,
 we have $a^m\in L$ if and only if $a^{m+d}\in L$.	
	
 We can write each language~$L_i$ as $L_i=Z_i(a^{n_i})^*$ 
 where $Z_i = L_i \cap \{a^x\mid 0\le x < n_i\}$; 
 cf.~\cite[Proof of Theorem~8]{ps02}.
 Let $m\ge d \cdot f(\frac{n_1}{d},\frac{n_2}{d},  \ldots,\frac{n_k}{d})-k+1$.

 If~$a^m\in L$,
 then~$m=z_1+x_1 n_1 + z_2+x_2 n_2+\cdots + z_k + x_k n_k$
 where~$a^{z_i}\in Z_i$ and~$x_i\ge0$.
 Since~$m\ge d \cdot f(\frac{n_1}{d}, \frac{n_2}{d},  \ldots,\frac{n_k}{d})-k+1$, we get
\begin{align*}
 &x_1 n_1+ x_2 n_2+\cdots +x_k n_k \ge 
 d \cdot  f(\frac{n_1}{d},\frac{n_2}{d},\ldots,\frac{n_k}{d})-k+1-z_1-z_2-\cdots- z_k\ge\\ 
 & d\cdot f(\frac{n_1}{d},\frac{n_2}{d},\ldots,\frac{n_k}{d})-k+1-(n_1-1)-(n_2-1)- 
       \cdots -(n_k-1)=    
   d\cdot g(\frac{n_1}{d},\frac{n_2}{d}, \ldots,\frac{n_k}{d}) +1.
\end{align*}
 Since~$x_1 n_1+x_2 n_2+\cdots +x_k n_k $ is a multiple of~$d$,
 it follows from  Lemma~\ref{le:largest_cannot}
 that we have $x_1 n_1+x_2 n_2+\cdots +x_k n_k +d = x_1' n_1 + x_2' n_2+\cdots + x_k' n_k$
 for some $x_1',x_2',\ldots,x_k'\ge0$.
 Therefore
 $
    m+d = z_1 +x_1'n_1+ z_2 + x_2' n_2 + \cdots + z_k + x_k' n_k,
 $
 so~$a^{m+d} \in L$.
	
 Conversely, 
 if~$a^{m+d}\in L$, then $m+d=z_1+x_1 n_1 + z_2+x_2 n_2+\cdots + z_k + x_k n_k$
 where~$a^{z_i}\in Z_i$ and $x_i\ge0$. 
 Since~$m\ge d \cdot f(\frac{n_1}{d}, \frac{n_2}{d},  \ldots,\frac{n_k}{d})-k+1$, 
 similarly as in the previous paragraph, we get
\begin{align*}
 x_1 n_1+x_2 n_2+\cdots +x_k n_k-d 
 \ge d\cdot g(\frac{n_1}{d}, \frac{n_2}{d},\ldots,\frac{n_k}{d}) +1,
\end{align*}
 and therefore~$ x_1 n_1+x_2 n_2+\cdots+x_k n_k-d = x_1' n_1+x_2'n_2+\cdots+x_k'n_k$
 for some~$x_1', x_2',\ldots,x_k'\ge~0$.
 Thus~$m=z_1 +x_1'n_1+ z_2 + x_2' n_2 + \cdots + z_k + x_k' n_k$
 and~$a^m\in L$.
	
 To get tightness, consider unary cyclic languages~$L_i=a^{n_i-1}(a^{n_i})^*$
 recognized by unary cyclic~$n_i$-state automata. Let $L=L_1L_2\cdots L_k$.
 As shown above, the language~$L$ is recognized by a unary DFA~$A$
 with a tail of length~$d\cdot f(\frac{n_1}{d},\ldots,\frac{n_k}{d})-k+1$
 and a cycle of size $d$.
 Next, we have~$a^m\in L$ 
 if and only if
 $
  m=(n_1-1)+(n_2-1)+\cdots+(n_k-1)+ x_1 n_1 +x_2 n_2 +\cdots+x_k n_k 
 $
 for some~$x_1,x_2,\ldots,x_k\ge0$. 
 Since~$x_1 n_1 +x_2 n_2 +\cdots+x_k n_k$ is a multiple of $d$,
 the cycle of size $d$ has exactly one final state,
 and therefore it is minimal.
 Furthermore, 
 a string~$a^{d\cdot f(\frac{n_1}{d},\frac{n_2}{d},\ldots,\frac{n_k}{d})-k+\ell d}$ is in $L$ 
 if and only if 
 \[
  d\cdot f(\frac{n_1}{d},\frac{n_2}{d},\ldots,\frac{n_k}{d})-k+\ell d=
  (n_1-1)+(n_2-1)+\cdots+(n_k-1)+ x_1 n_1 +   \cdots+x_k n_k
  \]
 for some~$x_1,x_2,\ldots,x_k\ge0$, 
 which holds if and only if
 \[
   d\cdot g(\frac{n_1}{d},\frac{n_2}{d},\ldots,\frac{n_k}{d})
   +\ell d=x_1n_1+x_2n_2+\cdots+x_kn_k.
 \]
 By Lemma~\ref{le:largest_cannot}, 
 it follows that~$a^{d\cdot f(\frac{n_1}{d},\frac{n_2}{d},\ldots,\frac{n_k}{d})-k}\notin L$,
 while $a^{d\cdot f(\frac{n_1}{d},\frac{n_2}{d},\ldots,\frac{n_k}{d})-k+d}\in L$.
 Hence $A$ is minimal.
\end{proof}

By~\cite[Proposition~2.2]{ga14},
if~$n_1\le n_2\le  \cdots \le n_k$,
then~$g(n_1,n_2,\ldots,n_k)\le n_1n_k$. %~\cite[Proposition~2.2]{ga14}.
This gives an upper bound $n_1n_k/d +n_1+\cdots+n_k-k+1+d$
for  concatenation of $k$ cyclic languages  
where $n_1\le n_2 \le \cdots\le n_k$ and $d=\gcd(n_1,n_2,\ldots,n_k)$.
The result of the previous theorem can be generalized as follows. 

\begin{corollary}
\label{cor_unary}
    For~$i=1,2,\ldots,L_k$,
    let~$L_i=a^{\mu_i}Y_i$
     where~$Y_i$
    is~$\lambda_i$-cyclic be a language
    recognized by a DFA of size~$(\lambda_i,\mu_i)$.
    Let~$d=\gcd(\lambda_1,\lambda_2,\ldots,\lambda_k)$.
    Then the language~$L_1L_2\cdots L_k$
    is recognized by a DFA of size~$(\lambda,\mu)$
    where~$\lambda=d$
    and~$\mu=\mu_1+\mu_2+\cdots+\mu_k+d\cdot f(\frac{\lambda_1}{d},\frac{\lambda_2}{d},\ldots,\frac{\lambda_k}{d})-k+1$,
    and this upper bound is tight.
\end{corollary}

\begin{proof}
    The language~$L_1L_2\cdots L_k$
    is a concatenation of the singleton
    language~$a^{\mu_1+\mu_2+\cdots+\mu_k}$
    recognized by a DFA of size~$(1,\mu_1+\mu_2+\cdots+\mu_k+1)$
    and the concatenation of cyclic languages~$Y_1Y_2\cdots Y_k$.
    Now the result follows from the previous theorem
    since we can  merge the final state of the  automaton
    for the singleton language
    with the initial state of the DFA for~$Y_1Y_2\cdots Y_k$; cf.~\cite[Theorem~6]{ps02}.
    The upper bound is met by languages~$L_i=a^{\mu_i+\lambda_i-1}(a^{\lambda_i})^*$.
\end{proof} 
 
In the case of concatenation of two languages,
the length of the resulting cycle
may be equal to the least common multiple of the lengths of  cycles
in given automata providing that they have final states
in their tails~\cite[Theorems~10 and~11]{ps02}.
The next example shows that in some cases
this is the optimal way how to get the maximum complexity 
of concatenation of languages recognized by~$m$-state and~$n$-state
unary DFAs, respectively.
 
\begin{example}\rm
   Given an~$m$-state and~$n$-state unary DFA,
   their concatenation requires~$mn$ states if~$\gcd(m,n)=1$.
   If~$\gcd(m,n)>1$, then we may try to take DFAs
   with smaller cycles of sizes~$m-i$ and~$n-j$, 
   and inspect the complexity
   of concatenation of languages
   recognized by automata of sizes~$(m-i,i)$
   and~$(n-j,j)$.
   
   As shown in~\cite[Theorem~11]{ps02}
   the minimal DFA for concatenation of the unary
   languages~$\{\eps\}\cup a^{m-1}(a^{m-2})^*$
   and~$\{\eps\}\cup a^{n-1}(a^{n-2})^*$,
   that are recognized by automata of 
   sizes~$(m-2,2)$ and~$(n-2,2)$,
   with the set of final states~$\{0,m-1\}$ and~$\{0,n-1\}$, respectively,
   has~$2\lcm(m-2,n-2)+3$ states.
   By our computations,
   the smallest~$m$ and~$n$, 
   for which such automata provide the maximum
   complexity among all automata
   of sizes~$(m-i,i)$ and~$(n-j,j)$,
   are~$m=137\,712$ and~$n=127\,206$.

   Nevertheless, it looks like sometimes
   it could be helpful to decrease the lengths
   of cycles not by two, but just by one,
   and setting the   final state sets
   to~$\{0,m-2\}$ and~$\{0,n-2\}$, respectively;
   our aim is to have a state in both tails,
   and then, to get minimal DFAs,
   the states~$m-1$ and~$n-1$
   have to be  non-final.
   Then, similarly as in the proof of \cite[Theorem~11]{ps02}
   we show that the minimal DFA
   recognizing the concatenation of these two languages has~$2\lcm(m-1,n-1)-1$ states
   provided that~$\gcd(m-1,n-1)>1$ and
   neither~$m-1$ nor~$n-1$ is a multiple of the other.

   Our next goal is to find~$m$ and~$n$
   such that the maximum of complexities of concatenation
   of languages recognized by all automata
   of sizes~$(m-i,i)$ and~$(n-j,j)$
   is achieved if~$i=j=1$ 
   and~$\gcd(m-1,n-1)=2$ by the  languages mentioned in the previous paragraph. In such a case, we have~$2\lcm(m-1,n-1)-1=(m-1)(n-1)-1$.

   By~\cite[Theorems~10 and~12]{ps02},
   the complexity of concatenation of languages
   recognized by automata of sizes~$(m-i,i)$
   and~$(n-j,j)$ is
   at most~$(m-i)(n-j)+i+j$
   if~$\gcd(m-i,n-j)=1$, and at most~$2\lcm(m-i,n-j)+i+j-1$
     if~$\gcd(m-i,n-j)>1$.
   In both cases, the resulting complexity is at most~$(m-i)(n-j)+i+j$.
    Denote this number by~$c_{i,j}=(m-i)(n-j)+i+j$.
  The reader may verify that
\begin{align*}
    c_{i,j} &< (m-1)(n-1)-1 \text{ for all~$i,j\ge1$ and~$(i,j)\ne(1,1)$}, \\
    c_{0,j} & <(m-1)(n-1)-1 \text{ if $j\ge2$ and~$ n+2<m$},\\   
    c_{i,0}  & <(m-1)(n-1)-1  \text{ if $i\ge3$ and~$m<2n-3$}.   
\end{align*}
 If follows that 
  the complexity~$(m-1)(n-1)-1$ could possibly
  be exceeded only by automata of sizes~$(m-i,i)$
  and~$(n-j,j)$ where~$(i,j)\in \{(0,0), (0,1), (1,0), (2,0)\}$.
   Assume that in all of this cases,
   we have~$\gcd(m-i,n-j)\ge3$.
   Then, providing that~$m,n \ge8$, the
   complexity of the  corresponding
   concatenations in these four cases is at most
     $$2 \lcm (m-i,n-j)+i+j-1 < \frac{2}{3}(m-i)(n-j) +i+j
  \le\frac{2}{3} mn +3 <(m-1)(n-1)-1.$$

 Now, let~$m=471$ and~$n=315$. 
 Then~$\gcd(m-1,n-1)=2$ and~$n+2 < m < 2n -3$.
 Moreover, \, $\gcd(471,315)=3$, $\gcd(471,314)=157$,
 $\gcd(470,315)=5$, and~$\gcd(469,315)=7$.
 This means that the maximum complexity of concatenation of a 471-state and 315-state
 unary DFA is achieved by automata
 of sizes~$(470,1)$ and~$(314,1)$
 recognizing languages~$\{\eps\}\cup a^{469}(a^{470})^*$ and~$\{\eps\}\cup a^{313}(a^{314})^*$,
 that is, by automata that have a final state
 in their tails.
\qed   
\end{example}

Motivated by our previous examples,
we finally  
consider the state complexity of the concatenation of~$k$ lan\-guages
recognized by unary automata
that  have
final states in their tails.
While in our previous two theorems,
the length of the resulting cycle 
was equal to the greatest common divisor
of the lengths of cycles in the given automata,
here, similarly to the case of concatenation of two languages (cf.~\cite[Theorems~10,~11]{ps02}),
it may be equal to their least common multiple.
We cannot obtain~a~tight~upper~bound here,
nevertheless, we provide an example that meets our upper bound.
 
\begin{theorem}
\label{thm:unary_final_tails}
    For~$i=1,2,\ldots,k$,
    let~$A_i$ be a 
    unary DFA of size~$(\lambda_i,\mu_i)$.
    For each non-empty set~$I=\{i_1,i_2,\ldots,i_\ell\}\subseteq\{1,2,\ldots,k\}$,
    let~
    \begin{align*}        d_I&=\gcd(\lambda_{i_1},\lambda_{i_2},\ldots,\lambda_{i_\ell}),
    \\
    f(I)&=f(\frac{\lambda_{i_1}}{d_I},\frac{\lambda_{i_2}}{d_I},\ldots,\frac{\lambda_{i_\ell}}{d_I}),
    \end{align*}
    and set~$d_\emp=1$ and~$f(\emp)=0$.
    Then the language~$L(A_1)L(A_2)\cdots L(A_k)$
    is recognized by a DFA of size~$(\lambda,\mu)$
    where
    \begin{align*}         \lambda&=\lcm(\lambda_1,\lambda_2,\ldots,\lambda_k)
    \\
    \mu &= \max\{\mu_1+\mu_2+\cdots+\mu_k - k+1+d_I\cdot f(I)
    \mid I\subseteq\{1,2,\ldots,k\}\}.
    \end{align*}	
\end{theorem}

\begin{proof}
    Let~$L_i=L(A_i)$ and~$L=L(A_1)L(A_2)\cdots L(A_k)$.
    We have~$L_i=X_i \cup a^{\mu_i}Y_i$
    where~$X_i=L(A_i)\cap\{a^x\mid 0\le x <\mu_i\}$
    and~$Y_i=\{a^x \mid a^{\mu_i+x}\in L(A)\}$.
    Then
    \[
       L= \bigcup_{I\subseteq\{1,2,\ldots,k\}}\prod_{j\notin I}X_j
       \prod_{i\in I}a^{\mu_i}Y_i.    
    \]
    
    For each subset~$I$, 
    the language~$\prod_{j\notin I}X_j$
    is a finite language
    recognized by a unary DFA of 
    size~$(1,1+\sum_{j\notin I}(\mu_j-1))$,
    and by Corollary~\ref{cor_unary},
    the language~$\prod_{i\in I}a^{\mu_i}Y_i$
    is recognized by a DFA 
    of size~$(d_I,1+d_I\cdot F(I)+ \sum_{i\in I}(\mu_i-1))$. 
    The concatenation of these two languages
    is recognized by a DFA of 
    size~$(d_I,\mu_1+\mu_2+\cdots+\mu_k-k+1+d_I\cdot f(I))$; cf.~\cite[Theorem~6]{ps02}.
    Then, the union of these concatenations
    is recognized by a DFA
    of size~$(\lambda,\mu)$
    by~\cite[Theorem~4]{ps02}.
    \end{proof}

\begin{example}\rm
    Consider unary DFAs~$A_1,A_2,A_3$
    of sizes~$(12,2)$, $(20,2)$, and~$(30,2)$,
     respectively, 
    with~$F_1=\{0,13\}$, $F_2=\{0,21\}$,
    and~$F_3=\{0,31\}$.  
    
    We have~$\lcm(12,20,30)=60$,
    $4\cdot f(3,5)=6\cdot f(2,5)=10\cdot f(2,3)=60$,
    and $2\cdot f(6,10,15)=2\cdot 2\cdot f(3,5,15)=2\cdot2\cdot 5\cdot f(3,1,3)=
     2\cdot2\cdot 5\cdot3\cdot  f(1,1,1)= 
   2\cdot2\cdot 5\cdot3\cdot 2=120$.
     
     The size of the minimal automaton recognizing the language~$L(A_1)L(A_2)L(A_3)$
     is~$(60,124)$ where~$124=2+2+2-3+1+\max\{60,120\}$.
\qed
\end{example}  

The above example shows that our upper bound
given by Theorem~\ref{thm:unary_final_tails}
is met by unary automata of sizes~$(12,2),(20,2),(30,2)$.
The tightness of this upper bound
in a general case remains open. 
  
\section{Conclusions}
\label{a:conclusions}

We examined in detail the state complexity of the multiple concatenation of~$k$ languages.
First, we described witness DFAs~$A_1,A_2,\ldots,A_k~$ over the~$(k+1)$-letter 
alphabet~$\{b,a_1,a_2,\ldots,a_k\}$,
in which each~$a_i$ performs the circular shift in~$A_i$ and the identity in the other automata,
while~$b$ performs a contraction.
Using symbols~$a_1, a_2,\ldots,a_k$,
we proved the reachability of all valid states  in the subset automaton for the concatenation
by carefully setting the~$i$th component without changing the already set~$(i+1)$th component. 
The transitions on~$b$ guaranteed the co-reachability of all singleton sets
in the NFA for concatenation, and therefore we obtained the proof of distinguishability 
of all states in the corresponding subset automaton for free.
However, to get co-reachability of singletons,
our witness automata were required to have at least three states.
Nevertheless, we described witness automata over a~$(k+1)$-letter
alphabet also in the case where some of them  have only two states.

Then we provided special
binary witnesses for the concatenation of two languages.
Using our results concerning witnesses over a~$(k+1)$-letter alphabet,
as well as the results for the special binary automata,
we described witnesses for the concatenation of~$k$ languages
over a~$k$-letter alphabet. This solves an open problem stated in \cite{clp18}.
For~$k=3$, we proved that the ternary alphabet is optimal in the sense
that the upper bound for the concatenation of three languages
cannot be met by any binary languages.
This provides a partial answer to the second open problem from  \cite{clp18}.
 
We also considered multiple concatenation on binary and ternary languages,
and obtained lower bounds~$n_1-1+(1/2^{2k-2}) 2^{n_2+n_3+\cdots+n_k}$
and~$(1/2^{2k-2})n_1 2^{n_2+n_3+\cdots+n_k}$, respectively.
This shows that the state complexity of multiple concatenation
remains exponential in~$n_2,n_3,\ldots,n_k$ in the binary case,
and that a trivial upper bound can be met,
up to a multiplicative constant depending on~$k$,
by ternary languages.

Finally, we investigated multiple concatenation on unary 
languages.
We obtained a tight upper bound for cyclic languages,
and we showed that for~$k\ge3$, 
it is much smaller than a trivial upper bound~$n_1 n_2 \cdots n_k$,
which is met by cyclic unary languages if~$k=2$ and~$\gcd(n_1,n_2)=1$ \cite[Theorem~5.4]{yzs94}.
We also provided a tight upper bound
for languages recognized by automata
that do not have final states in their tails.

Some problems remain open. First,  our~$k$-letter witnesses  require~$n_i\ge3$ for~$i=2,3,\ldots,k-1$,
while the~$(2k-1)$-letter witnesses in \cite[Theorem~5]{gy09} work  with~$n_i\ge2$.
Is~it~possible to define~$k$-letter  witnesses also in such a case? We can   do this using~$k+1$ letters,
or with~$k$ letters
if \emph{all} automata
have two states.
 
We proved the optimality of a ternary alphabet 
for the concatenation of three languages.
However, we cannot see any generalization of the proof. 
Is a~$k$-letter alphabet
for describing witnesses for the concatenation of~$k$ languages optimal?

Next, we provided upper bounds in the case where exactly one
automaton has one state, and using a binary alphabet 
we proved that they are tight if~$k=2$.
What is the state complexity of multiple concatenation
if some languages may be equal to~$\Sigma^*$?

Finally, in the unary case,
we obtained an upper bound for multiple concatenation
of languages recognized by unary automata 
that may have final states in their tails.
The tightness of this upper bound remains open. 
 
\bibliographystyle{splncs04}
\bibliography{catenation}

\end{document}